\documentclass[a4paper,11pt]{article} 
\usepackage[english]{babel}
\usepackage[a4paper]{geometry}
\usepackage{amsmath}
\usepackage{amsthm}
\usepackage{amssymb}
\usepackage{bbm}
\usepackage{xcolor}
\usepackage{enumerate}
\usepackage{dsfont}
\usepackage{cancel}
\usepackage{ulem} %\sout{} to cancel text

\usepackage{hyperref}         %%%% click to contents
\usepackage{tikz}
\usetikzlibrary{fit}
\usetikzlibrary{shapes.geometric}
\usetikzlibrary{decorations.pathmorphing}

\tikzset{
point/.style={circle,fill=black,inner sep=1pt},
vertex/.style={circle,fill=black,inner sep=1.5pt},   % VERTICI
bvertex/.style={circle,fill=black,inner sep=2.8pt},
Bvertex/.style={circle,fill=black,inner sep=4pt}, % VERTICE GRANDE
specialEP/.style={rectangle,fill=white,draw,inner sep=3pt},  % RETTANGOLO
whitevex/.style={circle,fill=white,draw, inner sep=2pt},
linelabel/.style={sloped,above,very near start, inner sep=1pt,execute at begin node=$\scriptstyle,execute at end node=$},
baseline=(current  bounding  box.center),doubled/.style={double distance= 1pt,line width=1.5pt},
th/.style={line width=0.5 pt, gray},  %linea thin GRIGIA
med/.style={line width=1 pt}  %linea con medio spessore
}

\xdefinecolor{myred}{rgb}{0.7,0,0.2} %per definire colori personalizzati
\xdefinecolor{mypurple}{rgb}{0.6,0.2,0.4} %{0.42,0.118,0.588} 
\xdefinecolor{myblue}{rgb}{0.259,0.2,0.6}   
\xdefinecolor{myorange}{rgb}{1,0.6,0.2}   
\definecolor{orange}{rgb}{1,0.5,0}
\xdefinecolor{purpleish}{cmyk}{0.75,0.75,0,0}

\def\bR{\mathbb{R}}

\def\bN{\mathbb{N}}

\def\bZ{\mathbb{Z}}

\def\cF{\mathcal{F}}

\def\cL{\mathcal{L}}

\def\cN{\mathcal{N}}

\def\cK{\mathcal{K}}
\def\cH{\mathcal{H}}

\def\eps{\varepsilon}
\def\ph{\varphi}

\def\wt{\widetilde}

\def\indic{\hbox{\raise-2pt \hbox{\indbf 1}}}

\let\==\equiv

\let\io=\infty
\let\0=\noindent

\def\*{{\hfill\break\null\hfill\break}}

\def\tende#1{\,\vtop{\ialign{##\crcr\rightarrowfill\crcr
             \noalign{\kern-1pt\nointerlineskip}
             \hskip3.pt${\scriptstyle #1}$\hskip3.pt\crcr}}\,}
\def\otto{\,{\kern-1.truept\leftarrow\kern-5.truept\to\kern-1.truept}\,}

\def \dx{{d}x}

\newtheorem{theorem}{Theorem}[section]  % use thm for %Theorems to keep numbering consistent

\newtheorem{proposition}[theorem]{Proposition}
\newtheorem{lemma}[theorem]{Lemma}

\numberwithin{equation}{section}

\def\tl#1{{\tilde{#1}}}
\def\be{\begin{equation}}
\def\ee{\end{equation}}

     \let\g=\gamma     \let\d=\delta     
       \let\th=\vartheta      
                          
\let\s=\sigma          \let\ph=\varphi   
   \let\o=\omega     
        \let\L=\Lambda

\definecolor{lightblue}{rgb}{0, 0.33, 0.71}

\def\aa{\mathfrak{a}}

\def \blue#1 {\textcolor{blue}{#1}}
\def \red#1 {\textcolor{red}{#1}}
\def \blue#1{\textcolor{blue}{#1}}
\DeclareFontFamily{U}{mathx}{\hyphenchar\font45}
\DeclareFontShape{U}{mathx}{m}{n}{
      <5> <6> <7> <8> <9> <10>
      <10.95> <12> <14.4> <17.28> <20.74> <24.88>
      mathx10
      }{}
\DeclareSymbolFont{mathx}{U}{mathx}{m}{n}
\DeclareFontSubstitution{U}{mathx}{m}{n}
\DeclareMathAccent{\widecheck}{0}{mathx}{"71}

\def\bskip{\\[-0.6cm]}

\title{The Lee-Huang-Yang energy for a dilute gas \\ of hard spheres: an upper bound} 

\date{\today}

\author{Giulia Basti\footnote{Department of Mathematics ``Guido Castelnuovo'', La Sapienza, Piazzale Aldo Moro, 5, 00185 Roma}, \, Morris Brooks\footnote{Institute of Mathematics, University of Zurich, Winterthurerstrasse 190, 8057 Zurich}, \, Serena Cenatiempo\footnote{Gran Sasso Science Institute, Viale Francesco Crispi 7, 67100 L'Aquila}, \, Alessandro Olgiati\footnote{Dipartimento di Matematica, Politecnico di Milano, Piazza Leonardo da Vinci 32, 20133 Milano}, \\ Benjamin Schlein$^\dagger$}

\begin{document}

\maketitle

\begin{abstract} 
We consider a quantum gas consisting of $N$ hard spheres with radius $\frak{a} > 0$, obeying bosonic statistics and moving in the box $\Lambda = [0;L]^3$ with periodic boundary conditions. We are interested in the ground state energy per unit volume in the thermodynamic limit, with $N, L \to \infty$ at fixed density $\rho = N / L^3$. We derive an upper bound for the ground state energy density, matching the famous Lee-Huang-Yang formula, up to lower order terms, in the dilute limit $\rho \frak{a}^3 \ll 1$. 
\end{abstract}

\section{Introduction} 

Since the early days of quantum mechanics, the study of dilute Bose gases has played a central role in condensed matter physics. Based on the work of Bogoliubov \cite{Bog}, Lee-Huang-Yang \cite{LHY} predicted in 1957 that the ground state energy per unit volume of a Bose gas at density $\rho > 0$ is given by 
\begin{equation} \label{eq:LHY0} e(\rho) = 4 \pi \frak{a} \rho^2  \Big[ 1 + \frac{128}{15\sqrt{\pi}} (\rho \frak{a}^3)^{1/2} + \dots \Big] \end{equation} 
up to smaller order corrections, in the dilute limit $\rho \frak{a}^3 \to 0$. Here $\frak{a} > 0$ is the scattering length of the interaction potential.  In particular, \eqref{eq:LHY0} implies that, at this level of precision, the ground state energy density is universal, ie. it only depends on the scattering length $\frak{a} > 0$ but it is independent of further details of the interaction. Physically, this means that the energy is determined, to leading and next-to-leading order, only by two-body scattering events, which are accurately described by $\frak{a}$ (turning this intuition into a mathematical statement requires however to identify cancellations between several a priori non-universal contributions to the energy). 

Heuristically, \eqref{eq:LHY0} can be understood by taking into account the fact that, at low energies and for low densities, most particles in the system occupy the same zero momentum state, a phenomenon known as Bose-Einstein condensation. The leading contribution on the r.h.s. of (\ref{eq:LHY0}) results from (appropriately renormalized) interactions among particles in the condensate. The second order term, on the other hand, is associated with scattering events between particles in the condensate and its orthogonal excitations. Putting this intuition on rigorous grounds, however, requires control of correlations over increasingly large length scales: up to scales larger than the scattering length $\frak{a}$, for the leading-order term, and up to scales larger than the healing length $(\rho \frak{a})^{-1/2} \gg \frak{a}$, for the second-order term. 

The first rigorous derivation of the leading order term in the expansion (\ref{eq:LHY0}) has been obtained by Dyson \cite{Dy}, as an upper bound, and by Lieb-Yngvason \cite{LY}, as a lower bound, for general repulsive interactions. At positive temperature, leading order estimates for the free energy of a dilute gas have been derived in \cite{Sei} (lower bound) and \cite{Y,BBCD} (upper bounds). 

In the last years, the focus has shifted towards the rigorous justification of the second order contribution to the ground state energy density, known as the Lee-Huang-Yang term. For particles interacting through weak, mean-field type, potentials, second order expansions for the ground state energy have been established in \cite{GS,Sei,LNSS,DN,P}. In the Gross-Pitaevskii regime, where $N$ particles move in a volume of order one and interact through a potential with scattering length of the order $1/N$, the ground state energy has been determined to the Lee-Huang-Yang order in \cite{BBCS1,BBCS2,BSS1,NT, BSS2,HST,B}. In this case, even the third order correction to the ground state energy, predicted in 1959 by Wu \cite{Wu}, Hugenholtz-Pines \cite{HP} and Sawada \cite{Sa}, has been recently established; see \cite{COSS}. 

In the thermodynamic limit, at low but fixed density $\rho > 0$, a lower bound matching (\ref{eq:LHY0}) up to the Lee-Huang-Yang order has been first shown in \cite{FS1}, for particles interacting through a repulsive, integrable potential, and then in \cite{FS2}, for more general interactions, including the hard-sphere potential. Very recently, these results have been also extended to low but positive temperatures; see \cite{HHNST,FJGMOT}. 

In the same thermodynamic setting, an upper bound capturing the Lee-Huang-Yang term has been first proven in \cite{YY}, assuming regular potential. A simpler second order upper bound for the ground state energy density has been later derived in \cite{BCS} and extended to positive temperatures in \cite{HHST}. Very recently, a third order upper bound, capturing the Lee-Huang-Yang term and also Wu's contribution, has been obtained in \cite{BOSS}. 

Compared with \cite{YY}, the upper bounds in \cite{BCS,HHST,BOSS} cover a larger class of interactions, but they still require integrability and thus they exclude the physically interesting case of a hard-sphere potential. For hard spheres, the best available upper bound resolves the ground state energy per unit volume up to the Lee-Huang-Yang order, but it does not achieve the correct constant; see \cite{BCGOPS}. 

%More precise lower bounds, capturing also %the second order correction in %(\ref{eq:LHY0}), have been recently shown %in \cite{FS1}, for particles interacting %through an integrable potential, and in 
%\cite{FS2}, for general repulsive %interactions (including hard spheres). An %upper bound matching  (\ref{eq:LHY0}) has %been first proven in \cite{YY}, assuming %regular potential. More recently, a %simpler second order upper bound for the %ground state energy per unit volume was %derived in \cite{BCS}. Compared with 
%\cite{YY}, this last result covers a %larger class of interactions, but it still %requires integrability and thus it %excludes the physically interesting case %of a hard-sphere potential. For hard %spheres, the best available upper bound %resolves the ground state energy up to the %Lee-Huang-Yang order, but it does not %resolve the correct constant; see 
%\cite{BCGOPS}. 

In this paper, we fill the gap in the literature, proving an upper bound for the ground state energy density of hard spheres, matching (\ref{eq:LHY0}) up to lower order corrections. Together with the lower bound from \cite{FS2}, this extends the universality of the Lee-Huang-Yang formula (\ref{eq:LHY0}), proving its validity also for the hard-sphere gas in the dilute limit.

We consider a system of $N$ hard spheres, moving in the torus $\Lambda \simeq [-L/2;L/2]^3$. The ground state energy of the gas is given by 
\[ E_N^\text{hs} = \inf \frac{1}{\| \psi_N \|^2} \big\langle \psi_N, \sum_{j=1}^N -\Delta_{x_j} \Psi_N \big\rangle \]
where the infimum is taken over all $\psi_N \in L^2_s (\Lambda^N)$, the subspace of $L^2 (\Lambda^N)$ consisting of functions that are symmetric with respect to permutations, satisfying the hard-core condition 
\begin{equation}\label{eq:hc} \psi_N (x_1, \dots , x_N) = 0 , \quad \text{if there are $i,j \in \{ 1,\dots , N \}$ with $|x_i - x_j| \leq \frak{a}$\,.} \end{equation} 
We are interested in the ground state per unit volume at density $\rho > 0$, defined by the limit 
\be \label{def:e-rho}
e(\rho) = \lim_{\substack{N,L \to \infty \\ N/L^3 = \rho}} \frac{E_N^\text{hs}}{L^3} \, , \ee
which is well-known to exist, see  \cite[Thm 3.5.11]{R}.%\cite{R}. 

\begin{theorem} \label{thm:main} 
For $\rho > 0$, let $e(\rho)$ be defined as in (\ref{def:e-rho}). Then there exists $C > 0$ and a sufficiently small $\delta > 0$ such that  
\begin{equation}\label{eq:main} e(\rho) \leq 4 \pi \frak{a} \rho^2 \Big[ 1 + \frac{128}{15\sqrt{\pi}} (\rho \frak{a}^3)^{1/2} + C  (\rho\frak{a}^3)^{1/2+\delta} \Big] \end{equation} 
for all $\rho > 0$ small enough. 
\end{theorem} 

\textbf{Remark.} In the recent paper \cite{BBCOS2}, we generalize the bound (\ref{eq:main}) to particles interacting through an arbitrary repulsive, even and compactly supported potential. The extension is based on the approach developed in the present work (but the analysis in \cite{BBCOS2} is more complicated, because we also need to compute the potential energy).

\textbf{Remark.} In contrast to \cite{YY,BCS,HHST,BOSS}, where the trial state on $\Lambda$ was constructed by pasting together a large number of identical states defined on smaller boxes, with size independent of $L$, here we work directly in the thermodynamic box. 

%\begin{remark} 
%In 1959, the prediction (\ref{eq:LHY0}) was extended to the next %order by Wu \cite{Wu} and Pines-Hugenholtz \cite{PH}. Recently, a %rigorous upper bound matching Wu's prediction has been shown in 
%\cite{BOSS}, for particles interacting through a square integrable %potential. For hard spheres, the derivation of an upper bound %resolving Wu's energy remains open. 
%\end{remark}
%\begin{remark}
%In the Gross-Pitaevskii regime, where $N$ particles are trapped in %the unit torus and the scattering length of the potential scales %as $1/N$, the analogous of (\ref{eq:LHY0}) has been known since 
%\cite{BBCS1,BBCS2}, assuming particles to interact through a non-%negative, spherical symmetric  and compactly supported $V \in L^3 %(\bR^3)$. In this setting, also Wu's corrections have been %recently resolved in \cite{COSS}. For hard spheres in the Gross-%Pitaevskii limit, an upper bound matching (\ref{eq:LHY0}) has been %shown in \cite{BCOPS}. 
%\end{remark} 

\medskip

The proof of Theorem \ref{thm:main} is based on the construction of an appropriate trial state and on the computation of its energy. As first proposed in \cite{Bijl,Dingle,Jastrow}, to approximate the ground state of a dilute gas it is natural to consider so-called Jastrow factors having the form 
\begin{equation}\label{eq:jastrow} \Psi_\text{J} (x_1, \dots , x_N) = \prod_{i<j}^N f_\ell (x_i -x_j) \, \end{equation}
 with $f_\ell$ describing two-particle correlations produced by the interaction, up to a length scale $\ell \gg \frak{a}$. For hard spheres, an obvious choice is given by 
 \begin{equation}\label{eq:fell0} f_\ell (x) = \left\{ 
 \begin{array}{ll}  
 1 - \frac{\frak{a}}{|x|} \chi_\ell (x) \quad &\text{if $|x| \geq \frak{a}$} , \\ 0 , \quad &\text{if $|x| \leq \frak{a}$} \end{array} \right.  \end{equation} 
 with a cutoff $\chi_\ell (x) = \chi (x/\ell)$, where $\chi \in C^\infty_0 (\bR^3)$, with $\chi (x) = 1$, if $|x| \leq 2$, and $\chi (x) = 0$, for $|x| \geq 4$. 
% $\chi (x) = 1$, if $|x| \leq 1$, and $\chi (x) = 0$, for $|x| \geq 2$. 
For (\ref{eq:jastrow}), we can compute 
\[ \nabla_{x_j} \Psi_\text{J} (x_1, \dots , x_N) = \sum_{i \not = j} \frac{\nabla f_\ell (x_j - x_i)}{f_\ell (x_j - x_i)} \Psi_\text{J} (x_1, \dots , x_N) \]
which leads us to the kinetic energy 
\begin{equation}\label{eq:kin0} \begin{split} \langle \Psi_\text{J} , &\sum_{j=1}^N -\Delta_{x_j} \Psi_\text{J} \rangle \\ = \; &2 \sum_{i<j}^N \int dx_1 \dots dx_N \,  \frac{|\nabla f_\ell (x_j - x_i)|^2}{f_\ell^2 (x_j - x_i)} |\Psi_\text{J} (x_1, \dots , x_N)|^2 \\ &+ \sum_{i,j,m}^N \int dx_1 \dots dx_N \, \frac{\nabla f_\ell (x_j - x_i)}{f_\ell (x_j-x_i)}  \cdot \frac{\nabla f_\ell (x_j -  x_m)}{f_\ell (x_j - x_m)} | \Psi_\text{J} (x_1 , \dots. ,x_N)|^2 \end{split} \end{equation} 
where the sum on the last line runs over $1\leq i,j, m \leq N$ all different.  With the choice (\ref{eq:fell0}), we easily find (see Lemma \ref{lm:fell} below) \begin{equation}\label{eq:nablaf2} \| \nabla f_\ell \|^2 \leq 4 \pi \frak{a} \Big(1 + C \frac{\frak{a}}{\ell} \Big)\,.\end{equation}  
Recalling that $0 \leq f_\ell (x) \leq 1$, we obtain  
\[ \begin{split} 
2 \sum_{i<j} \int dx_1 \dots dx_N \,  &\frac{|\nabla f_\ell (x_j - x_i)|^2}{f_\ell^2 (x_j - x_i)} |\Psi_\text{J} (x_1, \dots , x_N)|^2 
%\\ = \; &N (N-1) \int dx_1 \dots dx_N \, \frac{|\nabla f_\ell (x_1 - x_2)|^2}{f^2_\ell (x_1 - x_2)} \prod_{i<j}^N f_\ell^2 (x_i - x_j)  
\\ \leq \; &N^2 \int dx_1 \dots dx_N \,  |\nabla f_\ell (x_1 -x_2)|^2 \prod_{2 \leq i < j \leq N} f_\ell^2 (x_i -x_j)  \\ \leq \; &4\pi \frak{a} N^2 \Big( 1 + C \frac{\frak{a}}{\ell} \Big) \int dx_2 \dots dx_N \, \prod_{2 \leq i < j \leq N} f_\ell^2 (x_i -x_j)\,. \end{split} \]
Moreover, since $0 \leq 1 - f_\ell^2 (x) \leq C \frak{a} \chi_\ell (x) / |x|$ and  
\[ \prod_{m=2}^N f_\ell^2 (x_1 - x_m) = \prod_{m=2}^N (1- (1-f_\ell^2 (x_1 - x_m))) \geq 1 - \sum_{m=2}^N (1- f_\ell^2 (x_1 - x_m)) \]
we have   
\begin{equation}\label{eq:norm0} \begin{split} \| \Psi_\text{J} \|^2 &=  \int dx_1 \dots dx_N \, \prod_{i<j}^N f_\ell^2 (x_i - x_j) \\ &\geq \int dx_1 \dots dx_N \, \Big( 1 - \sum_{m=2}^N (1-f_\ell^2 (x_1 - x_m)) \Big) \prod_{2 \leq i < j \leq N} f_\ell^2 (x_i -x_j) \\ &\geq \Big( L^3 - C N \frak{a} \ell^2 \Big) \int dx_2 \dots dx_N \, \prod_{2 \leq i < j \leq N} f_\ell^2 (x_i - x_j) \,.
%\\ &= L^3 \big( 1 - C \rho 
%\frak{a} \ell^2 \big) \int dx_2 
%\dots dx_N \, \prod_{2 \leq i < j %\leq N} f_\ell^2 %(x_i - x_j)  
\end{split} \end{equation} 
Proceeding similarly to bound the second term on the r.h.s. of (\ref{eq:kin0}) and noticing that $L^3 - C N \frak{a} \ell^2 = L^3 (1 - C \rho \frak{a} \ell^2)$, we arrive at
%Hence, we find 
%\[ \frac{2}{\| \Psi_\text{J} \|^2} \sum_{i<j} \int dx_1 \dots dx_N \,  \frac{|\nabla f_\ell (x_j - x_i)|^2}
%{f_\ell^2 (x_j - x_i)} |\Psi_\text{J} (x_1, \dots , x_N)|^2 \leq 4\pi \frak{a} \rho N \Big( 1+ C \frac{\frak{a}}%{\ell} + C \rho \frak{a} \ell^2 \Big) \]
% To estimate the contribution of the second term on the r.h.s. of (\ref{eq:kin0}), we proceed similarly, %noticing that 
%\[  \begin{split} 
%\Big| \sum_{i,j,m}^N \int dx_1 \dots dx_N \, \frac{\nabla f_\ell (x_j - x_i)}{f_\ell (x_j-x_i)}  \cdot %&\frac{\nabla f_\ell (x_j -  x_m)}{f_\ell (x_j - x_m)} | \Psi_\text{J} (x_1 , \dots. ,x_N)|^2 \Big|  \\ & \leq C %N^3 \int dx_3 \dots dx_N \prod_{3 \leq i< j \leq N} f_\ell^2 (x_i -x_j) \end{split} \]
%and that, iterating (\ref{eq:norm0}), 
%\[ \| \Psi_\text{J} \|^2 \geq L^6 \Big( 1 - C \rho \frak{a} \ell^2 \Big) \int dx_3 \dots dx_N \, \prod_{3\leq i %< j \leq N} f_\ell^2 (x_i - x_j)  \]
%This implies that 
\begin{equation}\label{eq:dyson0}   \liminf_{\substack{N,L \to \infty : \\ N/L^3 = \rho}} \; \frac{1}{L^3}  \frac{\langle \Psi_\text{J}, \sum_{j=1}^N -\Delta_{x_j} \Psi_\text{J} \rangle}{\| \Psi_\text{J} \|^2} \leq  4 \pi \frak{a} \rho^2 \Big( 1 + C \frac{\frak{a}}{\ell} + C \rho \frak{a} \ell^2 \Big)\,.  \end{equation} 
With the optimal choice $\ell = \rho^{-1/3}$, we  derive the bound \begin{equation}\label{eq:dyson1} e (\rho) \leq 4 \pi \frak{a} \rho^2 (1 + C (\rho \frak{a}^3)^{1/3})  \end{equation} 
first proven by Dyson in \cite{Dy} (in his work, Dyson used a non-symmetric version of (\ref{eq:jastrow})). Notice that (\ref{eq:dyson0}) crucially relies on cancellations between identical factors appearing both in the expectation of the energy (in the numerator) and in the norm of $\Psi_\text{J}$ (in the denominator) on the l.h.s. of the inequality.

This simple computation captures the correct leading order contribution to the ground state energy per unit volume, but it misses the Lee-Huang-Yang corrections. Through a systematic cluster-type expansion of  (\ref{eq:jastrow}) to sufficiently high order, the cutoff in (\ref{eq:fell0}) can be increased  to $\ell = c (\rho \frak{a})^{-1/2}$, with $c > 0$ small enough, leading to an upper bound resolving the ground state energy per unit volume up to corrections of the Lee-Huang-Yang order; see  \cite{BCGOPS}. Since the convergence of the expansion used in \cite{BCGOPS} requires $\rho \frak{a} \ell^2 \ll 1$, it seems difficult to include the correct correlation structure in the trial state (\ref{eq:jastrow}), up to a length scale $\ell$, large in comparison with the healing length $(\rho \frak{a})^{-1/2}$, as needed to reach Lee-Huang-Yang precision. 

For ``soft'' interaction potentials, a different approach to derive upper bounds on the ground state energy has been proposed in \cite{GA} and, more recently, in \cite{ESY}. It is based on the use of Bogoliubov transformations to model correlations, up to and beyond the scale $(\rho \frak{a})^{-1/2}$. While quasi-free states generated by Bogoliubov maps do not reach the desired precision \cite{NRS}, they can resolve the Lee-Huang-Yang term after appropriate modifications by means of transformations describing scattering processes involving a particle in the Bose-Einstein condensate and three of its orthogonal excitations; see 
\cite{YY,BCS} (in the language of second quantization, Bogoliubov transformations, given by the exponential of quadratic expressions in creation and annihilation operators, have to be combined with additional maps, given by the exponential of cubic terms).

Unfortunately, it is problematic to extend this approach to the hard-sphere potential, because it is difficult to enforce the condition (\ref{eq:hc}) on a Bogoliubov state. Still, taking inspiration from the soft-potential case, it makes sense to modify (\ref{eq:jastrow}), keeping the Jastrow factor on short length scales $\frak{a} \ll \ell \ll (\rho \frak{a})^{-1/2}$ (making sure, in particular, that the hard-core condition (\ref{eq:hc}) is satisfied), but multiplying it with a Bogoliubov state, describing correlations up to (and just beyond) the healing length $(\rho \frak{a})^{-1/2}$. This approach works in the Gross-Pitaevskii scaling \cite{BCOPS}, where the effective density vanishes as $N \to \infty$.

%In the thermodynamic limit, on the other hand, it is not clear whether it can %work, because the presence of the Bogoliubov state makes it difficult to %identify  cancellations between numerator and denominator that played an %important role in (\ref{eq:dyson0}). 

To combine Jastrow factor and Bogoliubov transformation in the thermodynamic limit that we are considering in this paper, it is more convenient to work in the grand canonical ensemble, on the bosonic Fock space $\cF = \bigoplus_{n \geq 0} L^2 (\Lambda)^{\otimes_s n}$. In this setting, we implement the Jastrow factor (\ref{eq:jastrow}) through the operator $J : \cF \to \cF$, whose action on an arbitary $\Phi = \{ \Phi^{(n)} \}_{n \geq 0} \in \cF$ is defined by 
\begin{equation}\label{eq:jastrowF} (J \Phi)^{(n)}  (x_1,\dots, x_n) = \prod_{i<j}^n f_\ell (x_i - x_j) \Phi^{(n)} (x_1, \dots , x_n) \end{equation} 
for all $n \in \bN$ (with the convention that $(J \Phi)^{(n)} = 0$, for $n=0,1$). 

To compute the energy of the Jastrow factor generated by (\ref{eq:jastrowF}), it is important to determine the action of annihilation operators on $J$. For any $x \in \Lambda$, we consider the  operator-valued distribution  $a_x$, defined by \[ (a_x \Phi)^{(n)} (x_1, \dots ,x_n) = \sqrt{n+1} \,\Phi^{(n+1)} (x,x_1, \dots , x_n).\] We find 
\begin{equation} \label{eq:axJ-0} a_x J = J(x) J a_x \end{equation}
where $J(x) : \cF \to \cF$ is given by  
\[ (J(x) \Phi)^{(n)} (x_1, \dots , x_n) = \prod_{j=1}^n f_\ell (x - x_j) \Phi^{(n)} (x_1, \dots , x_n) \]
for all $x \in \Lambda$ (the identity (\ref{eq:axJ-0}) will be shown in Section \ref{sec:trial}, where we will also introduce the formalism of second quantization, with precise definitions of annihilation and creation operators). 

From the equivalence of ensembles in the thermodynamic limit, to prove the upper bound (\ref{eq:main}) for the ground state energy per unit volume, it is enough to construct a normalized trial state $\Psi \in \cF$, with expected number of particles $\langle \Psi, \cN \Psi \rangle \geq \rho L^3$ and with kinetic energy per unit volume matching (\ref{eq:main}). As a first (too) simple attempt to reach this goal, we define   
\begin{equation} \label{eq:psi-simple} \widetilde\Psi = \frac{1}{Z} J W(\rho_0) \Omega \end{equation} 
where $Z = \| J W (\rho_0) \Omega \|$ is a normalization constant, $\Omega = \{ 1, 0, 0, \dots \}$ is the vacuum vector in $\cF$, the unitary Weyl operator   
\[ W (\rho_0) = e^{\sqrt{\rho_0} \int dx (a_x^* - a_x)} \]
generates the coherent state $W (\rho_0) \Omega$, modeling a Bose-Einstein condensate with, in average, $N_0 = \rho_0 L^3$ particles, all at zero momentum, and $J$ (defined as in (\ref{eq:jastrowF})) creates a Jastrow factor, describing correlations (on top of the completely uncorrelated coherent state). Recalling that Weyl operators act on annihilation operators as shifts, ie. 
\[ W^* (\rho_0) a_x W (\rho_0) = a_x + \sqrt{\rho_0} ,  \]
we easily find 
\begin{equation} \label{eq:a-simple} a_x \widetilde{\Psi} = \frac{1}{Z} a_x J W (\rho_0) \Omega= \frac{1}{Z} J(x) J a_x W (\rho_0) \Omega = \sqrt{\rho_0} J(x) \widetilde{\Psi}\,. \end{equation} 
This allows us to compute the expectation of the kinetic energy $d\Gamma (-\Delta)$ in the state (\ref{eq:psi-simple}). Writing 
\[ d\Gamma (-\Delta) = \int dx \, \nabla_x a_x^* \nabla_x a_x \]
in terms of creation and annihilation operators, we find 
\begin{equation*} \begin{split} \langle \widetilde{\Psi}, & \, d\Gamma (-\Delta) \widetilde{\Psi} \rangle \\ = \; &\int dx \, \langle \nabla_x a_x \widetilde{\Psi} , \nabla_x a_x \widetilde{\Psi} \rangle = \rho_0 \int dx \, \langle \nabla_x J(x) \widetilde{\Psi}, \nabla_x J(x) \widetilde{\Psi} \rangle \\ = \; & \rho_0 \sum_{n \geq 1} \sum_{j=1}^n \int dx \int  dx_1 \dots dx_n \, \frac{|\nabla f_\ell (x - x_j)|^2}{f_\ell^2 (x-x_j)} \prod_{m=1}^n f_\ell^2 (x-x_m) |\widetilde{\Psi}^{(n)} (x_1, \dots ,x_n)|^2 \\ & + \rho_0 \sum_{n\geq 1} \sum_{i<j}^n \int dx \int dx_1\dots dx_n \,
  \frac{\nabla f_\ell (x - x_j)}{f_\ell (x-x_j)}  \cdot \frac{\nabla f_\ell (x - x_i)}{f_\ell (x-x_i)} \prod_{m=1}^n f_\ell^2 (x-x_m) \\ &\hspace{9cm} \times  | \widetilde{\Psi}^{(n)} (x_1, \dots ,x_n)|^2  \\ =: \, &\text{I} + \text{II}\,. \end{split}  \end{equation*}
Since $0 \leq f_\ell (x) \leq 1$ and recalling (\ref{eq:nablaf2}), the first term on the r.h.s. is bounded by 
\[ \begin{split}
% \rho_0 \sum_{n \geq 1} \sum_{j=1}^n &\int dx \int \frac{|\nabla f_\ell (x - x_j)|^2}{f_\ell^2 (x-x_j)} %\prod_{m=1}^n f_\ell^2 (x-x_m) |\widetilde{\Psi}^{(n)} (x_1, \dots ,x_n)|^2 \\ &\hspace{2cm}
\text{I} \leq 4\pi \frak{a} \rho_0 \Big( 1 + C \frac{\frak{a}}{\ell} \Big) \sum_n n \| \widetilde{\Psi}^{(n)} \|^2 \leq 4 \pi \frak{a} \rho_0 N_0 \Big( 1+ C \frac{\frak{a}}{\ell} \Big)\,. \end{split} \]
As for the second term, we apply (\ref{eq:a-simple}) and observe that $J(x), J(y) \leq 1$ (and that $a_x J(y) = f_\ell (x-y) J(y) a_x$, as shown below in (\ref{eq:pull})) to conclude that 
\[ \begin{split} 
\text{II} \leq \; &C \frac{\rho_0}{\ell}  \sum_{n\geq 1} \sum_{i<j}^n \int_{|x_i - x_j| \leq C\ell} dx_1 \dots dx_n \, |\widetilde{\Psi}^{(n)} (x_1, \dots , x_n)|^2 \\ = \; &C \frac{\rho_0}{\ell}  \int_{|x-y| \leq C \ell}  dx dy \, \langle \widetilde{\Psi}, a_x^* a_y^* a_y a_x \widetilde{\Psi} \rangle \leq C \rho_0^2 \ell^2 N_0\,. \end{split} \]
It follows that 
\[ \frac{1}{L^3} \langle \widetilde{\Psi}, d\Gamma (-\Delta) \widetilde{\Psi} \rangle \leq 4 \pi \frak{a} \rho_0^2 \Big( 1 + C \frac{\frak{a}}{\ell} +C \rho_0 \frak{a} \ell^2 \Big) \,. \]
Choosing $\rho_0 = \rho$ (in fact, we need to choose $\rho_0 = \rho (1 +C \rho)$ a bit larger than $\rho$, to make sure that $\langle \widetilde{\Psi}, \cN \widetilde{\Psi} \rangle = \rho L^3$; the difference in the energy is subleading, though) and $\ell = \rho^{-1/3}$, we obtain again Dyson's bound (\ref{eq:dyson1}). 

This computation is similar to the one leading to (\ref{eq:dyson0}) for the energy of the $N$-particle Jastrow factor (\ref{eq:jastrow}). It has, however, one crucial difference: it does not require cancellations of  normalization factors between numerator and denominator. Thanks to this observation, we will be able to modify the simple trial state (\ref{eq:psi-simple}), without worrying about messing up the normalization. Instead of (\ref{eq:psi-simple}), we will consider the trial state $\Psi = Z^{-1} J W (\rho_0) T \Omega$, inserting the Bogoliubov transformation $T$, generating correlations at large scales, up to (and just beyond) the healing length $(\rho \frak{a})^{-1/2}$. Since the action of $T$ on creation and annihilation operators is explicit, we will again be able to derive an identity, more complicated than (\ref{eq:a-simple}), but still simple enough, for the action of $a_x$ on $\Psi$. We will use this identity to estimate the number of particles and the kinetic energy of $\Psi$; this will lead us to (\ref{eq:main}). 

The paper is organized as follows. In Section \ref{sec:trial}, we introduce the formalism of second quantization, we define precisely our trial state $\Psi$ and we state bounds on the expectation of number of particles and kinetic energy operators in $\Psi$. In Section \ref{sec:proof}, we show that these bounds imply our main result, Theorem \ref{thm:main}. In Section \ref{sec:apri}, we establish precise estimates on the distribution of the local number of particles and the local number of excitations in the state $\Psi$. These estimates immediately imply the desired bound on the global number of particles. Moreover, they give us the main tools to control error terms arising in the computation of the energy of $\Psi$, which is deferred to Section \ref{sec:energycomputation}.  

\medskip

{\it Acknowledgements.} M.B. and B.S. gratefully acknowledge financial support from the Swiss National Science Foundation through the Grant ``Dynamical and energetic properties of Bose-Einstein condensates''.  G.B., S.C.  and A.O. gratefully acknowledges financial support from the European Research Council through the ERC Starting Grant MaTCh, grant agreement n. 101117299. G.B., S.C. and A.O. also warmly acknowledge the GNFM (Gruppo Nazionale per la Fisica Matematica) - INDAM. A.O. acknowledges financial support from the MUR Grant “Dipartimento di Eccellenza 2023-2027” of Dipartimento di Matematica, Politecnico di Milano, and from the Politecnico di Milano ``Seed Fund Grant''.

\section{The Trial State} 
\label{sec:trial}

We consider the bosonic Fock space 
\begin{equation*}
	\mathcal{F}=\bigoplus_{n=0}^{+\infty} L^2(\Lambda )^{\otimes_\mathrm{sym}n}
\end{equation*}
constructed on the one-particle Hilbert space $L^2 (\Lambda)$, where $\Lambda= [ -L/2 ; L/2]^3$, with periodic boundary conditions. 

On $\cF$, we define the number of particle operator $\cN = \bigoplus_{n=0}^\infty n$ and the kinetic energy operator $\cK = \bigoplus_{n=0}^\infty \sum_{j=1}^n -\Delta_{x_j}$.

Moreover, for $f \in L^2 (\Lambda)$ we introduce the creation and annihilation operators $a^* (f), a(f)$. For $\Psi \in \cF$ and $n \in \bN$, their action is defined by    
\[ \begin{split} (a (f) \Psi)^{(n)}  (x_1, \dots , x_n) &= \sqrt{n+1} \int_{\Lambda}  dx \bar{f} (x) \Psi^{(n+1)} (x, x_1, \dots , x_n) \\ 
(a^* (f) \Psi)^{(n)} (x_1, \dots, x_n) &= \frac{1}{\sqrt{n}} \sum_{j=1}^n f (x_j) \Psi^{(n-1)} (x_1, \dots , x_{j-1}, x_{j+1}, \dots , x_n)\,. \end{split} \]
They satisfy the canonical commutation relations 
\begin{equation*}\label{eq:CCR} 
\big[ a (f) , a^* (g) \big] = \langle f , g \rangle , \qquad \big[ a (f) , a(g) \big] = \big[ a^* (f) , a^* (g) \big] = 0  
\end{equation*}  
for every $f,g \in L^2 (\Lambda)$. It is also convenient to define operator-valued distributions $a_x, a_x^*$, $x \in \Lambda$, so that 
\[  a(f) = \int \bar{f} (x) a_x \, dx , \qquad a^* (f) = \int f(x) a_x^* \, dx \]
and \[ [a_x, a_y^* ] = \delta (x-y), \qquad [a_x , a_y] = [a_x^* , a_y^*] = 0 \] for all $x, y \in \Lambda$. Explicitly, the action of $a_x$ is described by 
\begin{equation} \label{eq:ax} (a_x \Psi)^{(n)} (x_1, \dots , x_n) = \sqrt{n+1} \, \Psi^{(n+1)} (x, x_1, \dots , x_n) \end{equation} 
for arbitrary $\Psi \in \cF$ and $n \in \bN$. In terms of these operator-valued distributions, we can write 
\begin{equation}\label{eq:cNcK} \cN = \int dx \, a_x^* a_x \, ,  \qquad \cK = \int dx \nabla_x a_x^* \nabla_x a_x\,. \end{equation} 

In this section, we construct a normalized trial state $\Psi \in \cF$, satisfying the hard-core condition, with particle density $\langle \Psi, \cN \Psi \rangle / |\Lambda|$ close to $\rho > 0$ and energy density $\langle \Psi, \cK \Psi \rangle / |\Lambda|$ close to the Lee-Huang-Yang prediction, in the limit of small $\rho \frak{a}^3 \ll 1$. To reach this goal, we consider first a Bose-Einstein condensate with density $0 < \rho_0 < \rho$ to be fixed later on. The condensate is modeled by a coherent state $W (\rho_0) \Omega$, generated by the unitary Weyl operator 
\[ W (\rho_0) = e^{a^* (\sqrt{\rho_0}) - a (\sqrt{\rho_0})} = e^{\sqrt{\rho_0} \int dx \, (a_x^* - a_x)} \,  \]
acting on the vacuum vector $\Omega = \{ 1, 0, 0, \dots \}$. The action of $W(\rho_0)$ on creation and annihilation operators is determined by 
\begin{equation}\label{eq:shift} W(\rho_0)^* a_x W(\rho_0) = a_x + \sqrt{\rho_0}, \qquad W(\rho_0)^* a_x^* W(\rho_0) = a_x^* + \sqrt{\rho_0}\,. \end{equation} 
Hence, the coherent state $W (\rho_0) \Omega$ is an eigenvector of all annihilation operators, with 
\[ a_x W(\rho_0) \Omega = \sqrt{\rho_0} \, W(\rho_0) \Omega  \]
and its expected number of particles is given by 
\[ \langle W(\rho_0) \Omega, \cN W(\rho_0) \Omega \rangle =  \rho_0 L^3 =: N_0 \,.\]
More precisely, the number of particles in the coherent state 
\[ W (\rho_0) \Omega = e^{-N_0/2} \Big\{ 1, \sqrt{\rho_0}, \dots \frac{\rho_0^{n/2}}{\sqrt{n!}}, \dots \Big\} \]
is Poisson distributed, with average and variance given by $N_0$. 

The coherent state $W (\rho_0) \Omega$ describes a completely uncorrelated state. To make sure that our trial state satisfies the hard-core condition (\ref{eq:hc}) and that its energy is close to the true ground state energy, we need to add correlations.  To this end, we consider the solution
\begin{equation}\label{eq:0-en} f(x) = \left\{ \begin{array}{ll} 1 - \frac{\frak{a}}{|x|} , &\quad \text{if } |x| \geq \frak{a} \\ 0 , &\quad \text{if } |x| < \frak{a} \end{array} \right. \end{equation} 
of the zero energy scattering equation
\[ -\Delta f (x) = 0 , \quad \text{for } |x| \geq \frak{a} \]
with the hard-core condition $f(x) = 0$ for $|x| = \frak{a}$ and the boundary condition $f (x) \to 1$, as $|x| \to \infty$. We fix $\ell = \frak{a} (\rho \frak{a}^3)^{-\delta}$, for some $\delta > 0$ which will be chosen small enough, and we define 
\begin{equation}\label{def:f-ell} \begin{split}  f_\ell (x) &= \chi_\ell (x) f (x) + (1 - \chi_\ell (x)) \\ &= 1 - \chi_\ell (x) (1 - f(x)) = \left\{ \begin{array}{ll} 1 - \frac{\frak{a}}{|x|} \chi_\ell (x),  &\text{if } |x| \geq \frak{a} \\ 0 , &\text{if } |x| \leq \frak{a} \end{array} \right. \end{split}  \end{equation} 
where $\chi_\ell (x) = \chi (x/\ell)$ and $\chi \in C^\infty_c (\bR^3)$, with $\chi (x) = 1$ for $|x| \leq 2$, $\chi (x) = 0$ for $|x| > 4$. By definition, $f_\ell$ describes the two-particle correlations generated by the hard-core potential, up to the length scale $\ell = \frak{a} (\rho \frak{a}^3)^{-\delta}$, much larger than the scattering length $\frak{a}$ but much smaller than the healing length $(\rho \frak{a})^{-1/2}$. In the next lemma, we collect some important properties of $f_\ell$.
\begin{lemma}\label{lm:fell} Let $f_\ell$ be as defined in \eqref{def:f-ell} and 
\[
    \o_\ell (x) = 1 - f_\ell (x)= \left\{ \begin{array}{ll} 1 \; &\quad \text{if } |x| < \frak{a}  \\ \frac{\frak{a}}{|x|} \chi_\ell (x) &\quad \text{if } |x| \geq \frak{a} \,. \end{array} \right.  
\]
Then $0 \leq f_\ell(x) , \o_\ell (x) \leq 1$ and, for $|x|>\frak{a}$, 
\[ |\nabla f_\ell (x)| = |\nabla \o_\ell (x)| \leq C \frac{\frak{a}}{|x|^2} \chi_{2\ell} (x) \,. \]
Moreover,  
\begin{equation} \label{eq:nabla-fell}
\Big|\int |\nabla f_\ell(x)|^2 \dx - 4 \pi \aa \Big| \leq C \frac{\aa^2}{\ell} \,.
\end{equation}
%Moreover, setting $\o_\ell=1-f_\ell$ we have
%\[
%0 \leq \o_\ell(x) \leq  \frac{\aa \chi_\ell(x)}{|x|}\,, \qquad |\nabla \o_\ell(x)| \leq \frac{C \aa \chi_\ell(x)}{|%x|^2}\,.
%\]
\end{lemma}
%By definition
%\begin{equation} \label{eq:nabla-fell}
%\Big|\int |\nabla f_\ell|^2 dx - 4 \pi \aa \Big| \leq C \aa^2 \ell^{-1}\,.
%\end{equation}
\begin{proof} The bounds on $f_\ell, \o_\ell$ follow directly from (\ref{def:f-ell}). To prove (\ref{eq:nabla-fell}), we observe that 
\[ \int |\nabla f_\ell (x)|^2 dx = \int |\nabla \o_\ell (x)|^2 dx \]
and we compute 
\[ \nabla \o_\ell (x) = - \frac{\frak{a} x}{|x|^3} \chi (x/\ell) + \frac{\frak{a}}{\ell |x|} \nabla \chi (x/\ell)\,. \]
Hence 
\[ \begin{split} \int |\nabla f_\ell (x)|^2 dx = \; &\int_{|x| \geq \frak{a}} \frac{\frak{a}^2}{|x|^4} \chi^2 (x/\ell) \, dx \\ &+ \int_{|x| \geq \frak{a}} \frac{\frak{a}^2}{|x|^2 \ell^2} |\nabla \chi (x/\ell)|^2 dx - 2 \frac{\frak{a}^2}{\ell} \int \frac{x}{|x|^4} \cdot \nabla \chi (x/\ell) \chi (x/\ell) dx\,. \end{split} \]
With 
\[ \int_{|x| \geq \frak{a}} \frac{\frak{a}^2}{|x|^4} dx = 4\pi \frak{a} \]
and 
\[ \begin{split} \int_{|x| \geq 2 \ell}  \frac{\frak{a}^2}{|x|^4} dx \, ,  \; \;  \int_{2 \ell \leq |x| \leq 4 \ell} \frac{\frak{a}^2}{\ell^2 |x|^2} dx \, ,  \; \int_{2\ell \leq |x| \leq 4 \ell} \frac{\frak{a}^2}{\ell |x|^3} dx &\leq C \frac{\frak{a}^2}{\ell}  \end{split} \]
we obtain (\ref{eq:nabla-fell}). 
\end{proof}

With $f_\ell$, we introduce the operator $J: \cF \to \cF$, setting 
\begin{equation}\label{eq:def-Ja} (J \Phi)^{(n)}(x_1, \ldots, x_n) = \prod_{i<j}^n f_\ell (x_i-x_j) \Phi^{(n)} (x_1, \ldots, x_n) \, .
\end{equation} 
For $x \in \Lambda$, it is also convenient to define $J(x) : \cF \to \cF$ by   
\[ (J(x) \Phi)^{(n)} (x_1, \dots , x_n) = \prod_{j=1}^n f_\ell (x-x_j) \Phi^{(n)} (x_1, \dots , x_n)\,. \]
With (\ref{eq:ax}), we have, for $\Phi \in \cF$,  
\[ \begin{split} 
(a_x J \Phi)^{(n)} (x_1, \dots, x_n) &= \sqrt{n+1} \, (J\Phi)^{(n+1)} (x,x_1, \dots , x_n) \\ &= \sqrt{n+1} \, \prod_{k=1}^n f_\ell (x - x_k) \prod_{1 \leq i<j \leq n} f_\ell (x_i - x_j) \Phi^{(n+1)} (x, x_1, \dots , x_n) 
\\ &= (J(x) J a_x \Phi)^{(n)} (x_1, \dots , x_n)\,. \end{split} \]
In other words, 
\begin{equation}\label{eq:aJJ} a_x J = J(x) J a_x , \qquad J a_x^* = a_x^* J(x) J  \end{equation} 
for all $x \in \Lambda$. Similarly, we obtain the pull-through formulas 
\begin{equation}\label{eq:pull} a_y J(x) = f_\ell(x-y) J(x) a_y, \qquad J(x) a_y^* = f_\ell (x-y) a_y^* J(x) \end{equation} 
for all $x,y \in \Lambda$. We find $0 \leq J (x) , J  \leq 1$ for all $x \in \Lambda$. Recalling that $\omega_\ell=1-f_\ell$ and noting  
\[ 1- \prod_{j=1}^n f_\ell (x-x_j) \leq \sum_{j=1}^n (1-f_\ell (x-x_j)) = \sum_{j=1}^n \omega_\ell (x-x_j) \]
we obtain the useful bound   
\begin{equation} \label{eq:1-J} 0 \leq 1 - J(x) \leq d\Gamma ((\omega_\ell)_x) = \int dy \, \o_\ell (x-y) a_y^* a_y\,. \end{equation} 

The Jastrow factor generated by $J$ models correlations among particles, up to the scale $\ell = \frak{a} (\rho \frak{a}^3)^{-\delta}$. To resolve the energy to Lee-Huang-Yang precision, we need to describe correlations up to larger distances, comparable with the healing length $(\rho \frak{a})^{-1/2}$. To achieve this goal, we introduce a second length scale $\ell_0 = (\rho \frak{a})^{-1/2} (\rho \frak{a}^3)^{-\varepsilon} \gg \ell$, for some $\varepsilon > 0$, which will also be chosen small enough. The only restrictions on the choice of $\delta$ and $\varepsilon$ that will emerge along the proof are $0<\delta<\varepsilon/2$ and $C(\delta+\varepsilon)<1$ for some large enough $C>0$. These are certainly satisfied if e.g. $\varepsilon=1/(2C)$ and $\delta=1/(8C)$. We are going to implement correlations on scales $\ell \lesssim |x| \lesssim \ell_0$ through a Bogoliubov transformation 
\begin{equation}\label{eq:BT-T} T  = \exp \left[ \frac{1}{2} \int dx dy \, \eta (x-y) (a_x^* a_y^* - a_x a_y) \right] = \exp \Big[ \frac{1}{2} \sum_{p \in \Lambda^*} \widehat{\eta}_p \big(\hat a_p^* \hat a_{-p}^* - \hat a_p \hat a_{-p} \big) \Big] \end{equation} 
with an appropriate kernel $\eta$, with Fourier coefficients $\widehat{\eta}_p$, for $p \in \Lambda^* = 2\pi \bZ^3 / L$. Here, we define the Fourier coefficients of $h \in L^1 (\Lambda)$ by
\[ \widehat{h}_p = \int dx \, h (x) \, e^{-i p \cdot x} \]
for all $p \in \Lambda^*$, so that  
\[ h (x) = \frac{1}{|\Lambda|} \sum_{p \in \Lambda^* } \widehat{h}_p \, e^{ip \cdot x} \]
and we have the identity 
\begin{equation}\label{eq:parse} \langle h , g \rangle = \int dx \, \bar h (x) g (x) = \frac{1}{|\Lambda|} \sum_{p \in \Lambda^*} \overline{\widehat{h}}_p  \widehat{g}_p\,.  \end{equation}
In (\ref{eq:BT-T}), we also introduced the momentum-space creation and annihilation operators $\hat a_p = a (\varphi_p)$, $\hat a^*_p = a^* (\varphi_p)$, with the normalized plane wave $\varphi_p (x) = e^{i p \cdot x} / |\Lambda|^{1/2}$. The action of the unitary operator $T$ on the operators $\hat a_p, \hat a_p^*$ is explicitly given by 
\begin{equation}\label{eq:action-T} \begin{split} T^* \hat a_p T &= \widehat{\gamma}_p \hat a_p + \widehat{\sigma}_p \hat a_{-p}^* , \qquad T^* \hat a_p^* T = \widehat{\gamma}_p \hat a_p^* + \widehat{\sigma}_p  \hat a_{-p} \end{split} \end{equation} 
with the notation $\widehat{\gamma}_p = \cosh \widehat{\eta}_p$, $\widehat{\sigma}_p = \sinh \widehat{\eta}_p$, for $p \in \Lambda^*$. 

The choice of the kernel $\eta$ should achieve two goals. On short length scales $\ell \lesssim |x| \ll \ell_0$, $\eta$ should correct the function $f_\ell$ used in the definition (\ref{eq:def-Ja}) of the Jastrow factor and transform it approximately into $f_{\ell_0}$, which describes correlations up to the scale $\ell_0$. To this end, we have to make sure that 
\[ 1 + \frac{1}{\rho} \eta (x) \simeq \frac{f_{\ell_0} (x)}{f_\ell (x)} \]
or, equivalently,    
\begin{equation}\label{eq:short-corr} \eta (x) \simeq \rho \frac{\o_\ell (x) - \o_{\ell_0} (x)}{f_\ell (x)} \end{equation}
for $\ell \lesssim |x| \ll \ell_0$.  On the other hand, on length scales $|x| \simeq (\rho \frak{a})^{-1/2}$  comparable with the healing length, the action of $T$ should diagonalize the emerging quadratic Hamiltonian 
\begin{equation}\label{eq:Heff} H_\text{eff} = \sum_{p \in \Lambda^*_{+}} (p^2 + \rho_0 \widehat{V}_\text{eff} (p)) \hat a_p^* \hat a_p + \frac{1}{2} \sum_{p \in \Lambda^*_+} \rho_0 \widehat{V}_\text{eff} (p) \big( \hat a_p^* \hat a_{-p}^* + \hat a_p \hat a_{-p} \big) \end{equation} 
with an effective potential $V_\text{eff}$, having the same scattering length $\frak{a}$ as the original interaction, but no hard-core. The fact that the second order corrections to the ground state energy density can be obtained by diagonalizing an Hamiltonian of the form \eqref{eq:Heff} originates in Bogoliubov’s heuristic argument~\cite{Bog}. The validity of this heuristics has been rigorously established in recent years in the Gross–Pitaevskii regime, see e.g. \cite{BBCS2,HST,B}. 
In fact, it is convenient to define $V_\text{eff}$ using the solution (\ref{eq:0-en}) of the zero-energy scattering equation. On $\bR^3$, we find $2\Delta f (x) = 2 \delta (|x| - \frak{a}) / \frak{a}$, with the Fourier transform 
\[ \int \frac{2}{\frak{a}} \delta (|x| - \frak{a}) e^{-i p \cdot x} dx = 8\pi \frac{\sin |p| \frak{a}}{|p|}\,. \]
For $p \in \Lambda^*$, we define 
\begin{equation}\label{eq:Fou-Veff}  \widehat{V}_\text{eff} (p) =  8\pi \frac{\sin |p| \frak{a}}{|p|} \end{equation} 
and we denote by $V_\text{eff} (x)$ the periodic function on $\Lambda$, with Fourier coefficients (\ref{eq:Fou-Veff}). From Bogoliubov theory, to diagonalize (\ref{eq:Heff}) we must have 
\begin{equation}\label{eq:th2eta} \tanh (2 \widehat{\eta}_p) \simeq - \frac{\rho_0 \widehat{V}_\text{eff} (p)}{p^2 + \rho_0 \widehat{V}_\text{eff} (p)} \end{equation} 
for $|p| \simeq \rho^{1/2}$, which implies that $\sinh \widehat{\eta}_p \simeq \widehat{s}_p$, $\cosh \widehat\eta_p \simeq \sqrt{1 + \widehat{s}^2_p}$, with\\
\begin{equation}\label{eq:hatsk} \begin{split} 
\widehat{s}_p &= - \frac{\rho_0 \widehat{V}_\text{eff}(p)}{\sqrt{\Big(|p|^2+\rho_0\widehat{V}_\text{eff}(p)+\sqrt{|p|^4+2|p|^2\rho_0\widehat{V}_\text{eff}(p)}\Big)^2-\rho_0^2\widehat{V}_\text{eff}(p)^2}}\,.
%\widehat{s}_p &= - \frac{p^2 + \rho_0 \widehat{V}_\text{eff} (p) -  \sqrt{|p|^4 + 2p^2 \rho_0 \widehat{V}_\text{eff} (p)}}{\sqrt{\rho_0^2 \widehat{V}_\text{eff} (p)^2 - \Big( p^2 + \rho_0 \widehat{V}_\text{eff} (p) -  \sqrt{|p|^4 + 2p^2 \rho_0 \widehat{V}_\text{eff} (p)}\Big)^2}}\,.
% \\
%\widehat{g} (k) &= \frac{\rho_0 \widehat{V}_\text{eff} (k)}{\sqrt{\rho_0^2 \widehat{V}_\text{eff} (k)^2 - \big( k^2 + 
%\rho_0 \widehat{V}_\text{eff} (k) -  \sqrt{|k|^4 + 2k^2 \rho_0 \widehat{V}_\text{eff} (k)}\big)^2}} 
\end{split} 
\end{equation}
%leading to 
%\[ \begin{split} \widehat{s}^2 (k) &= \frac{|k|^2 + \rho_0 \widehat{V}_\text{eff} (k) - \sqrt{|k|^4 + 2k^2 \rho_0 
%\widehat{V}_\text{eff} (k)}}{2 \sqrt{|k|^4  + 2k^2 \rho_0 \widehat{V}_\text{eff} (k)}} \\
%\widehat{g}^2 (k) &= \frac{|k|^2 + \rho_0 \widehat{V}_\text{eff} (k) + \sqrt{|k|^4 + 2k^2 \rho_0 \widehat{V}_\text{eff} %(k)}}{2 \sqrt{|k|^4  + 2k^2 \rho_0 \widehat{V}_\text{eff} (k)}} 
%\end{split} \] 

Instead of choosing the kernel $\eta$, it is convenient to set  
\begin{equation}\label{def:tl-sigma} \wt{\sigma} (x) = \frac{\chi_{\ell_0} (2x) (1 - \chi_{\ell} (2x))}{f_\ell (x)} \big[ \rho_0 \o_\ell (x) + s (x) \big] \end{equation} 
with $s (x)$ defined through the Fourier coefficients (\ref{eq:hatsk}), for $p \not = 0$ (and $\widehat{s}_0 = 0$). By definition, $\tilde{\sigma}$ is supported on $\ell \leq |x| \leq 2\ell_0$. Moreover, for $|k| \gg (\rho \frak{a})^{1/2}$, (\ref{eq:hatsk}) implies that 
\[ \widehat{s}_k \simeq -\rho_0 \widehat{V}_\text{eff} (k) / 2k^2 \simeq - \rho_0 \, \widehat{\omega}_{\ell_0} (k), \]
 which is consistent with (\ref{eq:short-corr}). Here and below $\widehat{\omega}_{\ell_0} (k)$ denote the Fourier coefficients of $\omega_{\ell_0}$. Furthermore, for $\ell \ll |x| \lesssim \ell_0$, we have $\omega_\ell (x) = 0$ and $\tilde{\sigma} (x) = s (x)$, which is consistent with (\ref{eq:th2eta}). To stay orthogonal to the condensate, it is useful to modify $\tilde{\sigma}$, defining instead 
\begin{equation}\label{def:sigma}
\sigma(x) = \tilde \sigma(x)-\ell_0^{-3}\Big(\int \tl \sigma(z)dz\Big)\ph(x/\ell_0)
\end{equation}
 with $\ph$ smooth and compactly supported, with $\int \ph(x)dx=1$ and such that $\ph(x)=1/2$ for $|x|\leq 1$ and $\ph(x)=0$ for $|x| >  2$. Indeed, (\ref{def:sigma}) guarantees that $\widehat{\sigma}_0 = 0$. 
 
 With (\ref{def:sigma}), we can now define 
 \begin{equation}\label{eq:hatetap} 
	\widehat{\eta}_p =\sinh^{-1}(\widehat{\sigma}_p)\,.
\end{equation}
We will make use of the Bogoliubov transformation (\ref{eq:BT-T}), with kernel $\eta$ defined through its Fourier coefficients (\ref{eq:hatetap}). The action of $T$ on creation and annihilation operators is determined by (\ref{eq:action-T}), with $\widehat{\sigma}_p$ being the Fourier coefficients of the kernel (\ref{def:sigma}) and with $\widehat{\gamma}_p = \cosh \widehat\eta_p$. Some important properties of these kernels (and some other kernels that will play an important role in our analysis) are collected in the following lemma, whose proof is deferred to Appendix \ref{sec:eta}. 

\begin{lemma} \label{lm:eta}
The coefficients $\widehat{s}_k$, defined in \eqref{eq:hatsk}, satisfy 
\begin{equation} \label{eq:sfourier} |\widehat{s}_k| \leq C \min \Big\{  \frac{\rho^{1/4}}{|k|^{1/2}} , \frac{\rho}{k^2} \Big\} \end{equation} 
for all $k \in \Lambda^*_+$. For the function $s : \Lambda \to \bR$ with Fourier coefficients $\widehat{s}_k$ we find the pointwise bounds 
\begin{equation}\label{eq:s-point} | s (x)| \leq C \frac{\rho}{|x|} \min \Big\{ 1 , \frac{1}{(\rho^{1/2} |x|)^{3/2}} \Big\} , \qquad |\nabla s (x)| \leq C \frac{\rho |\hspace{-.05cm} \log \rho |}{x^2} \end{equation} 
for all $|x| \geq 2\frak{a}$. Moreover 
\begin{equation}\label{eq:est-combi}
 \int_{|x|\geq \ell_0} |s(x)|^2 dx  \leq C\ell_0^{-2}\rho^{\frac{1}{2}}  , \;   \int_{2\frak{a} \leq |x|\leq  2\ell} 
 |s(x)|^2 dx  \leq C\ell\rho^{2}, \;  \int_{|x|\geq \ell_0} |\nabla s(x)|^2 dx  \leq  C\ell_0^{-4}\rho^{\frac{1}{2}}\, . \end{equation} 
 %
 %
%\begin{align}
 %\label{Th:est_1}
%        \int_{|x|\geq \ell_0} |s(x)|^2\mathrm{d}x & \leq C\ell_0^{-2}\rho^{\frac{1}{2}},\\
%         \label{Th:est_1_1}
%        \int_{|x|\leq  2\ell} |s(x)|^2\mathrm{d}x & \leq C\ell\rho^{2},\\
%            \label{Th:est_2}
%        \int_{|x|\geq \ell_0} |\nabla s(x)|^2\mathrm{d}x &  \leq  C\ell_0^{-4}\rho^{\frac{1}{2}}\,.
%\end{align}
%
Let now $\zeta\in \{\tilde\sigma, \sigma,\gamma-\mathbbm{1},\mathbbm{1} - \gamma^{-1} , \gamma^{-1}*\sigma\}$. Then we have 
\begin{equation} \label{eq:zeta-combi1} 
 \|\zeta\|^2_2 \leq C \rho^{\frac{3}{2}}, \qquad \| \zeta \|_\infty \leq C\, \frac{\rho}{\ell} , \qquad  \|\zeta\|_1  \leq C \big(\rho^{\frac{1}{2}}\ell_0\big)^3 , 
 \end{equation} 
 and also 
 \begin{equation}\label{eq:zeta-combi2}
   \|\nabla \zeta\|^2_2 \leq C \rho^{2} , \qquad  \|\nabla \zeta\|_\io  \leq C \, \frac{\rho}{\ell^2}\,,
   \end{equation} 
%
%  \begin{align}
%            \label{Th:L_2_est}
%        \|\zeta\|^2_2& \leq C \rho^{\frac{3}{2}},\\
%          \label{Th:est_3_alt}
%             \|\zeta\|_1 & \leq C \big(\rho^{\frac{1}{2}}\ell_0\big)^3,\\
%\label{Th:est_3b} 
%\| \zeta \|_\infty &\leq C \rho / \ell \\ 
%  \label{Th:H_1_est}
%        \|\nabla \zeta\|^2_2& \leq C \rho^{2}  \\
%          \label{Th:est_4_alt_2}
% \|\nabla \zeta\|_\io & \leq C \frac{ \rho}{\ell^2}\ .
%\end{align} 
and, for any $m \in \bN$, the pointwise decay estimate 
\begin{equation}
            \label{eq:decay}
 |\zeta(x)|  \leq C \, \frac{\rho}{\ell}\Big( \frac{\rho^{\frac{1}{4}}\ell_0^{\frac{3}{2}}}{|x|}\Big)^m \leq C \, \frac{\rho}{\ell} \, \frac{1}{\big( \rho^{1/2+3\eps/2} |x| \big)^{m}} \,.
\end{equation}
For $\zeta \in \{ \tilde{\sigma}, \sigma \}$, we get stronger estimates. In particular
\begin{equation}\label{eq:s2-restricted}
\int_{|x|\geq \ell_0} \big[ |\tilde{\sigma} (x)|^2 + |\sigma (x)|^2 \big] dx  \leq C\ell_0^{-2}\rho^{\frac{1}{2}}  , \qquad  \int_{|x|\leq  2\ell} 
 \big[ |\tilde{\sigma} (x)|^2 + | \sigma (x)|^2 \big] dx \leq C\ell\rho^{2}, 
 \end{equation} 
and  
\begin{equation} 
            \label{Th:est_3}
             \|\sigma\|_1  \leq 2\|\tilde \sigma\|_1\leq C \big( \rho^{\frac{1}{2}}\ell_0\big)^\frac{1}{2}\,.
             \end{equation}
Additionally, we get the pointwise bounds 
\begin{equation}\label{eq:snabla-point} |\nabla \tilde{\sigma} (x)| , |\nabla \sigma (x) | \leq C \frac{\rho | \hspace{-.03cm} \log \rho|}{x^2} \end{equation} 
for all $x \in \Lambda$, which imply  
\begin{equation}\label{eq:snabla-p}\| \nabla \sigma \|_p, \| \nabla \tilde{\sigma} \|_p  \leq C \rho |\hspace{-.03cm} \log \rho  | \left\{ \begin{array}{ll}  \ell^{3/p-2} \quad \text{if } p > 3/2 \\ \ell_0^{3/p-2} \quad \text{if } p < 3/2 \end{array} \right.\end{equation} 
and
\begin{equation}\label{eq:nu-nabla-p} \| \nabla (\gamma^{-1}*\sigma) \|_p  \leq C (\rho^{1/2} \ell_0 )^3 \rho |\hspace{-.03cm} \log \rho | \left\{ \begin{array}{ll}  \ell^{3/p-2} \quad &\text{if } p > 3/2 \\  \ell_0^{3/p-2} \quad &\text{if } p < 3/2 \end{array} \right.\end{equation} 
for all $1 \leq p \leq \infty$.
\end{lemma}

We are now ready to define our trial state.  For $0<\rho_0<\rho$ to be fixed later 
%$\rho=\rho_0+\|\s\|_2^2$   
we set
\begin{equation}\label{eq:trial} 
\Psi = Z^{-1} J W(\rho_0) T \Omega \in \cF 
\end{equation} 
where $Z = \| J W (\rho_0) T \Omega \|_{\cF}$ is a normalization constant, ensuring that $\| \Psi \|_\cF = 1$. Our analysis is based on the observation that the action of annihilation operators on the trial state $\Psi$ can be expressed through a simple identity. 
\begin{lemma}\label{lm:axpsi} 
Let $\Psi$ be defined as in (\ref{eq:trial}). Then, for every $x \in \Lambda$, we have 
\begin{equation}\label{eq:id} a_x \Psi = \sqrt{\rho_0}J(x)\Psi+J(x)\int dy\,(\gamma^{-1} * \sigma)(x-y)a^*_y J(y)\Psi.
\end{equation}
\end{lemma} 
\begin{proof}
With (\ref{eq:aJJ}), we find, introducing the notation $\tau_x (y) = \tau (x-y)$, for $\tau = \gamma, \sigma$, 
\[ \begin{split} a_x \Psi  &= Z^{-1} a_x J W (\rho_0) T \Omega \\ &= Z^{-1} J(x) J a_x W(\rho_0) T  \Omega \\ &= Z^{-1} J(x) J W(\rho_0) T (a (\gamma_x) + a^* (\sigma_x) + \sqrt{\rho_0})  \Omega \\ &=  \sqrt{\rho_0} J(x) \Psi + Z^{-1} J(x) J a^* (\gamma^{-1} * \sigma_x) W(\rho_0)  T \Omega \end{split} \] 
using (\ref{eq:shift}) and the explicit action (\ref{eq:action-T}) of the Bogoliubov transform $T$, which translates to
\[ T^* a_x T = a (\gamma_x) + a^* (\sigma_x) , \qquad T^* a_x^* T = a^* (\gamma_x) + a (\sigma_x)  \]
in position space. Hence, applying again (\ref{eq:aJJ}), we arrive at 
\[ \begin{split}  a_x \Psi &= \sqrt{\rho_0} J(x) \Psi + Z^{-1} \int dy \, (\gamma^{-1} * \sigma) (x-y) J(x) J a_y^* W(\rho_0) T \Omega \\ &= \sqrt{\rho_0} J(x) \Psi +  \int dy \, (\gamma^{-1} * \sigma) (x-y) J(x) a_y^* J(y) \Psi \,. \end{split} \] 
\end{proof} 

In the next sections, we will use the identity (\ref{eq:id}) to derive precise estimates on the number of particles and on the energy in the trial state $\Psi$, as stated in the following two propositions. 
\begin{proposition} \label{prop:N_on_psi}
	We have 
	\begin{equation} \label{eq:cN_ub}
		\big| \langle \Psi,\mathcal{N}\Psi\rangle - (\rho_0+\|\sigma\|_2^2)L^3 \big| \leq C \rho^{7/4-11\eps-\delta} L^3
	\end{equation}    
    %\begin{equation} \label{eq:cN_ub}
		% \textcolor{purple}{\langle \Psi,
        %\mathcal{N}\Psi\rangle = (\rho_0+\|
        %\sigma\|_2^2)L^3 + R_\mathcal{N}}
        %%\big| \langle \Psi,
        %\mathcal{N}\Psi\rangle - (\rho_0+\|
        %\sigma\|_2^2)L^3 \big| \leq C 
        %\rho^{7/4-11\eps-\delta} L^3
	%\end{equation}
    %with
    %\[
    %   \textcolor{purple}{\big|%R_\mathcal{N}\big|\leq %C\rho^{7/4-11\eps-\d} L^3}
    %\]
	if $\eps , \delta > 0$ are small enough.  	
\end{proposition}

\begin{proposition} \label{prop:Psi-energy} 
%Let $\Psi$ be defined as in \eqref{eq:trial} with $\ell= \aa (\rho\aa^3)^{-\d}$ and $\ell_0=(\rho \aa)^{-1/2}(\rho\aa^3)^{-\eps}$. 
Let 
\begin{equation} \begin{split} \label{eq:E-rho}
E_{\rho} =\; &  \rho_0^2 \int |\nabla f_\ell(x)|^2 \dx  + 2\rho_0 \|\s\|^2  \int |\nabla f_\ell(x)|^2 \dx  + \int |\nabla f_\ell (x)|^2 |\sigma (x)|^2 dx \\ &+ 2\int f_\ell (x) \nabla f_\ell (x) \sigma (x) \nabla \sigma (x) dx  +  \int |f_\ell(x)|^2 |\nabla \s(x)|^2 \dx \\
		& + 2\rho_0 \int |\nabla f_\ell(x)|^2 \big[(\g \ast \s)(x) + (\s \ast \s)(x)\big] \dx \\
& + 2\rho_0 \int  f_\ell(x) \nabla f_\ell(x) \cdot \nabla \s (x)  \dx\,.
\end{split}
	\end{equation}
 Then, recalling (\ref{eq:cNcK}), there exists a constant $C>0$ such that
	\begin{equation} 
		L^{-3} \langle \Psi,\mathcal{K} \Psi\rangle \le E_{\rho}   + C \rho^{5/2+\delta} 
	\end{equation}
	for all $0 < \delta < \eps$ with $\eps$ small enough. 
\end{proposition}

\section{Proof of Theorem \ref{thm:main}} 
\label{sec:proof}

In the following we show how Theorem \ref{thm:main} follows from  Proposition \ref{prop:N_on_psi} and Proposition~\ref{prop:Psi-energy}. 

To this end, we fix the density $0 < \rho_0 < \rho$ of the condensate, requiring that \begin{equation}\label{eq:rho0} 
\rho + C_1 \rho^{7/4-11\eps-\delta} \leq \rho_0 + \| \sigma \|_2^2 \leq \rho + C_2 \rho^{7/4-11\eps-\delta}
\end{equation}
for appropriate constants $0 < C_1 < C_2$. From \eqref{eq:cN_ub}, this immediately implies (choosing $C_1 > 0$ large enough) that $\langle \Psi, \cN \Psi \rangle \geq \rho L^3$. 

To obtain (\ref{eq:rho0}), we set   \begin{equation}\label{eq:rho02} \rho_0 = \rho - \| \sigma^\rho \|_2^2 + 2C_1  \rho^{7/4-11\eps-\delta}, \end{equation} where $\sigma^\rho$ is defined as $\sigma$ in \eqref{def:tl-sigma}, \eqref{def:sigma}, but now with $\rho_0$ replaced everywhere by $\rho$ (in particular, we replace $\widehat{s}_p$ by new coefficients $\widehat{s}^\rho_p$, defined as in (\ref{eq:hatsk}), but with $\rho_0$ replaced by $\rho$). To show that \eqref{eq:rho02} leads to the desired bounds (\ref{eq:rho0}) we write $\widehat{s}_p = f (\rho_0 \widehat{V}_\text{eff} (p) / p^2)$, with 
\[ f(y) 
%= - \frac{1+y - \sqrt{1+2y}}{\sqrt{y^2 - \big(1 + y - \sqrt{1+2y}\big)^2}} 
%= - \frac{y}{(1+2y)^{1/4} (1+ \sqrt{1+2y})}
= - \frac{y}{\sqrt{2} (1+2y)^{1/4} \sqrt{1+y+\sqrt{1+2y}}} .
\]
Computing
\[ f' (y) = - \frac{1+\sqrt{1+2y}}{4 (1+2y)^{5/4}} \]
and recalling that for $\rho_0$ small enough $\rho_0\widehat{V}_\text{eff}(p)/p^2>-1/2$ where the derivative blows up,  we get
%Computing  
%\[ f' (y) = - \frac{1+\sqrt{1+2y}}{4 (1+2y)^{5/4}} \]
%we find $|f' (y)| \leq 1/2$ for all $y > 0$. 
%Thus
\[ |\widehat{s}_p - \widehat{s}^\rho_p | \leq C |\rho - \rho_0| \frac{|\widehat{V}_\text{eff} (p)|}{p^2} \]
and $\| s - s^\rho \|_2 \leq C |\rho - \rho_0| \leq C \rho^{3/2}$ (combining  \eqref{eq:rho02} and (\ref{eq:zeta-combi1})). This also implies (since $\| \omega_\ell \|_2 \leq \rho^{-\delta}$) 
\[ \| \wt{\sigma} - \wt{\sigma}^\rho \|_2 \leq C \rho^{3/2-\delta} \]
and, recalling that $\wt{\sigma} (z), \wt{\sigma}^\rho (z)$ vanish for $|z| \geq 2\ell_0$, $\| \sigma - \sigma^\rho \|_2 \leq C \rho^{3/2-\delta}$. Hence, \eqref{eq:rho0} follows from 
\[ \begin{split} 
\big| (\rho_0 + \| \sigma \|^2) - (\rho + 2C_1 \rho^{7/4-11\eps-\delta})  \big| &\leq \big| \| \sigma^\rho \|_2^2 - \| \sigma \|_2^2 \big| \\ &\leq \| \sigma - \sigma^\rho \|_2 \big( \| \sigma \|_2 + \| \sigma^\rho \|_2 \big) \leq C \rho^{9/4-\delta}  . \end{split} \]

%This is possible because $\| \sigma \|%_2^2 \simeq \rho_0^{3/2} \gg 
%\rho^{7/4-11\eps-\delta}$, if $\eps, 
%\delta > 0$ are small enough. 

Next, we focus on the energy density $\langle \Psi, \cK \Psi \rangle / L^3$ of the trial state $\Psi$. In the next lemma, we identify the major contributions to $E_\rho$.
\begin{lemma} \label{lm:E-rho}  
Let $E_\rho$ defined in  \eqref{eq:E-rho}. Recall the choices $\ell= \aa (\rho\aa^3)^{-\d}$ and $\ell_0=(\rho \aa)^{-1/2}(\rho\aa^3)^{-\eps}$. We set 
\begin{equation}\label{def:nu}
\begin{split} 
\th(x) = &\; - \rho_0 \omega_\ell (x) + f_\ell (x) \tilde{\sigma} (x) = - \rho_0 \chi_\ell (2x)  \omega_\ell (x) + \chi_{\ell_0} (2x) (1 - \chi_\ell (2x)) s (x) 
%\\ = & \; \chi_{\ell_0}(2x) \Big[ -\chi_\ell(2x) \rho_0 \o_{\ell_0}(x) + (1- \chi_\ell(2x))s(x)\Big]
\end{split} 
\end{equation} 
with $\tilde{\sigma}$ defined as in (\ref{def:tl-sigma}) and  we define 
\begin{equation} \label{def:tilde-E-rho}
\widetilde{E}_\rho = \| \nabla \th \|_2^2 + 16 \pi \aa \rho_0 \| s\|_2^2 + 8 \pi \aa \rho_0 \big((g-\mathbbm{1}) \ast s\big)(0) 
\end{equation}
with $g(x)$ defined through its Fourier coefficients $\widehat g_k= \sqrt{1 + \widehat{s}_k^2}$. Then,
\begin{equation}    \label{eq:calR}
|E_\rho - \widetilde{E}_\rho | \leq  C \rho^{5/2+\delta}   
\end{equation}
for all $0< \d <\eps/2$, with $\eps > 0$ small enough. 
\end{lemma}

\begin{proof} 
Decomposing $\gamma = (\gamma -\mathbbm{1}) + \mathbbm{1}$ we can write   
\begin{equation}\label{eq:E-rho-1} \begin{split} 
E_{\rho} =\; &  \rho_0^2 \int |\nabla f_\ell(x)|^2 \dx  +  \int |\nabla f_\ell(x)|^2 |\s(x)|^2 dx\\
		& +  2\int f_\ell(x) \s(x)   \nabla  f_\ell(x) \cdot \nabla \s(x) dx  +  \int |f_\ell(x)|^2 |\nabla \s(x)|^2 dx \\
		& + 2\rho_0 \int |\nabla f_\ell(x)|^2 \s(x) \dx + 2\rho_0 \int  f_\ell(x) \nabla f_\ell(x) \cdot \nabla \s(x)  \dx \\
&+  16 \pi \aa \rho_0 \| s\|_2^2 + 8 \pi \aa \rho_0   ((g-\mathbbm{1}) \ast \sigma)(0) + \sum_{i=1}^3 \mathcal{R}_i 
%+  16 \pi \aa \rho \| s\|_2^2 + 8 \pi \aa \rho   ((g-\d) 
%\ast \sigma)(0) + \sum_{i=1}^3 \mathcal{R}_i\, 
\end{split} \end{equation} 
where we introduced the notations 
\[\begin{split}
\mathcal{R}_1 =\; &  2\rho_0 \|\s\|_2^2  \int |\nabla f_\ell(x)|^2 \dx  - 8 \pi \aa \rho_0 \| s\|_2^2 \\
\mathcal{R}_2 = \; & 2\rho_0 \int  |\nabla f_\ell(x)|^2 ((\gamma-\mathbbm{1}) \ast \sigma)(x) \dx - 8 \pi \aa \rho_0   ((g-\mathbbm{1}) \ast \sigma)(0) \\
\mathcal{R}_3 = \; & 2\rho_0 \int |\nabla f_\ell(x)|^2 (\sigma \ast \sigma)(x) \dx - 8 \pi \aa \rho_0 \| s\|_2^2\,.
\end{split}\]
Recognizing that the first three lines on the r.h.s. of  \eqref{eq:E-rho-1} form a perfect square and recalling the definition (\ref{def:nu}), we conclude that 
\begin{equation} \begin{split} \label{eq:E-rho-2}
E_{\rho}  =\; & \| \nabla \{- \rho_0 \o_\ell +f_\ell \s\}\|_2^2 +  16 \pi \aa \rho_0 \| s\|_2^2 + 8 \pi \aa \rho_0   ((g-\mathbbm{1}) \ast \sigma)(0) + \sum_{i=1}^3 \mathcal{R}_i \, \\
=\; &\| \nabla \th\|_2^2+  16 \pi \aa \rho_0 \| s\|_2^2 + 8 \pi \aa \rho_0   ((g-\mathbbm{1}) \ast s)(0) + \sum_{i=1}^5 \mathcal{R}_i 
\end{split}
\end{equation}
with
	\[ \begin{split}
    \mathcal{R}_4 &\;=  \| \nabla \{- \rho_0 \o_\ell +f_\ell \s\}\|_2^2 - \| \nabla \{- \rho_0 \o_\ell +f_\ell \tilde\s\}\|_2^2\\[3pt]
        \mathcal{R}_5 &\;= 8 \pi \aa \rho_0 ((g-\mathbbm{1}) \ast (\sigma-s))(0)\,.
	\end{split}
	\]
To complete the proof of the lemma, we show that $|\mathcal{R}_i| \leq C \rho^{5/2+\eps/2}$, for all $i=1,\ldots, 5$, if $0 < \delta < \eps$ and $\eps > 0$ is small enough. 
	
	We focus first on $\mathcal{R}_1$. We split $\mathcal{R}_1= \mathcal{R}_{11}+\mathcal{R}_{12}$ with
	\[\begin{split}
	\mathcal{R}_{11} =\; &  2\rho_0 \big(\|\s\|^2_2 - \|s\|^2_2\big)  \int |\nabla f_\ell(x)|^2 \dx    \\
	\mathcal{R}_{12} =\; &  2\rho_0 \|s\|^2_2  \Big(\int |\nabla f_\ell(x)|^2 \dx  - 4 \pi \aa \Big)   \,.
    %\\\mathcal{R}_{13}  =\; &  8 \pi \aa \|s\|^2_2 \big(\rho_0 -\rho\big)
	\end{split}\]
	With \eqref{eq:nabla-fell} 
    %and recalling that we have chosen  $\rho_0$ so that $|\rho-\rho_0| \leq C \r^{3/2}$ 
    we have
	\[
	       |\mathcal{R}_{12}| \leq C  \rho^{5/2+\d}\,.%, \qquad |\mathcal{R}_{13}| \leq C  \rho^{3}\,.
	\]
	Next, we claim that $|\mathcal{R}_{11}|\leq C \rho^{5/2 +\eps}$. To prove this bound, we show that 
	\begin{equation}\label{eq:sig-s} \| \sigma - s \|_2 \leq C \rho^{3/4+\eps} \end{equation} 
	which, together with Lemma \ref{lm:eta}, implies that 
	\begin{equation}\label{eq:compare-ssig} \big| \| \sigma \|_2^2 - \| s \|_2^2 \big| \leq C \big( \| \sigma \|_2 + \| s \|_2 \big) \| \sigma -s \|_2 \leq C \rho^{3/2 + \eps} \end{equation} 
	and therefore that $|\mathcal{R}_{11}| \leq C \rho^{5/2+\eps}$ (the bound (\ref{eq:sfourier}) implies, in particular, that $\| s \|_2 \leq C \rho^{3/4}$). To prove (\ref{eq:sig-s}), we first estimate $\| \tilde{\sigma} - s \|_2$, with $\tilde{\sigma}$ as defined in (\ref{def:tl-sigma}). To this end, we decompose $\tilde{\sigma} (x)= \chi_\ell (x/2) \tilde{\sigma} (x) + (1 - \chi_\ell (x/2)) \tilde{\sigma} (x)$. Observing that $f_\ell (x) = 1$ (and thus $\omega_\ell (x) = 0$) on the support of $1-\chi_\ell (x/2)$, we find 
\[ \begin{split} \tilde{\sigma} (x) - s(x) = \; &\frac{\chi_\ell (x/2) (1 - \chi_\ell (2x))}{f_\ell (x)} \rho_0 \omega_\ell (x) + \frac{\chi_\ell (x/2) (1 - \chi_\ell (2x))}{f_\ell (x)} s (x) \\ &- \chi_{\ell} (x/2) s (x) - (1 - \chi_{\ell_0} (2x)) s (x)\,.  \end{split} \]
Since $f_\ell (x) \geq C > 0$, on the support of $1- \chi_\ell (2x)$, we arrive at
\[ \| \tilde{\sigma} - s \|_2 \leq C \rho \| \omega_\ell \|_2 + C \| \chi_{\ell} (\cdot/2) s \|_2  + \| (1 - \chi_{\ell_0} (2\cdot)) s \|_2 \leq C \rho^{3/4+\eps} \] 
if $\eps, \delta > 0$ are small enough. Here, we estimated $\| \omega_\ell \|_2 \leq C \ell^{1/2}$ and we applied Lemma~\ref{lm:eta} to bound the other terms. With 
\[ \| \tilde{\sigma} - \sigma \|_2 \leq  \ell_0^{-3}  \| \tilde{\sigma} \|_1   \| \ph (\cdot/\ell_0) \|_2 \leq C \rho^{3/4 + \eps} \]
we conclude that $\| \sigma - s \|_2 \leq C \rho^{3/4+\eps}$ (using $\| \tilde{\sigma} \|_1 \lesssim \rho^{-\eps/2}$, from Lemma \ref{lm:eta}). 
 
We switch to $\mathcal{R}_2$. We first rewrite $\mathcal{R}_2 = \sum_{i=1}^3 \mathcal{R}_{2i}$ with 
\[ \begin{split}
\mathcal{R}_{21} = \; & 2\rho_0 \int  |\nabla f_\ell(x)|^2 \big(((\gamma-\mathbbm{1}) \ast \sigma)(x)-((\gamma-\mathbbm{1}) \ast \sigma)(0)\big) \dx  \\
\mathcal{R}_{22} = \; & 2\rho_0 ((\gamma-\mathbbm{1}) \ast \sigma)(0) \big( \|\nabla f_\ell\|^2  - 4 \pi \aa \big)  \\
%\mathcal{R}_{23} = \; & 8 \pi \aa (\rho_0-\rho)   ((\gamma-\d) \ast \sigma)(0)   \\
\mathcal{R}_{23} = \; & 8 \pi \aa \rho_0  \big(((\gamma-\mathbbm{1}) \ast \sigma)(0)  - ((g-\mathbbm{1}) \ast \sigma)(0)\big) \,.\\
%\mathcal{R}_{24} = \; &  8 \pi \aa \rho  \big(((\gamma-\d) \ast \sigma)(0)  - ((g-\d) \ast \sigma)(0)\big) \\
\end{split}\] 

To bound $\mathcal{R}_{21}$ we use (recalling (\ref{eq:parse}))
\[
\big|((\gamma-\mathbbm{1}) \ast \sigma)(x)-((\gamma-\mathbbm{1}) \ast \sigma)(0)\big| \leq C |x| \frac{1}{|\L|} \sum_{p\in \L^*} |p||\widehat{(\g-\mathbbm{1})}(p)| |\widehat{\s}_p| \leq C |x| \| \g-\mathbbm{1}\|_2  \| \nabla \s\|_2\,.
\]
With Lemma \ref{lm:fell} and Lemma \ref{lm:eta}, we find 
\[
|\mathcal{R}_{21}| \leq C \rho \| \g-\mathbbm{1}\|_2  \| \nabla \s\|_2  \int |\nabla f_\ell(x)|^2  |x|  \dx  \leq  C \rho^{11/4 -\d} \,.
\]
On the other hand, Lemma \ref{lm:fell} and the bound $((\gamma-\mathbbm{1}) \ast \sigma)(0) \leq \| \g-\mathbbm{1}\|_2  \| \s\|_2 \leq C \rho^{3/2}$ yield
\[
\mathcal{R}_{22} \leq C \rho^{5/2+\d} \,.
\]
Finally, to bound $\mathcal{R}_{23}$ we write
\[
((\g - g) \ast \s)(0) = \frac 1 {|\L|} \sum_{p \in \L^*} (\widehat \g_p - \widehat g_p ) \widehat \s_p = \frac 1 {|\L|} \sum_{p \in \L^*} \Big( \sqrt{1+\widehat \s_p^2} - \sqrt {1+\widehat s^2_p}  \, \Big) \widehat \s_p\,.
\]
Using the fact that $x\to \sqrt{1+x^2}$ is Lipshitz and (\ref{eq:sig-s}), we find
\[
\big|( (\g - g) \ast \s)(0)\big| \leq  \frac C {|\Lambda|} \sum_{p \in \L^*} |\widehat \s_p- \widehat s_p| | \widehat \s_p | \leq C \|  \s-  s\|_2 \|  \s\|_2 \leq C \rho^{3/2+\eps}  
\]
implying that $|\mathcal{R}_{23}| \leq C \rho^{5/2+\eps}$. 

The term $\mathcal{R}_3$ can be bounded similarly to $\mathcal{R}_2$, using \eqref{eq:sig-s}, $
(\s \ast \s)(0)=\|\s\|^2_2$, and the bound
\[
\big|(\s \ast \s)(x)-(\s \ast \s)(0)\big| \leq C |x| \frac{1}{|\L|} \sum_{p \in \L^*} |p| |\widehat{\s}_p|^2 \leq C |x| \| \s\|_2 \| \nabla \s\|_2\,.
\]

The term $\mathcal{R}_4$ can be written as 
\begin{equation*} \mathcal{R}_4 =   \langle ( \nabla  f_\ell) (\sigma - \tilde{\sigma}) + f_\ell (\nabla \sigma - \nabla \tilde{\sigma}), -2 \rho_0 \nabla \omega_\ell + (\nabla f_\ell) \sigma + f_\ell (\nabla \sigma) + (\nabla f_\ell) \tilde{\sigma} + f_\ell (\nabla \tilde{\sigma})\rangle\,. \end{equation*} 
%We claim that 
%\begin{equation}\label{eq:R5-claim} |R_5| \leq C \rho^{5/2+\eps/2} \end{equation}
%if $0 < \delta < \eps$, and $\eps$ is small enough.
 Recalling that 
\[ \sigma (x) - \tilde{\sigma} (x) = - \ell_0^{-3} \ph (x/\ell_0) \Big( \int \tilde{\sigma} (z) dz \Big) \]
and (again from Lemma \ref{lm:eta}) that $\| \tilde{\sigma} \|_1 \leq C \rho^{-\eps/2}$, we find 
\begin{equation}\label{eq:R51}  \|\sigma - \tilde{\sigma}\|_\infty \leq C \rho^{3/2 +3\eps} .\end{equation} 
%and 
%\begin{equation}\label{eq:R53} \begin{split} \int_{|x| \leq 4 \ell} |\nabla \sigma (x) - \nabla \tilde \sigma %(x)|^2 dx &\leq C \rho^{-\eps} \ell^3 \ell_0^{-8} \leq C \rho^{4+7\eps-3\delta} \leq C \rho^{4+4\eps} \, ,  %\end{split} \end{equation} 
Using (\ref{eq:R51}) together with the bounds $\| \nabla f_\ell \|_2 \leq C$ and also $\| \nabla \sigma \|_2, \| \nabla \tilde{\sigma} \|_2 \leq C \rho$ as well as $\|  \sigma \|_\infty, \|  \tilde{\sigma} \|_\infty \leq C \rho$, we obtain, by Cauchy-Schwarz, 
\[  \big| \langle ( \nabla  f_\ell) (\sigma - \tilde{\sigma}) , -2 \rho_0 \nabla \omega_\ell + (\nabla f_\ell) \sigma + f_\ell (\nabla \sigma) + (\nabla f_\ell) \tilde{\sigma} + f_\ell (\nabla \tilde{\sigma})\rangle \big| \leq C \rho^{5/2+5/2\eps} \leq C \rho^{5/2+\delta} \]
if $\delta > 0$ is small enough. Moreover, 
\[ \langle f_\ell (\nabla \sigma - \nabla \tilde{\sigma}) ,  -2 \rho_0 \nabla \omega_\ell + (\nabla f_\ell) \sigma + (\nabla f_\ell) \tilde{\sigma} \rangle = 0 \]
because the support of $\nabla \sigma - \nabla \tilde{\sigma}$ is disjoint from the support of $\nabla f_\ell$. Finally, we estimate 
\[ \begin{split} | \langle f_\ell (\nabla \sigma - \nabla \tilde{\sigma}) , f_\ell \nabla \sigma \rangle | &\leq C \rho^{-\eps/2} \ell_0^{-4} \Big| \int f^2_\ell (x) \nabla \sigma (x) \cdot \nabla \ph (x/\ell_0) dx \Big| \\ &\leq C  \rho^{-\eps/2} \ell_0^{-4} \int_{|x| \geq \ell_0} |\nabla \sigma (x)| |\nabla \ph (x/\ell_0)| dx  \end{split} \]
because $\nabla \ph (x) = 0$, if $|x| \leq 1$. With Cauchy-Schwarz and Lemma \ref{lm:eta}, we conclude that 
\[  | \langle f_\ell (\nabla \sigma - \nabla \tilde{\sigma}) , f_\ell \nabla \sigma \rangle |  \leq C \ell_0^{-3} \rho^{1 - \eps/2} |\log \rho |\leq C \rho^{5/2+\eps/2}  \leq C \rho^{5/2+\delta} . \]
The contribution $\langle f_\ell (\nabla \sigma - \nabla \tilde{\sigma}) , f_\ell \nabla \tilde{\sigma} \rangle$ can be handled similarly. This shows that $|\mathcal{R}_4| \leq C \rho^{5/2+\delta}$.

Finally, recalling the definition of $\widehat g_k$ above \eqref{eq:calR} and the bound  (\ref{eq:sig-s}), we find 
\[
|\mathcal{R}_5| \leq C \rho \| g-\mathbbm{1} \|_2 \| \s-s\|_2 \leq C \rho \| s \|_2 \| \s-s\|_2 \leq C \rho^{5/2+\eps}\,.
\]
which concludes the proof of the lemma.
%
%can estimate, by Cauchy-Schwarz, all contributions to the inner product in (\ref{eq:R5pro}) that %contain at least one factor $\nabla f_\ell$ or $\nabla \omega_\ell = -\nabla f_\ell$ (all these terms are %given by integrals restricted to $|x| \leq 4 \ell$); they all satisfy (\ref{eq:R5-claim}). The remaining %contributions can be estimated as follows:  
%\[ \begin{split} | \langle f_\ell (\nabla \sigma - \nabla \tilde{\sigma}) , f_\ell \nabla \sigma \rangle | &\leq %C \rho^{-\eps/2} \ell_0^{-4} \Big| \int f^2_\ell (x) \nabla \sigma (x) \cdot \nabla \ph (x/\ell_0) dx \Big| \\ %&\leq C  \rho^{-\eps/2} \ell_0^{-4} \int_{|x| \geq \ell_0} |\nabla \sigma (x)| |\nabla \ph (x/\ell_0)| dx  
%\end{split} \]
%because $\nabla \ph (x) = 0$, if $|x| \leq 1$. With Cauchy-Schwarz and Lemma \ref{lm:eta}, we %conclude that 
%\[  | \langle f_\ell (\nabla \sigma - \nabla \tilde{\sigma}) , f_\ell \nabla \sigma \rangle |  \leq C 
%\ell_0^{-9/2} \rho^{1/4-\eps/2} \leq \rho^{5/2+4\eps} \]
%The contribution $\langle f_\ell (\nabla \sigma - \nabla \tilde{\sigma}) , f_\ell \nabla \tilde{\sigma} 
%\rangle$ can be handled similarly. This shows (\ref{eq:R5-claim}). 
\end{proof}

We focus now on the expression for $\widetilde{E}_\rho$ in \eqref{def:tilde-E-rho} and show it can be bounded above by the Lee-Huang-Yang energy. 
\begin{lemma} \label{lm:tilde-E-rho}
Let $\widetilde E_\rho$ be defined as in  \eqref{def:tilde-E-rho}. Then
\[
\widetilde E_\rho \leq  4 \pi  \aa \rho^2 \Big[ 1 + \frac{128}{15\sqrt{\pi}} (\rho \frak{a}^3)^{1/2}\Big] + C \rho^{5/2+\eps}
\]
for all $\d, \eps > 0$ small enough. 
\end{lemma}

\begin{proof}  We start from the observation that 
\begin{equation}\label{eq:nabla-nu}
\| \nabla \th\|^2 \leq  \|\nabla (\th + {\rho_{0}}\omega \chi_{\ell_0} (2\cdot)) \|^2 + 4\pi \frak{a} \rho^2_0 + C \rho^{5/2+\eps} \,.
\end{equation}
To prove (\ref{eq:nabla-nu}), we first observe, from the definition (\ref{def:nu}) of $\th$, that 
\begin{equation}\label{eq:nuomega} \th (x) + \rho_{0} \omega (x) \chi_{\ell_0} (2x) =  \chi_{\ell_0}(2x) (1-\chi_\ell(2x))(\rho_0 \omega (x) + s(x))\,. \end{equation} 
%which in particular implies that 
%\[ \big| \nu (x) + \rho_{0} \omega (x) \chi_{\ell_0} (2x) \big| \lesssim \rho \frac{\chi_{\ell_0} (2x) (1 - 
%\chi_\ell (2x))}{|x|}  \]
Since moreover 
\[ \big|(1 - \chi_\ell (2x))  \Delta (\omega (x) \chi_{\ell_0} (2x)) \big| \lesssim \frac{1}{\ell_0^{2} |x|} \chi (\ell_0 \leq |x| \leq 2 \ell_0) \]
because $\Delta \omega (x) =0$ on the support of $1- \chi_\ell (2x)$, we conclude that  
\[ \begin{split} \big| \langle \nabla (\th &+ \rho_0 \omega \chi_{\ell_0} (2 \cdot)) ,  \rho_0 \nabla (\omega \chi_{\ell_0} (2 \cdot)) \rangle \big| \\ &=  \big| \langle \th + \rho_0 \omega \chi_{\ell_0} (2\cdot) , \rho_0 \Delta (\omega \chi_{\ell_0} (2\cdot)) \rangle \big| \\ &\leq C \rho \ell_0^{-2} \int_{\ell_0 \leq |x| \leq 2 \ell_0} \big(\rho_0 \omega (x) + s(x)) \frac{1}{|x|} dx  \\ &\leq C \rho^2 \ell_0^{-1} + C \rho \ell_0^{-2} \Big[ \int_{|x| \leq 2 \ell_0} \frac{1}{|x|^2} dx \Big]^{1/2} \Big[ \int_{|x| \geq \ell_0} |s(x)|^2 dx \Big]^{1/2}  \leq C \rho^{5/2+\eps}\,. \end{split} \]
Writing $\th = (\th + \rho_0 \omega \chi_{\ell_0} (2 \cdot)) - \rho_0 \omega \chi_{\ell_0} (2 \cdot)$, we can therefore estimate
\[ \| \nabla \th \|^2 \leq \| \nabla (\th + \rho_0 \omega \chi_{\ell_0} (2 \cdot)) \|^2 + \rho_0^2 \| \nabla (\omega \chi_{\ell_0} (2 \cdot)) \|^2 + C \rho^{5/2+\eps}\,. \]
From $|\nabla  (\chi_{\ell_0} (2x))| \leq C \ell_0^{-1} \chi (|x| \leq 2\ell_0) $, we also find 
\[ \begin{split}  \rho_0^2 \| \nabla (\omega \chi_{\ell_0} (2 \cdot)) \|^2 &\leq \rho_0^2 \int_{|x| \leq 2 \ell_0}  |\nabla \omega (x)|^2 dx + C \rho^{5/2+\eps} \\ & \leq 4\pi \frak{a} \rho_0^2 \big( 1 + C \frak{a} / \ell_0) + C \rho^{5/2+\eps} \leq 4\pi \frak{a} \rho_0^2 + C \rho^{5/2+\eps} \end{split}  \]
and thus (\ref{eq:nabla-nu}). Recalling the definition (\ref{eq:rho0}) of $\rho_0$ and the bound (\ref{eq:compare-ssig}), (\ref{eq:nabla-nu}) yields  
\[ \| \nabla \th \|^2 + 8\pi \frak{a} \rho \| s \|^2 \leq 4 \pi \frak{a} \rho^2 + \| \nabla (\th + \rho_0 \omega \chi_{\ell_0} (2 \cdot)) \|^2 + C \rho^{5/2 + \eps}\,. \]
From (\ref{def:tilde-E-rho}), we arrive at \begin{equation}\label{eq:Erho1} \begin{split} \widetilde{E}_{\rho} \leq \; &4 \pi \frak{a} \rho^2 + \left\|\nabla (s + {\rho_{0}}\omega_{\ell_0}) \right\|^2 \\
&+\frac{1}{|\L|}\sum_{k\in \L^*_+} \rho_0\widehat{V}_\mathrm{eff}(k) \Big[  \widehat s_k^2 +  \widehat{(g-\mathbbm{1})}_k \widehat s_k\Big] +\sum_{i=1}^3 \mathcal{E}_i + C \rho^{5/2+\eps} 
\end{split}\end{equation} 
with the error terms 
\begin{align*}
\mathcal{E}_1 =\; & \left\|\nabla (\th + {\rho_{0}}\omega \chi_{\ell_0} (2\cdot)) \right\|^2 - \left\|\nabla (s + {\rho_{0}}\omega_{\ell_0}) \right\|^2\\
  \mathcal{E}_2 =\; & 8\pi  \mathfrak{a} {\rho_0} \|s\|^2 -\frac{1}{|\L|}\sum_{k\in \L^*_+} \rho_0\widehat{V}_\mathrm{eff}(k)  \widehat s_k^2 \\
 \mathcal{E}_3 =\; & 8\pi  \mathfrak{a} \rho_0\,   ((g-\mathbbm{1}) \ast s)(0) - \frac{1}{|\L|}\sum_{k\in \L^*_+} \rho_0\widehat{V}_\mathrm{eff}(k)\,  \widehat{(g-\mathbbm{1})}(k) \widehat s_k \,.
\end{align*}

Expanding the norm $\| \nabla (s + \rho_0 \omega_{\ell_0}) \|^2$ in Fourier space, we further obtain
\begin{equation} \begin{split}
\widetilde{E}_\rho \leq \; & 4 \pi \aa \rho^2 
+\frac{1}{|\L|}\sum_{k\in \L^*_+} \Big[ |k|^2 \widehat s_k^2  + {\rho_{0}}\widehat{V}_\mathrm{eff}(k) \big(\widehat{s}_k^2+ \widehat{g}_k \widehat s_k\big) +\frac{\rho^2_0 \widehat{V}^2_\mathrm{eff}(k)}{4 |k|^2}\Big] \\ &+\sum_{i=1}^5 \mathcal{E}_i + C \rho^{5/2+\eps} 
\end{split}\end{equation}
with 
\[
\mathcal{E}_4 =   \frac{\rho_{0}}{|\L|}\sum_{k\in \L^*_+} \big( 2|k|^2\widehat \omega_{\ell_0}(k) -\widehat{V}_\mathrm{eff}(k) \big) \widehat s_k\,, \quad \mathcal{E}_5 \;=  \frac{\rho_0^2}{|\L|}\sum_{k\in \L^*_+}\Big(  |k|^2  \widehat \omega^2_{\ell_0}(k) -  \frac{ \widehat{V}^2_\mathrm{eff}(k)}{4 |k|^2} \Big)\,.
\]
With the definition (\ref{eq:hatsk}), we have, for every $k \neq 0$, 
\[
\widehat s_k^2=  \frac{|k|^2 +\rho_0\widehat V_\mathrm{eff}(k)- \sqrt{|k|^4 + 2  |k|^2 \rho_0\widehat V_\mathrm{eff}(k)} }{2\sqrt{|k|^4 + 2 |k|^2 \rho_0\widehat V_\mathrm{eff}(k)}}\,,\quad
\widehat g_k\widehat s_k=  \frac{-\rho_0\widehat V_\mathrm{eff}(k)}{2\sqrt{|k|^4 + 2  |k|^2 \rho_0\widehat V_\mathrm{eff}(k)}}\,.
\]
Therefore, we can rewrite 
\begin{equation}  \label{eq:E_rho-1}
\widetilde{E}_\rho \leq  4 \pi \aa \rho^2  +\frac{1}{2|\L|}\sum_{k\in \L^*_+} F(k)+\sum_{i=1}^5 \mathcal{E}_i + C \rho^{5/2+\eps} 
\end{equation}
with
\begin{equation}\label{def:Fk} 
F(k):=  \sqrt{|k|^4 + 2 |k|^2 \rho_0\widehat V_\mathrm{eff}(k)} - |k|^2 - \rho_0 \widehat V_\mathrm{eff}(k)+\frac{\rho^2_0 \widehat{V}^2_\mathrm{eff}(k)}{2 |k|^2}\,.
\end{equation}
Rewriting 
\begin{equation}\label{eq:Fk2} F(k) = - \frac{\rho_0^2 \widehat{V}_\text{eff}^2 (k)}{\sqrt{|k|^4 + 2 |k|^2 \rho_0\widehat V_\mathrm{eff}(k)} + |k|^2 + \rho_0 \widehat{V}_\text{eff} (k)} + \frac{\rho^2_0 \widehat{V}^2_\mathrm{eff}(k)}{2 |k|^2} \end{equation} 
we immediately find $|F(k)| \leq C \rho^2 / |k|^2$ (recall $|\widehat{V}_\text{eff} (k)| \leq C$). Taylor expanding the square root in (\ref{def:Fk}), we also obtain 
\[ |F(k)| \leq C \frac{\rho^3}{|k|^4} \]
for $|k| \geq C \rho_0^{1/2}$. With these two bounds, we conclude that the contributions to the sum of $F(k)$ arising from $|k| \leq \rho_0^{1/2+\eps}$ and from $|k| \geq \rho_0^{1/2-\eps}$ are negligible. Hence
\[ \widetilde{E}_\rho \leq  4 \pi \aa \rho^2  +\frac{1}{2|\L|}\sum_{\substack{k\in \L^*_+: \\ \rho_0^{1/2+\eps} \leq |k| \leq \rho_0^{1/2-\eps}}} F(k)+\sum_{i=1}^5 \mathcal{E}_i + C \rho^{5/2+\eps}\,. \]
Replacing $\widehat{V}_\text{eff} (k) = 8\pi \sin (|k| \frak{a}) / |k|$ with $8\pi \frak{a}$ in (\ref{def:Fk}), we define 
\[ G(k) = \sqrt{|k|^4 + 16\pi  \frak{a} \rho_0  |k|^2} - |k|^2 - 8\pi \frak{a} \rho_0 + \frac{(8\pi \frak{a} \rho_0)^2}{2 |k|^2}\,. \]
Expressing also $G$ in the form (\ref{eq:Fk2}) and recalling that $|\widehat{V}_\text{eff} (k) - 8\pi \frak{a}| \leq C |k|^2$, we find 
\[ |F(k) - G(k)| \lesssim \frac{\rho^2 |\widehat{V}_\text{eff} (k) - 8\pi \frak{a}|}{|k|^2} \leq C \rho^2 \]
and therefore  
\[  \frac{1}{|\Lambda|} \sum_{|k| \leq \rho_0^{1/2-\eps}} |F(k) - G(k)| \leq C  \rho^{7/2-3\eps} \leq C \rho^{5/2+\eps} \]
if $\eps > 0$ is small enough. Hence
\[ \widetilde{E}_\rho \leq  4 \pi \aa \rho^2  +\frac{1}{2|\L|}\sum_{\substack{k\in \L^*_+: \\ \rho_0^{1/2+\eps} \leq |k| \leq \rho_0^{1/2-\eps}}} G (k)+\sum_{i=1}^5 \mathcal{E}_i + C \rho^{5/2+\eps}\,. \]
Observing that \[ |\nabla G (k)| \leq C \frac{\rho^2}{|k|^3} \Big(1 + C \frac{\rho}{k^2} \Big) \leq C \rho^{1/2-4\eps} \] on the annulus $\rho_0^{1/2+\eps} \leq |k| \leq \rho_0^{1/2-\eps}$, we conclude that, in every relevant cell of the lattice $\Lambda^* = 2\pi \bZ^3/ L$, the variation of $G$ is bounded by $C \rho^{1/2-4\eps} / L$. This implies that, for sufficiently large $L$, 
\[ \begin{split}  &\frac{1}{|\Lambda|} \sum_{k\in \L^*_+: \rho_0^{1/2+\eps} \leq |k| \leq \rho_0^{1/2-\eps}} G (k) \\ &\leq \frac{1}{(2\pi)^3} \int_{\rho_0^{1/2+\eps} \leq |k| \leq \rho_0^{1/2-\eps}}  \Big[ \sqrt{|k|^4 + 16\pi  \frak{a} \rho_0  |k|^2} - |k|^2 - 8\pi \frak{a} \rho_0 + \frac{(8\pi \frak{a} \rho_0)^2}{2 |k|^2} \Big] dk + C \rho^{5/2+\eps}  \\ &= \frac{\rho_0^{5/2}}{(2\pi)^3} \int_{\rho_0^{\eps} \leq |k| \leq \rho_0^{-\eps}}  \Big[ \sqrt{|k|^4 + 16\pi  \frak{a}  |k|^2} - |k|^2 - 8\pi \frak{a} + \frac{(8\pi \frak{a})^2}{2 |k|^2} \Big] dk +C \rho^{5/2+\eps}  \\  &\leq \frac{\rho_0^{5/2}}{(2\pi)^3} \int  \Big[ \sqrt{|k|^4 + 16\pi  \frak{a}  |k|^2} - |k|^2 - 8\pi \frak{a} + \frac{(8\pi \frak{a})^2}{2 |k|^2} \Big] dk +C \rho^{5/2+\eps} \\ &= 4\pi \frak{a} \rho^2 \cdot \frac{128}{15 \sqrt{\pi}} (\rho \frak{a}^3)^{1/2} + C \rho^{5/2+\eps} \,. \end{split} \]
In the last step, we computed the integral and we used the bound $|\rho - \rho_0| \leq C \rho^{3/2}$ to replace $\rho_0^{5/2}$ by $\rho^{5/2}$. We find  
\[ \wt{E}_\rho \leq 4 \pi \frak{a} \rho^2 \left( 1 + \frac{128}{15 \sqrt{\pi}} (\rho \frak{a}^3)^{1/2} \right) + \sum_{i=1}^5 \mathcal{E}_i + C \rho^{5/2+\eps}\,. \]
 To conclude the proof of the lemma, we estimate the error terms $\mathcal{E}_i, i=1,\dots ,5$. 

%da qui prendiamo potenziale 8 pi sin ka / k, ma qualsiasi potentiziale che e' piu' o meno costante su |%k| < \rho^{1/2-\eps} andrebbe bene. ... poi si mete a posto il resto

% giulia

% arnaud

To bound $\mathcal{E}_1$, we define $Z(x)= \rho_0 \omega_{\ell_0} (x) +s(x)$. Since $\omega (x) = \omega_{\ell_0} (x)$ on the support of $\chi_{\ell_0} (2x)$, from (\ref{eq:nuomega}) we can write  
\begin{equation}\label{eq:E1-Z}
\th(x) + \rho_0 \omega (x) \chi_{\ell_0} (2x) = (1 -\chi_\ell (2x)) Z(x) - (1- \chi_{\ell_0} (2x)) Z(x)
\end{equation} 
which leads us to 
   \begin{equation}  \label{Eq:Explicit_E_2} 
   \begin{split} 
       | \mathcal{E}_1| 
        %& =\left\|\nabla Z\right\|^2-\left\|\nabla \left\{(1-\chi_\ell)Z-(1-\chi_{\ell_0})\widehat{q}\right\}\right\|
         \leq \; &\left\langle \nabla Z, \left[(1  -  \chi_\ell (2\cdot) )^2 -1 \right] \nabla Z \right\rangle  \\ &+ 2 \big|  \langle (1-\chi_\ell (2\cdot)) \nabla Z, -\frac{2}{\ell} \nabla \chi (2\cdot / \ell) Z(x) - \nabla \{ (1 - \chi_{\ell_0} (2\cdot)) Z \} \rangle\big|  \\ &+ C\ell^{-2}  \|  \nabla \chi (2\cdot / \ell) Z \|^2 + C \| \nabla \{ (1 - \chi_{\ell_0} (2\cdot)) Z \} \|^2 \end{split}  \end{equation} 
 and implies that 
 \begin{equation}\label{eq:E1ZZ} |\mathcal{E}_1| \leq C \ell^3 \| \nabla Z \|_\infty^2 + C \ell^2 \| \nabla Z \|_\infty \| Z \|_\infty + C \ell \| Z \|_\infty^2 + C \| \nabla Z |_{|x| \geq \ell_0} \|^2 + C\ell_0^{-2} \| Z |_{|x| \geq \ell_0} \|^2.  \end{equation} 
 From Lemma \ref{lm:eta} and noticing that 
\[ \int_{\ell_0 \leq |x| \leq 4 \ell_0} |\omega (x)|^2 dx  \leq C \ell_0  , \qquad \int_{\ell_0 \leq |x| \leq 4 \ell_0} |\nabla \omega (x)|^2 \leq C \ell_0^{-1} \]
we arrive at
 \[ C \ell_0^{-2} \| Z |_{|x| \geq \ell_0} \|_2^2  \leq C\rho^{5/2+\eps} , \qquad C  \| \nabla Z |_{|x| \geq \ell_0} \|^2_2 \leq C \rho^{5/2+\eps}\,. \]
To bound the first three term on the r.h.s. of (\ref{eq:E1ZZ}), we first estimate the Fourier coefficients $\widehat{Z} (k)$. From Lemma \ref{lm:eta}, we have $|\widehat{Z} (k)| \leq C \rho / |k|^2$. Moreover, observing that 
   \[
   \Big|\widehat{s}_k + \frac{\rho_0 \widehat V_{\rm{eff}}(k)}{2|k|^2}\Big| \leq \frac{C 
\rho_0^2( \widehat{V}_{\rm{eff}}(k))^2}{|k|^4}
   \]
   for $|k| \gg \rho^{1/2}$, we can split 
 $\widehat Z(k)= \widehat Z_1(k)+ \widehat Z_2(k)$  with
	\be \label{def:Z1}
	\widehat Z_1(k) =  \rho_0 \widehat{\o}_{\ell_0}(k)- \frac{\rho_0 \widehat{V}_{\rm{eff}}(k) }{2|k|^2}
    \ee
    and  
    \be
    \qquad \big| \widehat  Z_2(k)\big| \leq  C \rho^2 \min \{ |k|^{-4}, |k|^{-6}\}
    \ee
for $|k| \gg \rho^{1/2}$, since $|\widehat{V}_\text{eff} (k)| \leq C \min \{ 1 , |k|^{-1}\}$. Writing 
\[ 2k^2 \widehat{\omega}_{\ell_0} (k) = 2 \int \min \Big\{ 1 , \frac{\frak{a}}{|x|} \Big\} \chi (x/\ell_0) \Delta e^{-ik\cdot x} dx \]
and integrating by parts (with Gauss theorem), we find 
\begin{equation}\label{eq:ome-V} 2k^2 \widehat{\omega}_{\ell_0} (k) - \widehat{V}_\text{eff} (k) = \widehat{z}_1 (k \ell_0) \end{equation} 
with $z_1 (x) = 4 \frak{a} \nabla \chi (x) \cdot x/|x|^3 - 2\frak{a} \Delta \chi (x) / |x|$. This implies that $\widehat{Z}_1 (k) = \rho_0 \widehat{z}_1 (k \ell_0) / 2k^2$ and therefore, introducing the notation 
\begin{equation}\label{eq:p-fourier} \| \widehat{h} \|_p^p = \frac{1}{L^3} \sum_{k \in \Lambda^*} |\widehat{h}_k|^p \end{equation} 
for $1 \leq p < \infty$ and $\| \widehat{h} \|_\infty = \sup_{k \in \Lambda^*} |\widehat{h}_k|$ for sequences $\widehat{h} = \{ \widehat{h}_k \}_{k \in\Lambda^*}$, that 
\[ \| \widehat{Z}_1 \|_1 \leq \frac{C \rho}{|\Lambda|} \sum_{\substack{k \in \Lambda^* \\ |k| \leq 
\ell_0^{-1}}} \frac{\| \widehat{z}_1 \|_\infty}{k^2} + C \rho \ell_0^{-3/2} \Big( \frac{1}{|\Lambda|} 
\sum_{\substack{k \in \Lambda^* \\ |k| \geq \ell_0^{-1}}} \frac{1}{|k|^4} \Big)^{1/2} \leq C \rho 
\ell_0^{-1} \]
because $\| \widehat{z}_1 (\ell_0 \cdot) \|_2 = \| z_1 (x/\ell_0) / \ell_0^3 \|_2 \leq C \ell_0^{-3/2}$. Since moreover 
\[ \| \widehat{Z}_2 \|_1 \leq \frac{C\rho}{|\Lambda|} \sum_{\substack{k \in \Lambda^* \\ |k| \leq K 
\rho^{1/2}}} \frac{1}{k^2} + \frac{C \rho^2}{|\Lambda|}  \sum_{\substack{k \in \Lambda^* \\ |k| > K \rho^{1/2}}} 
\frac{1}{|k|^4} \leq C \rho^{3/2} \]
we conclude that 
\[ \| Z \|_\infty \leq \| \widehat{Z} \|_1 \leq \| \widehat{Z}_1 \|_1 + \| \widehat{Z}_2 \|_1 \leq C \rho^{3/2}\,. \]
Similarly, we can bound 
\be \begin{split} \label{eq:nablaZ-io} 
\| \nabla Z\|_\infty &\leq \| |\cdot| \widehat Z \|_1 \leq \frac{C}{|\L|} \sum_{\substack{k \in \L^*\\ |k| \leq K\rho^{1/2}}} \frac{\rho}{|k|} +   \frac{C}{|\L|} \sum_{\substack{k \in \L^*\\ |k| \geq K\rho^{1/2}}}  \bigg[ \frac{\rho |\widehat z_1( k\ell_0 )|}{|k|}  + |k| |\widehat{Z}_2(k)|   \bigg] \\ &\leq C \rho^2 + C \rho \left[ \frac{1}{|\L|} \sum_{k\in \L^*} |\widehat{z}_1 (k \ell_0)|^2 |k|^2 \right]^{1/2} \left[ \frac{1}{|\L|} \sum_{|k| \geq K \rho^{1/2}}  \frac{1}{|k|^4} \right]^{1/2}  \\ &\hspace{.4cm} + \frac{C \rho^2}{|\Lambda|} \sum_{|k| \geq K \rho^{1/2}} \min \Big\{ \frac{1}{|k|^3} , \frac{1}{|k|^5} \Big\}  \\ &\leq C \rho^2 + C \rho^{3/4}  \ell_0^{-5/2}  + C \rho^2 |\log \rho | \leq C \rho^2 |\log \rho |\,.      
\end{split}\ee
From (\ref{eq:E1ZZ}), we conclude that $|\mathcal{E}_1| \leq C \rho^{5/2+\eps}$. 

To bound $\mathcal{E}_2$ we combine $|\widehat{V}_\mathrm{eff}(k)  - 8 \pi \aa| \leq C \min \{ k^2 ,1 \}$ with the estimate $|\widehat{s}_k| \leq C \rho / k^2$ for all $k \in \Lambda_+^*$, from Lemma \ref{lm:eta}. We obtain 
\[ \begin{split}
 | \mathcal{E}_2 | &\leq  \frac{C \rho}{|\L|}\sum_{\substack{k\in \L^*_+ \\ |k|\leq \rho^{1/2 -\eps}}} |\widehat{V}_\mathrm{eff}(k)  - 8 \pi \aa|  |\widehat s_k|^2 + \frac{C \rho}{|\L |}\sum_{\substack{k\in \L^*_+ \\ |k|\geq \rho^{1/2 -\eps}}}  |\widehat s_k|^2 \leq C \rho^{5/2+\eps} 
 %   \\
 % & \leq  \frac{C \rho_0}{|\L|}\sum_{\substack{k\in \L^*_+ \\ |k|\leq \rho^{1/2 -\eps}}} |k|^2 
 %   + \frac{C \rho}{|\L|}\sum_{\substack{k\in \L^*_+ \\ |k|\geq \rho^{1/2 -\eps}}}  \frac{(\rho \aa)^2}{|k|^4} \%leq C \rho^{5/2+\eps}
  \end{split}
  \]
  for all $\eps >0$ small enough. Similarly, with $|\widehat{(g-\mathbbm{1})} (k)| \leq |\widehat{s}_k| \leq \rho / k^2$ we also find $|\mathcal{E}_3| \leq C \rho^{5/2+\eps}$.

Finally, we consider $\mathcal{E}_4$ and $\mathcal{E}_5$. Recalling (\ref{eq:ome-V}) and Lemma \ref{lm:eta}, we find 
\[ |\mathcal{E}_4| = \Big| \frac{\rho_0}{|\L|} \sum_{k \in \Lambda_+^*} \widehat{z}_1 (k \ell_0) \widehat{s}_k \Big| \leq C \rho \| \widehat{z}_1 (\ell_0\cdot) \|_2 \| \widehat{s} \|_2\leq C \rho^{7/4} \ell_0^{-3/2} \leq C \rho^{5/2+3\eps/2} \]
since $\| \widehat{z}_1 (\ell_0 \cdot) \|_2 \leq C \ell_0^{-3/2}$, and 
\[ \begin{split}  |\mathcal{E}_5| &= \Big| \frac{\rho_0^2}{|\L|} \sum_{k \in \L_+^*} \frac{1}{4k^2} \widehat{z}_1 (k\ell_0) \big(2k^2 \widehat{\omega}_{\ell_0} (k) + \widehat{V}_\text{eff} (k) \big) \Big| \\ &\leq C \rho^2 \| \widehat{z}_1 \|_\infty \frac{1}{|\L|} \sum_{|k| \leq \rho^{1/2+\eps}} \frac{1}{k^2} + C \rho^2 \| \widehat{z}_1 ( \ell_0\cdot) \| \Big[ \frac{1}{|\L|} \sum_{|k| \geq \rho^{1/2+\eps}} \frac{1}{|k|^4} \Big]^{1/2} \\ &\leq C \rho^{5/2+\eps} + C \rho^{7/4-\eps/2}  \ell_0^{-3/2}  \leq C \rho^{5/2+\eps}\,. \end{split}  \]    
\end{proof}

%Theorem \ref{thm:main} follows from Proposition \ref{prop:Psi-energy}, together with Lemmas \ref{lm:E-rho}, \ref{lm:tilde-E-rho}, and Proposition \ref{prop:N_on_psi} and equivalence of ensembles. 

 Proposition \ref{prop:Psi-energy} together with Lemmas \ref{lm:E-rho} and \ref{lm:tilde-E-rho} yield the upper bound
 \be \label{eq:UB-fin}
 L^{-3} \langle \Psi, \cK \Psi \rangle \leq 4 \pi  \aa \rho^2 \Big[ 1 + \frac{128}{15\sqrt{\pi}} (\rho \frak{a}^3)^{1/2}\Big] + C \rho^{5/2+\delta} 
 \ee
if $0 < \delta < \eps/2$ and $\eps > 0$ is small enough. At the same time, we recall from (\ref{eq:rho0}) that $\langle \Psi, \cN \Psi \rangle \geq \rho L^3$. To conclude the proof of Theorem \ref{thm:main} we use equivalence of ensembles, arguing similarly as in \cite[Appendix A]{BCS}. Starting with the periodic trial state $\Psi$ we construct, following \cite[Lemma A.1]{BCS}, a new normalized Fock space vector $\Psi_D$ on the box with size $L+2\cL$ for a fixed $\cL > 0$,  satisfying Dirichlet boundary conditions and such that 
\[ \langle \Psi_D, \cN \Psi_D \rangle = \langle \Psi , \cN \Psi \rangle \geq \rho L^3 \]
and 
\begin{equation}\label{eq:LHY-D} \langle \Psi_D, \cK \Psi_D \rangle \leq \langle \Psi, \cK \Psi \rangle + \frac{C}{L \cL} \langle \Psi, \cN \Psi \rangle  \,.\end{equation}
This implies 
\[ \frac{1}{(L+2\cL)^3} \langle \Psi_D, \cK \Psi_D \rangle \leq  4 \pi  \aa \rho^2 \Big[ 1 + \frac{128}{15\sqrt{\pi}} (\rho \frak{a}^3)^{1/2}\Big] + C \rho^{5/2+\delta}  \]
if $L > 0$ is large enough (the new error term can be absorbed in the contribution proportional to $\rho^{5/2+\delta}$). On the other hand, setting  \[ e^D_L (\rho) = \frac{E^D (\rho L^3 , L)}{L^3} \]
with $E^D (N,L)$ the ground state energy of $N$ particles in a box of size $L$, with Dirichlet boundary conditions, we find 
\[ \begin{split} \frac{1}{(L+2\cL)^3} \langle \Psi_D, \cK \Psi_D \rangle &= \frac{\mu}{(L+2\cL)^3} \langle \Psi_D, \cN \Psi_D \rangle + \frac{1}{(L+2\cL)^3} \langle \Psi_D, (\cH - \mu \cN) \Psi_D \rangle \\ &\geq \frac{\mu \rho}{(1+2\cL/L)^3} + \sum_{m \geq 1} \left[ e^D_{L+2\cL} \Big( \frac{m}{(L+2\cL)^3} \Big) - \mu \frac{m}{(L+2\cL)^3} \right] \| \Psi_D^{(m)} \|^2   \,. \end{split} \]
As explained in \cite[Eq. (A.16)]{BCS}, we can bound 
\[ e^D_{L+2\cL} (\rho) \geq (1 + \frak{a}/L)^3 \, e \big(\rho \, (1+ \frak{a}/\cL)^{-3}\big) \]
where we recall that $e (\rho)$ is the ground state energy per unit volume, in the thermodynamics limit $N, L \to \infty$, at density $\rho$ (the limiting energy density is independent of the boundary conditions). Thus, we obtain 
\[ \begin{split} 
&\frac{1}{(L+2\cL)^3} \langle \Psi_D, \cK \Psi_D \rangle \\ &\geq \frac{\mu \rho}{(1+2\cL/L)^3} \\ &\hspace{.3cm} + (1+ \frak{a}/L)^3 \sum_{m = 1} ^{CL^3} \left[ e \big(\frac{m}{(L+ 2\cL)^3 (1+ \frak{a}/L)^3} \big) - \mu \frac{m}{(L+2\cL)^3 (1+ \frak{a}/L)^3} \right] \| \Psi_D^{(m)} \|^2 \\ &\geq \mu \rho (1+ 2\cL/L)^{-3} - (1+ \frak{a}/L)^3 e^* (\mu) \end{split} \]
with $e^*$ the Legendre transform of $e$. From (\ref{eq:LHY-D}) and letting $L \to \infty$, we arrive at 
\[ 4 \pi  \aa \rho^2 \Big[ 1 + \frac{128}{15\sqrt{\pi}} (\rho \frak{a}^3)^{1/2}\Big] + C \rho^{5/2+\delta}  \geq \mu \rho - e^* (\mu) \,. \]
Since this is true for all $\mu > 0$, and since the limiting energy density is known to be convex, we obtain
\[ e (\rho) \leq 4 \pi  \aa \rho^2 \Big[ 1 + \frac{128}{15\sqrt{\pi}} (\rho \frak{a}^3)^{1/2}\Big] + C \rho^{5/2+\delta} \]
which concludes the proof of Theorem \ref{thm:main}.

\section{A-priori bounds on trial state} 
\label{sec:apri}

In this section, we prove estimates on the local number of particles and of excitations in the trial state $\Psi$ defined in (\ref{eq:trial}). This will, first of all, allows us to show Prop. \ref{prop:N_on_psi}. Moreover, it will give us a-priori control on $\Psi$, which will be crucial in the next section, to estimate the expectation of the kinetic energy and thus to prove Prop. \ref{prop:Psi-energy}. 

We recall the hierarchy of length scales: $\ell = C\rho^{-\delta}$, $\ell_0 =C \rho^{-1/2-\eps}$. For $w \in \Lambda$ and $r > 0$, we denote by   
\begin{equation}\label{eq:cNRw} \cN_{r} (w) = \int_{|y-w| \leq r} dy \, a_y^* a_y \end{equation} 
the operator measuring the number of particles in a ball of radius $r$ around $w$.  
%(notice that $\cN_{\ell_0} (w) = d\Gamma (\chi_{w,\ell_0})$ is just the second quantization of the %characteristic function $\chi_{w,\ell_0}$ of the ball of radius $\ell_0$, centered at $w$).  
Moreover, introducing the distributions 
%It is also convenient to introduce the operator valued distributions 
\begin{equation}\label{eq:bb} \begin{split} b^*_x &= a^*_x - \sqrt{\rho_0} = W (\rho_0) a^*_x W(\rho_0)^* \\  b_x &= a_x - \sqrt{\rho_0} = W(\rho_0) a_x W(\rho_0)^* \end{split} \end{equation} 
creating and annihilating excitations of the Bose-Einstein condensate at a point $x \in \Lambda$, we define the operator 
% Similarly to (\ref{eq:cNRw}), for $w \in \Lambda$, $R > 0$, we introduce the operator
\begin{equation}\label{eq:tcNRw} \wt{\cN}_r (w) = \int_{|x-w| \leq r} dx \, b_x^* b_x  \, .\end{equation} 

In the next lemma, we estimate moments of the operators (\ref{eq:cNRw}), (\ref{eq:tcNRw}), in the state $\Psi$, for $r$ slightly larger than the length scale $\ell_0$ (this restriction is needed, because we will use the fast decay (\ref{eq:decay}) of the kernels $\sigma, \gamma^{-1} * \sigma$, at distances larger than $\ell_0$).

\begin{lemma} \label{lm:Nm} 
Let $R = C \rho^{-1/2-3\eps}$. For every $m \in \bN$, there exists $C = C (m) > 0$ such that 
\begin{equation}\label{eq:Nm} \langle \Psi, \cN^m_{R} (w) \Psi \rangle \leq C  \rho^{-m/2-12 m\eps}   \end{equation} 
and 
\begin{equation}\label{eq:Nmb} \langle \Psi, \wt{\cN}^m_R (w) \Psi \rangle \leq C \rho^{-10 m \eps} \end{equation}
for every $w \in \Lambda$, if $\rho > 0$ and $\eps, \delta > 0$ in the definition of the trial state $\Psi$ are small enough (how small $\rho > 0$ needs to be depends on $m$). 
\end{lemma} 

{\bf Remark.} From the definition (\ref{eq:trial}), it is clear that $\Psi$ is translation invariant. Hence, $\langle \Psi, \cN_R^m (w) \Psi \rangle$ and $\langle \Psi, \widetilde{\cN}_R^m (w) \Psi \rangle$ are actually independent of $w \in \Lambda$. 

\medskip

{\bf Remark.} The lemma implies that the expectation of the total number of particles $\cN$ and the total number of excitations $\wt{\cN}$ are bounded by 
\begin{equation}\label{eq:rmkN} \langle \Psi, \cN \Psi \rangle = \int dx \langle \Psi, a_x^* a_x \Psi \rangle \leq C R^{-3} \int dz \langle \Psi, \cN_R (z) \Psi \rangle \leq C\rho^{1-3\eps} L^3 \end{equation} 
and, respectively, 
\begin{equation}\label{eq:rmkwtN} \langle \Psi, \wt{\cN} \Psi \rangle = \int dx \langle \Psi, b_x^* b_x \Psi \rangle \leq C  \rho^{3/2-\eps} L^3\,. \end{equation} 
Later, we will prove a more precise version of (\ref{eq:rmkN}).

\medskip

In order to show Lemma \ref{lm:Nm}, we use the following result, which allows us to estimate moments of the number of particles in balls of radius $R_2$ by the corresponding moments in slightly smaller balls of radius $R_1$.  
\begin{lemma} \label{lm:R'R}
Let $m \in \bN$. Then there exists $C = C (m) > 0$ such that 
\begin{equation}\label{eq:R'R} \langle \Phi, \cN_{R_2}^m (w) \Phi \rangle \leq \left( 1 + C (R_2 - R_1)/R_1 \right) \langle \Phi, \cN_{R_1}^m (w) \Phi \rangle \end{equation}  
for every translation invariant $\Phi \in \cF$, all $0< R_1 \leq R_2 \leq 2 R_1$ and all $w \in \Lambda$. 
\end{lemma}

{\bf Remark.} Since \[ W (\rho_0) \cN_R (w) W(\rho_0)^* = \wt{\cN}_R (w) \, ,  \] \eqref{eq:R'R}  immediately implies the analogous bound with the local number of particles $\cN_R (w)$ replaced by the local number of excitations $\wt{\cN}_R (w)$. In other words, for every $m \in \bN$, there exists a constant $C > 0$ such that 
\begin{equation}\label{eq:R'R-b} \langle \Phi, \wt{\cN}^m_{R_2} (w) \Phi \rangle \leq \left( 1 + C (R_2 - R_1)/R_1 \right) \langle \Phi, \wt{\cN}^m_{R_1} (w) \Phi \rangle \end{equation} 
for every $0 < R_1 \leq R_2 \leq 2 R_1$ and for every translation invariant $\Phi \in \cF (\Lambda_L)$.

\begin{proof} 
For $m \in \bN$ and $\mu$ a non-negative $\sigma$-finite Borel measure on $\Lambda$, we set 
\[ F_\mu (R) := \int_\Lambda dy \; \left( \int_{|x-y| \leq R} d\mu (x) \right)^m \, . \]
We claim that 
\begin{equation} \label{eq:F2-F1} F_\mu (R_2) \leq \big(1 + C (R_2 - R_1)/R_1\big) F_\mu (R_1) \end{equation}
for all $R_1 \leq R_2 \leq 2 R_1$ and for a constant $C > 0$ depending only on $m$. 

From (\ref{eq:F2-F1}), the bound  (\ref{eq:R'R}) follows by the observing that, on the $N$-particle sector of Fock space $\cF$, $\cN_R (w)$ acts as multiplication with 
\[ \sum_{j=1}^N \chi (|x_j - w| \leq R) = \int_{|x-w| \leq R} d\mu (x)  \]
where we defined the measure $\mu (x) = \sum_{j=1}^N \delta (x_j - x)$. Therefore, working on sectors with fixed number of particles, (\ref{eq:F2-F1}) implies that 
\[ \int_\Lambda dw \, \langle \Phi, \cN^m_{R_2} (w) \Phi \rangle \leq \big(1 + C (R_2 - R_1)/R_1 \big) \int_\Lambda dw \, \langle \Phi, \cN^m_{R_1} (w) \Phi \rangle \]
which immediately yields (\ref{eq:R'R}), if we assume $\Phi$ to be translation invariant. 

To show (\ref{eq:F2-F1}), we use Fubini to write 
\[ \begin{split} F_\mu (R) &= \int d\mu (x_1) \dots \int d\mu (x_m) \int_\Lambda dy \,  \prod_{j=1}^m \chi (|x_j - y| \leq R) \\ &= \int d\mu (x_1) \dots \int d\mu (x_m) \, K_R (x_1, \dots , x_m) \end{split} \]
where
\[ K_R (x_1, \dots , x_m) = \text{Vol } \Big( B_R (x_1) \cap B_R (x_2) \cap \dots \cap B_R (x_m) \Big) \]
with $\text{Vol} (A)$ the Lebesgue measure of $A \subset \Lambda$. We observe that $K_R (x_1, \dots , x_m) =0$, unless $|x_i - x_j| < 2R$, for all $i,j \in \{ 1, \dots , m\}$. 
Moreover, for $R_1 \leq R_2 \leq 2R_2$, we find $K_{R_2} (x_1, \dots , x_m) \geq K_{R_1} (x_1, \dots , x_m)$ by monotonicity and  
\begin{equation}\label{eq:contiK} 
K_{R_2} (x_1, \dots , x_m) \leq K_{R_1} (x_1, \dots , x_m) + C R_1^2 (R_2 - R_1) \prod_{j=2}^m \chi (|x_1 - x_j| \leq 2R_2) \end{equation} 
for a constant $C > 0$, depending only on $m$. To prove (\ref{eq:contiK}), it is enough to remark that, when we increase the radius of the $m$ balls from $R_1$ to $R_2$, the Lebesgue measure of their intersection can grow at most by 
\begin{equation}\label{eq:R2-R1} \frac{4}{3} \pi m (R_2^3 - R_1^3) \leq C (R_2 - R_1) R_1^2 \end{equation} 
if $R_1 \leq R_2 \leq 2 R_1$. We conclude that 
\[ \begin{split} F_\mu (R_2) &= \int d\mu (x_1) \dots \int d\mu (x_m) \, K_{R_2} (x_1, \dots , x_m) \\ &\leq  F_\mu (R_1) + C (R_2 - R_1) R_1^2 \int d\mu (x_1) \int_{B_{4R_1} (x_1)} d\mu (x_2) \dots \int_{B_{4R_1} (x_1)} d\mu (x_m) \end{split} \]
for all $R_1 \leq R_2 \leq 2 R_1$. To bound the remaining integral, we find $k \in \bN$ points $\tilde{y}_1, \dots , \tilde{y}_k \in B_4 (0)$, with 
\[ B_4 (0) \subset \bigcup_{j=1}^k B_{1/2} (\tilde{y}_j) \,. \]
Setting $y_j = R_1 \tilde{y}_j$, for $j=1,\dots , k$, this implies that 
\[ B_{4R_1} (x_1) \subset \bigcup_{j=1}^k B_{R_1 /2} (x_1 + y_j) \]
and therefore that 
\[ \begin{split} &\int d\mu (x_1)  \int_{B_{4R_1} (x_1)} d\mu (x_2) \dots \int_{B_{4R_1} (x_1)} d\mu (x_m) \\ &\leq \sum_{j_2, \dots j_m = 1}^k \int d\mu (x_1) \int_{B_{\frac{R_1}{2}} (x_1 + y_{j_2})} d\mu (x_2) \dots \int_{B_{\frac{R_1}{2}} (x_1 + y_{j_m})} d\mu (x_m)\\
&\leq \sum_{j_2, \dots j_m = 1}^k \frac{1}{|B_{\frac{R_1}{2}} (0)|} \int dw \int_{B_{\frac{R_1}{2}} (w)} d\mu (x_1)  \int_{B_{\frac{R_1}{2}} (x_1 + y_{j_2})} d\mu (x_2) \dots \int_{B_{\frac{R_1}{2}} (x_1 + y_{j_m})} d\mu (x_m) \, ,
\end{split} \]
where in the last step we inserted a factor $\int dw \, \chi (|x_1 - w| \leq R_1/2) / |B_{R_1/2}(0)| = 1$ and then we exchanged the $w$ and the $x_1$ integrals. We obtain 
\[\begin{split} &\int d\mu (x_1)  \int_{B_{4R_1} (x_1)} d\mu (x_2) \dots \int_{B_{4R_1} (x_1)} d\mu (x_m) 
\\ &\leq \sum_{j_2, \dots j_m = 1}^k \frac{1}{|B_{\frac{R_1}{2}} (0)|} \int dw \int_{B_{R_1} (w)} d\mu (x_1)  \int_{B_{R_1} (w + y_{j_2})} d\mu (x_2) \dots \int_{B_{R_1} (w + y_{j_m})} d\mu (x_m)  \,.
\end{split} \]
With the elementary inequality $a_1 \dots a_m \lesssim C_m (a_1^m + \dots + a_m^m)$, we conclude that 
\[ \int d\mu (x_1)  \int_{B_{4R_1} (x_1)} d\mu (x_2) \dots \int_{B_{4R_1} (x_1)} d\mu (x_m) \leq C R_1^{-3} F_\mu (R_1) \]
for $C > 0$ depending only on $m$. Together with (\ref{eq:R2-R1}), this proves (\ref{eq:F2-F1}). 
\end{proof}

With Lemma \ref{lm:R'R}, we are now ready to show Lemma \ref{lm:Nm}. 

\begin{proof}[Proof of Lemma \ref{lm:Nm}]
We first show (\ref{eq:Nm}), proceeding by induction. The case $m = 0$ is trivial. We assume that (\ref{eq:Nm}) holds true for an exponent $m \in \bN$ and we estimate 
\begin{equation}\label{eq:Nm1} \begin{split} \langle \Psi, \cN_{R}^{m+1} (w) \Psi \rangle &= \int_{|x-w| \leq R} dx \, \langle \Psi, \cN_{R}^m (w) a_x^* a_x \Psi \rangle \\ &=   \int_{|x-w| \leq R} dx \, \langle a_x \Psi, (\cN_{R} (w) + 1)^m   a_x \Psi \rangle \\ &\leq C  \rho^{-m/2-12 m\eps} + \int_{|x-w| \leq R} dx \, \langle a_x \Psi, \cN_{R}^m (w)  a_x \Psi \rangle\,.   \end{split} \end{equation} 
With (\ref{eq:id}), applying Cauchy-Schwarz, recalling that $0 \leq J(x) \leq 1$ and noticing that $J(x)$ and $\cN_{R} (w)$ commute (they are both multiplication operators, on each sector with fixed number of particles), we find 
\begin{equation}\label{eq:Nm2} \begin{split} 
\langle &\Psi, \cN_{R}^{m+1} (w) \Psi \rangle\\ &\leq C  \rho^{-m/2-12 m\eps} + C (\rho R^3) \rho^{-3\eps/2}  \langle \Psi, \cN_{R}^m (w) \Psi \rangle \\ &\hspace{.3cm} + (1 + C \rho^{3\eps/2}) \int_{|x-w| \leq R} dx \int dy dy' \nu (x-y) \nu (x-y') \langle \Psi, J(y) a_y \cN^m_{R} (w) a_{y'}^* J(y') \Psi \rangle \,, \end{split} \end{equation} 
where we introduced the notation $\nu  = \gamma^{-1} * \sigma$. Applying the induction assumption in the second term and commuting $a_y, a_{y'}^*$ in the last term, we obtain 
\begin{equation}\label{eq:Nm3}  \begin{split} 
\langle &\Psi, \cN_{R}^{m+1} (w) \Psi \rangle \\&\leq C \rho^{-(m+1)/2 - 12 (m+21/24) \eps}  \\ &\hspace{.4cm} + (1 + C \rho^{3\eps/2})  \int_{|x-w| \leq R} dx dy \, \nu (x-y)^2 \langle J(y) \Psi, (\cN_{R}  (w)+ 1)^m J(y) \Psi \rangle  \\ &\hspace{.4cm} + (1 + C \rho^{3\eps/2})  \int_{|x-w| \leq R} dx \int dy dy' f_\ell (y-y')^2 \nu (x-y) \nu (x-y') \\ &\hspace{.8cm} \times  \langle J (y') a_{y'} \Psi, (\cN_{R} (w) + \chi (|y-w| \leq R) + \chi (|y'-w| \leq R))^m  J(y) a_y \Psi \rangle \,.\end{split} \end{equation} 
In the second term on the r.h.s. we use Lemma \ref{lm:eta} to estimate $\| \nu \|_2^2 \lesssim \rho^{3/2}$. In the third term, we use the fast decay of the kernels $\nu (z)$ for $|z| \gg \ell_0 \rho^{-\eps/2} = C\rho^{-1/2-3\eps/2}$, as established in Lemma \ref{lm:eta}, to  restrict the integrals to $|y-x| \leq \ell_1$, $|y'-x| \leq \ell_1$, where $\ell_1 = \ell_0 \rho^{-3\eps/4} = C\rho^{-1/2-7\eps/4}$ (with respect to the scale $\rho^{-1/2-3\eps/2}$ from (\ref{eq:decay}) in Lemma \ref{lm:eta}, we add an additional factor $\rho^{-\eps/4}$ to gain some smallness). Error terms can be bounded using the Cauchy-Schwarz inequality. For example, to estimate the contribution 
\begin{equation}\label{eq:ex-decay} \begin{split} 
\text{A} := &\, \int_{|x-w| \leq R} dx \int_{\substack{|y-x| \leq \ell_1 \\ |y' - x| \geq \ell_1}} dy dy' \,  f_\ell^2 (y-y') \nu (x-y) \nu (x-y') \\ &\hspace{.4cm} \times  \langle J (y') a_{y'} \Psi,  (\cN_{R} (w) + \chi (|y-w| \leq R) + \chi (|y'-w| \leq R))^m   J(y) a_y \Psi \rangle \end{split} \end{equation}
we can bound the $y'$-integral by a sum of integrals on balls of radius $\ell_1/2$, centered around lattice points $(\ell_1 /2) n$, with $n = (n_1, n_2, n_3) \in \bZ^3 \cap \Lambda$ and $|n| = |n_1| + |n_2| + |n_3| \geq 2$. With (\ref{eq:decay}) we obtain, for arbitrary  $k \in \bN$, 
\[ \begin{split} 
| \text{A} | \leq &\, 
C \rho^{1+ \eps k /4 + \delta}  \sum_{n \in \bZ^3 :  \ell_1 n /2 \in \Lambda, |n| \geq 2}  \int_{|x-w| \leq R} dx \int_{|y-x| \leq \ell_1}  dy dy' \, |\nu (x-y)|  \\ &\hspace{.5cm}  \times \frac{\chi (|x-y'- \ell_1 n/2| \leq \ell_1/2)}{|n|^k} \| a_{y'} (\cN_{R} (w) + 1)^{m/2} \Psi \|  \| a_{y} (\cN_{R}  (w) +1)^{m/2}\Psi \|
\end{split} \]
with a constant $C > 0$ depending on $k$. By Cauchy-Schwarz, we have 
 \[ \begin{split} 
| \text{A} |  &\leq C \rho^{-1/2+ \eps (k/4-12)+\delta}   \langle \Psi, (\cN_{R} (w) + 1)^m  \cN_{R + \ell_1} (w) \Psi \rangle^{1/2}  \\ &\hspace{.5cm} \times \sum_{n \in \bZ^3 :  |n| \geq 2}  |n|^{-k}  \langle \Psi, (\cN_{R} (w) + 1)^m  \cN_{R + \ell_1/2} (w - \ell_1 n/2) \Psi \rangle^{1/2}\,. \end{split} \]
From the translation invariance of $\Psi$ and from Lemma \ref{lm:R'R} (since $R+ \ell_1 \leq 2 R$), we find (fixing $k > 3$ to make sure that $\sum_{n \in \bZ^3} |n|^{-k} < \infty$) a constant $C > 0$ depending on $k, m$ such that  
\[ | \text{A} |  \leq C \rho^{-1/2+ \eps (k/4-12) + \delta} \langle \Psi,  ( \cN^{m+1}_{R} (w) + 1)  \Psi \rangle\,. \]
Choosing $k \in \bN$ sufficiently large, we conclude that, for any $\kappa > 0$, 
\begin{equation}\label{eq:tails} \begin{split} 
| \text{A} |   &
\leq \kappa \rho^{\eps} \langle \Psi, (\cN_{R}^{m+1} (w) + 1) \Psi \rangle 
\end{split} \end{equation} 
if $\rho > 0$ is small enough (depending on $m$).  Hence, we obtain 
\begin{equation}\label{eq:Nm4} \begin{split} 
\langle &\Psi, \cN_{R}^{m+1} (w) \Psi \rangle \\ &\leq C \rho^{-(m+1)/2-12 (m+21/24) \eps} + \kappa \rho^\eps \langle \Psi, \cN_{R}^{m+1} (w) \Psi \rangle \\ &\hspace{.4cm} + (1+C \rho^{3\eps/2}) \int_{|x-w| \leq R} dx \int_{|y-x| \leq \ell_1, |y'-x| \leq \ell_1} dy dy' f_\ell (y-y')^2 \nu (x-y) \nu (x-y') \\ &\hspace{2cm} \times  \langle J (y') a_{y'} \Psi, (\cN_{R} (w) + \chi (|y-w| \leq R) + \chi (|y'-w| \leq R))^m  J(y) a_y \Psi \rangle \,.\end{split} \end{equation} 
In the last term, we expand $(\cN_{R} (w) + \chi (|y-w| \leq R) + \chi (|y'-w| \leq R))^m$; except for $\cN^m_{R} (w)$, all other contributions can be bounded using the induction assumption (after applying the Cauchy-Schwarz inequality). For example,
\[ \begin{split} &\int_{|x-w| \leq R} dx \int_{\substack{|y-x| \leq \ell_1 \\ |y'-x| \leq \ell_1}} dy dy' |\nu (x-y)| |\nu (x-y')|  \| \cN_{R}^{(m-1)/2} (w) a_{y'} \Psi \| \| \cN_{R}^{(m-1)/2} (w) a_y \Psi \| \\ &\leq \int_{|x-w| \leq R} dx \int_{\substack{|y-x| \leq \ell_1 \\ |y'-x| \leq \ell_1}} dy dy' |\nu (x-y)|^2  \| \cN_{R}^{(m-1)/2} (w) a_{y'} \Psi \|^2 \\ 
&\leq C \rho^{-9\eps}  \langle \Psi, \cN^{m-1}_{R} (w) \cN_{R+ \ell_1} (w) \Psi \rangle \leq C \rho^{-9\eps} \langle \Psi, \cN^m_{R+ \ell_1} (w) \Psi \rangle \\ &\leq C \rho^{-9\eps} \langle \Psi, \cN^m_{R} (w) \Psi \rangle \leq C  \rho^{- m/2 - 12 (m+3/4) \eps} \, ,  \end{split} \]
where in the last line we applied Lemma \ref{lm:R'R} and the induction assumption. 

We arrive at 
\begin{equation} \label{eq:defT} 
\langle \Psi, \cN_{R}^{m+1} (w) \Psi \rangle  \leq C \rho^{-(m+1)/2 - 12 (m+21/24)\eps} + \kappa \rho^\eps \langle \Psi, \cN_{R}^{m+1} (w) \Psi \rangle  +  \text{T} \end{equation} 
with 
\[ \begin{split}  \text{T}  &=  (1+ C \rho^{3\eps/2} ) \int_{|x-w| \leq R} dx \int_{|x-y| \leq \ell_1, |x-y'| \leq \ell_1}  dy dy' f^2_\ell (y-y') \nu (x-y) \nu (x-y')  \\ &\hspace{8cm} \times \langle J (y') a_{y'} \Psi, \cN^m_{R} (w)  J(y) a_y \Psi \rangle \,.
\end{split} \]
We decompose $\text{T} = \text{S}_1 + \text{S}_2$, with 
\[ \begin{split} \text{S}_1 &=  (1+ C \rho^{3\eps/2})  \int_{|x-w| \leq R}dx \int_{|x-y||, |x-y'|  \leq \ell_1} dy dy'   (f_\ell^2 (y-y')-1)  \nu (x-y ) \nu (x-y') \\ &\hspace{6cm} \times \langle J(y') a_{y'} \Psi, \cN^m_{R} (w) J(y) a_y \Psi \rangle \, ,   \\
\text{S}_2 &= (1+ C \rho^{3\eps/2}) \int_{|x-w| \leq R} dx \int_{\substack{|x-y|  \leq \ell_1 \\ |x-y'| \leq \ell_1}} dy dy'   \nu (x-y ) \nu (x-y') \\ &\hspace{6cm} \times  \langle J(y') a_{y'} \Psi, \cN_{R}^m (w) J(y) a_y \Psi \rangle\,. \end{split} \]
To control $\text{S}_1$, we observe that, with Lemma \ref{lm:fell}, 
\be  \label{eq:uell}
1 - f_\ell^2 (y-y') \leq 2 \o_\ell (y-y') \leq \frac{C \chi (|y-y'| \leq \ell)}{|y-y'|} \,.
\ee
Therefore, recalling that $\ell = C\rho^{-\delta}$,  
\[ \begin{split} 
|\text{S}_1| \leq \; &C \int_{|x-w| \leq R}dx  \int_{|w-y|,|w-y'|  \leq R + \ell_1} dydy' \, \frac{\chi (|y-y'| \leq \ell)}{|y-y'|} |\nu (x-y')|^2 \| a_y \cN^{m/2}_{R} (w) \Psi \|^2 \\  \leq \; &C \rho^{3/2-2\delta} \langle \Psi, \cN^m_{R} (w)  \cN_{R + \ell_1} (w) \Psi \rangle \,.\end{split} \]
With Lemma \ref{lm:R'R} and using the translation invariance of $\Psi$, we obtain  
\begin{equation}\label{eq:S1bd}  \begin{split} 
|\text{S}_1| \leq \; & C \rho^{3/2-2\delta} \langle \Psi, \cN^{m+1}_{R} (w) \Psi \rangle \leq \kappa \rho^{\eps}  \langle \Psi, \cN^{m+1}_{R} (w) \Psi \rangle  \end{split} \end{equation} 
for any $\kappa > 0$, if $\rho > 0$ is small enough (and $\delta, \eps > 0$ are sufficiently small). 
 
 As for $\text{S}_2$, we write 
 \[  \begin{split} S_2 = &(1+ C \rho^{3\eps/2})  \int_{|x-w| \leq R}dx  \int_{\substack{|x-y|  \leq \ell_1 \\ |x-y'| \leq \ell_1}} dy dy'  \chi (|w-y| \leq R + \ell_1) \chi (|w-y'| \leq R + \ell_1)  \\ &\hspace{5cm} \times  \nu (x-y ) \nu (x-y')  \langle J(y') a_{y'} \Psi, \cN_{R}^m (w) J(y) a_y \Psi \rangle \,. \end{split} \]
We can now remove the constraints $|x-y| , |x-y'| \leq \ell_1$. The contributions arising from $|x-y| > \ell_1$ or $|x-y'| > \ell_1$ can be handled as we did above with the term $\text{A}$. We find 
 \begin{equation} \label{eq:SwtS} 0 \leq \text{S}_2 \leq \kappa \rho^\eps \langle \Psi, \cN_{R}^{m+1} (w) \Psi \rangle + \wt{\text{S}}_2 \end{equation} 
for any $\kappa > 0$, if $\rho$ is small enough, with 
  \begin{equation}\label{eq:wtS2} \wt{\text{S}}_2 =(1+ C \rho^{3\eps/2})  \int_{|x-w| \leq R} dx \int dy dy' \nu (x-y) \nu (x-y') \langle \Xi_y, \cN_{R}^{m} (w) \Xi_{y'} \rangle \end{equation} 
 where we defined $\Xi_y = \chi (|w-y| \leq R+ \ell_1) J(y) a_y \Psi$. Since 
\[  \int dy dy'   \nu (x-y ) \nu (x-y') \langle \Xi_y , \cN^m_{R} (w) \Xi_{y'} \rangle \geq 0 \] 
for every $x \in \Lambda$, we can estimate 
 \[ \begin{split}  0 \leq \wt{\text{S}}_2 &\leq (1+ C \rho^{3\eps/2})  \int dy dy' (\nu * \nu) (y-y')  \, \langle \Xi_y, \cN_{R}^{m} (w) \Xi_{y'} \rangle\,. \end{split} \]
 Switching to Fourier space, we obtain 
\[ \begin{split} 
0 \leq \wt{\text{S}}_2 \leq  \frac{(1+ C \rho^{3\eps/2})}{|\Lambda|} \sum_{k \in \Lambda^*} \widehat{\nu}_k^2  \, \langle \widehat{\Xi}_k , \cN^m_{R} (w)  \widehat{\Xi}_k \rangle \,.\end{split} \]
From (\ref{Th:est_3}), we find \[ \| \widehat{\gamma} \|_\infty \lesssim 1 + \| \widehat{\sigma} \|_\infty \lesssim 1+ \| \sigma \|_1 \lesssim \rho^{-\eps/2} \]
which implies that $\widehat{\nu}_k^2 =  \widehat{\sigma}_k^2  / \widehat{\gamma}_k^2 = 1 - 1/\widehat{\gamma}_k^{2} \leq 1 - C \rho^{\eps}$, for every momentum $k \in \Lambda^*$. Hence, adjusting slightly the constant $C > 0$ to absorb the error proportional to $\rho^{3\eps/2}$,  
\[ 
0 \leq \wt{\text{S}}_2 \leq \frac{\big( 1 - C \rho^{\eps} \big)}{|\Lambda|} \sum_{k \in \Lambda^*} \langle \widehat{\Xi}_k , \cN^m_{R} (w)  \widehat{\Xi}_k \rangle\,. \]
Switching back to position space, we obtain 
\[ \begin{split} 0 \leq \wt{\text{S}}_2 &\leq  \big( 1 - C \rho^{\eps} \big) \int_{|w-y| \leq R + \ell_1} dy  \langle a_y \Psi, \cN^m_{R} (w) a_y \Psi \rangle \end{split} \]
and therefore, with Lemma \ref{lm:R'R}, 
\[ \begin{split} 
0 \leq \wt{\text{S}}_2 &\leq  (1 - C \rho^{\eps})  \langle  \Psi, \cN^m_{R} (w) \cN_{R + \ell_1} (w) \Psi \rangle \\ &\leq (1 - C \rho^{\eps}) (1 + C \ell_1/R) \langle \Psi, \cN^{m+1}_{R} (w) \Psi \rangle \leq 
(1 - C \rho^{\eps}) \langle \Psi, \cN^{m+1}_{R} (w)  \Psi \rangle \,.\end{split} \]
Here we used the observation that $\ell_1/R = C \rho^{5\eps /4} \ll \rho^\eps$. 
From (\ref{eq:SwtS}), choosing $\kappa > 0$ small enough, we find  
\[ \text{S}_2 \leq (1-C \rho^\eps )  \langle \Psi, \cN^{m+1}_{R} (w) \Psi \rangle \] for an appropriate constant $C > 0$. Combining this bound with (\ref{eq:S1bd}) and (\ref{eq:defT}), we conclude that 
\[ \langle \Psi, \cN_{R}^{m+1} (w) \Psi \rangle \leq C \rho^{-(m+1)/2 -12 (m+21/24)\eps}  + (1 - C \rho^\eps )  \langle \Psi, \cN_{R}^{m+1} (w) \Psi \rangle \]
which implies that 
\[  \langle \Psi, \cN_{R}^{m+1} (w) \Psi \rangle \leq C  \rho^{-(m+1)/2-12 (m+1) \eps}  \]
for a constant $C > 0$ depending only on $m \in \bN$. 

Next, we prove (\ref{eq:Nmb}). Since the operators $b_x, b_x^*$ satisfy the same canonical commutation relations as the operators $a_x, a_x^*$ and since we have the same bound (\ref{eq:R'R-b}) for $\wt{\cN}_R (w)$, as we did for $\cN_R (w)$, we can proceed similarly as in the proof of (\ref{eq:Nm}). In the following, we only highlight the differences. 

As above, we proceed by induction over $m \in \bN$. Since the claim is trivial for $m=0$, we can assume (\ref{eq:Nmb}) to hold true for some $m \in \bN$. Similarly to (\ref{eq:Nm1}), we compute 
\[ \langle \Psi, \wt{\cN}^{m+1}_R (w) \Psi \rangle  \leq C \rho^{-10 m \eps} + \int_{|x-w| \leq R} dx \, \langle b_x \Psi, \wt{\cN}_R^m (w) b_x \Psi \rangle\,. \]
Instead of (\ref{eq:id}), we apply now the identity 
\begin{equation}\label{eq:id-b} \begin{split} 
b_x \Psi = \; &\sqrt{\rho_0} (J(x) -1) \Psi + \int dy  \, \nu (x-y) J(x) a_y^* J(y) \Psi \\ 
=\; &\sqrt{\rho_0} (J(x) -1) \Psi + \sqrt{\rho_0} \int dy  \, \nu (x-y) J(x) (J(y) - 1) \Psi \\ &+ \int dy  \, \nu (x-y) J(x) b_y^* J(y) \Psi  \end{split} \end{equation} 
which follows from (\ref{eq:id}), the definition (\ref{eq:bb}) and from the fact that $\int dy \, \nu (y) = 0$ (recall the notation $\nu = \gamma^{-1} \ast\sigma$). By Cauchy-Schwarz, we arrive at
\begin{equation}\label{eq:Nmb2} \begin{split}  
\langle \Psi, & \wt{\cN}^{m+1}_R (w) \Psi \rangle  \\ \leq \; &C \rho^{-10m\eps} + C \rho^{1-3\eps/2} \int_{|x-w| \leq R} dx \langle \wt{\cN}^{m/2}_R (w) \Psi , (J(x) - 1)^2  \wt{\cN}^{m/2}_R (w) \Psi \rangle \\ &+ C \rho^{1-3\eps/2} \int_{|x-w| \leq R} dx dy dy' \, \nu (x-y) \nu (x-y') \langle \Psi, (J(y) - 1)  \wt{\cN}_R^{m} (w) (J(y') -1) \Psi \rangle  \\ &+ (1 + C \rho^{3\eps/2}) \int_{|x-w| \leq R} dx dy dy' \, \nu (x-y) \nu (x-y') \, \langle \Psi, J(y) b_y \wt{\cN}_R^m (w) b_{y'}^* J(y') \Psi \rangle \,.\end{split} \end{equation} 
With (\ref{eq:1-J}), we can bound the second term by 
\[ \begin{split} C &\rho^{1-3\eps/2} \int_{|x-w| \leq R} dx dy \, \o_\ell (x-y) \langle \wt{\cN}^{m/2}_R (w) \Psi , a_y^* a_y \wt{\cN}^{m/2}_R (w) \Psi  \rangle \\ \leq \; &C  \rho^{1-3\eps/2} \int_{|y-w| \leq R+\ell} dx dy \, \o_\ell (x-y) \langle \wt{\cN}^{m/2}_R (w) \Psi , a_y^* a_y \wt{\cN}^{m/2}_R (w) \Psi  \rangle \\ \leq \; &C \rho^{1-3\eps/2} \ell^2 \langle \Psi,  \wt{\cN}^{m+1}_{R+\ell} (w) \Psi \rangle \leq C \rho^{1-3\eps/2-2\delta} \, \langle \Psi, \wt{\cN}_R^{m+1} (w) \Psi \rangle \leq \kappa \rho^\eps \langle \Psi, \wt{\cN}_R^{m+1} (w) \Psi \rangle \end{split} \]
for any $\kappa > 0$, if $\eps, \delta > 0$ and $\rho > 0$ are small enough. Here, we used \eqref{eq:R'R-b}. 

Up to errors arising from the restriction of the integrals to $|x-y| , |x-y'| \leq \ell_1$ (which can be bounded similarly to (\ref{eq:tails})), with Cauchy-Schwarz we can estimate the third term on the r.h.s. of (\ref{eq:Nmb2}) by 
\[ \begin{split} C &\rho^{1-3\eps/2} \int_{|x-w| \leq R, |x-y'| \leq \ell_1} dx dy dy' dz \, |\nu (x-y)|^2 \o_\ell (y'-z) \langle \wt{\cN}_R^{m/2} (w) \Psi, a_z^* a_z  \wt{\cN}_R^{m/2} (w) \Psi \rangle \\ \leq \; &C \rho^{5/2-3\eps/2} \ell^2 R^3  \int_{|z-w| \leq R+\ell_1+\ell} dz   \langle \wt{\cN}_R^{m/2} (w) \Psi, a_z^* a_z  \wt{\cN}_R^{m/2} (w) \Psi \rangle \\ \leq \; &C \rho^{7/2-3\eps/2-2\delta} R^6 \langle \Psi, \wt{\cN}_R^m (w) \Psi \rangle + C \rho^{5/2-3\eps/2-2\delta} R^3 \langle \Psi, \wt{\cN}_{R+\ell_1+\ell}^{m+1} (w) \Psi \rangle \\ \leq \; &C \rho^{-10 m \eps} + \kappa \rho^\eps \langle \Psi, \wt{\cN}_R^{m+1} (w) \Psi \rangle \end{split} \] 
if $\eps , \delta ,\rho > 0$ are small enough. Here, we wrote $a_z = b_z + \sqrt{\rho_0}, a_z^* = b_z^* + \sqrt{\rho_0}$ and we applied the induction assumption.

The last term on the r.h.s. of (\ref{eq:Nmb2}) can be bounded exactly as we did with the last term on the r.h.s. of (\ref{eq:Nm2}). Taking into account that the contribution arising from the commutator $[b_y, b_{y'}^*] = \delta (y-y')$ (analogous to the second term on the r.h.s. of (\ref{eq:Nm3})) can be bounded, with the induction assumption, by 
\[ \begin{split} \int_{|x-w| \leq R} dx dy \, |\nu (x-y)|^2 \langle J(y) \Psi, &(\wt{\cN}_R (w) + 1)^m J(y) \Psi \rangle \\ &\leq C \rho^{3/2} R^3 \langle \Psi, (\wt{\cN}_R (w) + 1)^m \Psi \rangle \leq \rho^{-9\eps-10m\eps} \end{split} \] 
we arrive at 
\[ \langle \Psi, \wt{\cN}_R^{m+1} (w) \Psi \rangle \leq C \rho^{-10 m \eps-9\eps} + (1 - C \rho^\eps) \langle \Psi,  \wt{\cN}_R^{m+1} (w) \Psi \rangle  \, , \]
which leads to 
\[ \langle \Psi, \wt{\cN}_R^{m+1} (w) \Psi \rangle \leq C \rho^{-10 (m+1) \eps} \, . \]
\end{proof}

Lemma \ref{lm:Nm} gives us control on localized monomials of arbitrary degree, with matching numbers of creation and annihilation operators.  
\begin{lemma} \label{lm:a-bds} 
Let $R = C \rho^{-1/2-3\eps}$, $k \in \bN \backslash \{ 0 \}$. Then we have 
\begin{equation}\label{eq:kas}  \int dx_1 \dots dx_k \, \prod_{j=2}^k \chi (|x_1 - x_j| \leq R)  \langle \Psi, a_{x_1}^* \dots a_{x_k}^* a_{x_k} \dots a_{x_1} \Psi \rangle \leq C \rho^{-(k-3)/2 - (12 k -9) \eps} L^3 \end{equation} 
and 
\begin{equation}\label{eq:kbs}  \int dx_1 \dots dx_k \, \prod_{j=2}^k \chi (|x_1 - x_j| \leq R)  \langle \Psi, b_{x_1}^* \dots b_{x_k}^* b_{x_k} \dots b_{x_1} \Psi \rangle \leq C \rho^{3/2 - (10k -9) \eps} L^3 \end{equation} 
if $\eps, \delta > 0$ and $\rho > 0$ are small enough (if $k=1$, we remove $\prod_{j=2}^k \chi (|x_1 - x_j| \leq R)$).
\end{lemma} 
\begin{proof} 
To show (\ref{eq:kas}), we write 
\[ \begin{split}  
\int dx_1 &\dots dx_k \, \prod_{j=2}^k \chi (|x_1 - x_j| \leq R)  \langle \Psi, a_{x_1}^* \dots a_{x_k}^* a_{x_k} \dots a_{x_1} \Psi \rangle \\ &= \int dx_1 \dots dx_{k-1} \, \prod_{j=2}^{k-1}  \chi (|x_1 - x_j| \leq R)
 \langle \Psi, a_{x_1}^* \dots a_{x_{k-1}}^* \cN_R (x_1) a_{x_{k-1}} \dots a_{x_1} \Psi \rangle\,. \end{split} \]
 Iterating $k$ times, we find 
 \[ \begin{split} \int dx_1 &\dots dx_k \, \prod_{j=2}^k \chi (|x_1 - x_j| \leq R)  \langle \Psi, a_{x_1}^* \dots a_{x_k}^* a_{x_k} \dots a_{x_1} \Psi \rangle 
 \\ &\leq \int dx_1 \langle \Psi, a_{x_1}^* \cN^{k-1}_R (x_1) a_{x_1} \Psi \rangle \\ &\leq CR^{-3} \int dz dx_1 \chi (|x_1 - z| \leq R) \langle \Psi, a_{x_1}^* \cN^{k-1}_R (x_1) a_{x_1} \Psi \rangle \\ &\leq CR^{-3} \int dz dx_1 \chi (|x_1 - z| \leq R) \langle \Psi, a_{x_1}^* \cN^{k-1}_{2R} (z) a_{x_1} \Psi \rangle \\ &\leq C R^{-3} \int dz \langle \Psi,   \cN^{k}_{2R} (z) \Psi \rangle \leq C \rho^{-(k-3)/2-(12k -9) \eps} L^3  \, ,\end{split} \]  
 where we used (\ref{eq:Nm}). Since the operators $b^*_x, b_x$ satisfy the same commutation relations as the operators $a_x^*, a_x$, (\ref{eq:kbs}) can be proven analogously, using (\ref{eq:Nmb}) instead of (\ref{eq:Nm}) in the last step.  
\end{proof}  

It will also be useful to bound the expectation of a single creation or annihilation operator. To this end, we state the following lemma.
\begin{lemma}\label{lm:1b}
We assume $\eps, \delta > 0$ and $\rho >0$ to be small enough. Then, we have 
\begin{equation}\label{eq:1bb} \Big| \sqrt{\rho_0} \int dx\langle \Psi, a_x \Psi \rangle - \rho_0 L^3 \Big| = \Big| \sqrt{\rho_0} \int dx \langle \Psi, b_x \Psi \rangle \Big| \lesssim \rho^{2-9\eps-2\delta} L^3 \end{equation} 
if $\eps, \delta > 0$ and $\rho >0$ are small enough. 
\end{lemma} 
\begin{proof} 
From (\ref{eq:id-b}), we recall  
\[ b_x \Psi = \sqrt{\rho_0} (J(x) -1) \Psi + \int dy \, \nu (x-y) f_\ell (x-y)  a_y^* J(x) J(y) \]
with the notation $\nu = \gamma^{-1}\ast \sigma$. Since $\int \nu (x) dx = 0$, we obtain 
\begin{equation}\label{eq:1b} \begin{split} \Big| \sqrt{\rho_0} &\int dx \langle \Psi, b_x \Psi \rangle \Big| \\ \lesssim \; &\rho_0 \int dx \langle \Psi, (1-J(x)) \Psi \rangle + \Big|  \sqrt{\rho_0} \int dx dy \, \nu (x-y) f_\ell (x-y) \langle \Psi, a_y^* (J(x)-1) J(y) \Psi \rangle \Big|  \\ &+ \Big|  \sqrt{\rho_0} \int dx dy \, \nu (x-y) (f_\ell (x-y) - 1) \langle \Psi, a_y^* J(y) \Psi \rangle \Big|\,.  \end{split} \end{equation} 
With (\ref{eq:1-J}), $\| \o_\ell \|_1 \lesssim \ell^2$ and (\ref{eq:rmkN}), we can bound the first term by 
\begin{equation}\label{eq:1b-1} \rho \int dx \langle \Psi, (1-J(x) ) \Psi \rangle \leq \rho \int dx dy \, \o_\ell (x-y) \langle \Psi, a_y^* a_y \Psi \rangle \lesssim  \rho^{2-3\eps-2\delta} L^3\,. \end{equation} 
To estimate the third term on the r.h.s. of (\ref{eq:1b}), we use $\| \nu \|_\infty \lesssim \rho$ and again $\| \o_\ell \|_1 \lesssim \ell^2$ and (\ref{eq:rmkN}). We find 
\[ \begin{split}  \Big|  \sqrt{\rho_0} \int dx dy \, &\nu (x-y) (f_\ell (x-y) - 1) \langle \Psi, a_y^* J(y) \Psi \rangle \Big| \\ &\lesssim \rho^{3/2} \ell^2 \int dy \, \| a_y \Psi \| \lesssim \rho^{3/2-2\delta} L^{3/2} \Big[ \int dy \, \| a_y \Psi \|^2 \Big]^{1/2} \leq \rho^{2-3\eps/2-2\delta} L^3\,. \end{split} \]
As for the second term on the r.h.s. of (\ref{eq:1b}) we find, using (\ref{eq:decay}) to restrict to $|x-y| \leq R$ and applying Cauchy-Schwarz, 
\[ \begin{split} 
 \Big| &\sqrt{\rho_0} \int dx dy \, \nu (x-y) f_\ell (x-y) \langle \Psi, a_y^* (1- J(x)) J(y) \Psi \rangle \Big| \\ &\lesssim \rho^2 L^3 + \sqrt{\rho} \Big[ \int_{|x-y| \leq R} dx dy \, \langle a_y \Psi, (1 - J(x)) a_y \Psi \rangle \Big]^{1/2} \\ &\hspace{4cm} \times \Big[ \int  dx dy \, |\nu (x-y)|^2 \langle \Psi, (1 - J(x)) \Psi \rangle \Big]^{1/2} \\
 &\lesssim \rho^2 L^3 + \rho^{5/4}  \Big[ \int_{|x-y| \leq R} dx dy dz \, \o_\ell (x-z) \langle \Psi, a_y^* a_z^* a_z a_y \Psi \rangle \Big]^{1/2} \\ &\hspace{4cm} \times \Big[ \int  dx dz \, \o_\ell (x-z) \langle \Psi, a_z^* a_z \Psi \rangle \Big]^{1/2} \\
 &\lesssim \rho^2 L^3 +  \rho^{5/4-2\delta}  \Big[ \int_{|z-y| \leq 2R} dy dz \, \langle \Psi, a_y^* a_z^* a_z a_y \Psi \rangle \Big]^{1/2} \Big[ \int  dz \, \langle \Psi, a_z^* a_z \Psi \rangle \Big]^{1/2} \\ &\lesssim \rho^{2-9\eps-2\delta} L^3
 \end{split} \]
 where we used (\ref{eq:rmkN}) and (\ref{eq:kas}) with $k=2$. This completes the proof of (\ref{eq:1bb}). 
\end{proof} 

Comparing (\ref{eq:Nm}) with (\ref{eq:Nmb}) or (\ref{eq:kas}) with (\ref{eq:kbs}), we observe that the estimates improve, when we factor out the Bose-Einstein condensate generated by the Weyl operator $W(\rho_0)$, switching from the original creation and annihilation operators $a^*, a_x$ to the operators $b_x^*, b_x$. In fact, we can derive even stronger bounds, if we also factor out the excitations of the condensate generated in (\ref{eq:trial}) by the Bogoliubov transformation $T$. For $x \in \Lambda$, we define the new creation and annihilation operators 
\begin{equation}\label{eq:def-cc}
\begin{split} 
c^*_x &=b^*(\g_x)-b(\s_x)= a^* (\gamma_x) - a (\sigma_x) - \sqrt{\rho_0} = T W (\rho_0) a^*_x W^* (\rho_0) T^* \\
c_x &=b(\g_x)-b^*(\s_x)= a (\gamma_x) - a^* (\sigma_x) - \sqrt{\rho_0} =  T W (\rho_0) a_x W^* (\rho_0) T^* \,.
\end{split} \end{equation} 
\begin{lemma} \label{lm:c's}
We have 
\begin{equation}\label{eq:cc-claim} \int dx \langle \Psi, c_x^* c_x \Psi \rangle \leq C \rho^{2-16\eps-2\delta} L^3\,. \end{equation} 
Moreover, let $R = C \rho^{-1/2-3\eps}$. Then, we have 
\begin{equation}\label{eq:cccc-claim} \int_{|x-y| \leq R} dx dy  \, \langle \Psi,  c_x^* c_y^* c_y c_x  \Psi \rangle \leq C \rho^{2-30\eps-\delta} L^3 \,.
\end{equation} 
More generally, for $\ph, \tau \in \{ \sigma, \gamma^{-1}\ast \sigma, \gamma -\mathbbm{1}, \mathbbm{1} \}$, we find (if $\tau = \mathbbm{1}$, we set $c (\tau_x) = c_x$) 
\begin{equation}\label{eq:cc-gh}
\int dx \, \langle \Psi, c^* (\tau_x) c (\tau_x) \Psi \rangle \leq C \rho^{2-22\eps-2\delta} L^3 \end{equation}
and  
\begin{equation}\label{eq:cccc-gh} 
\int_{|x-y| \leq R} dx dy  \, \langle \Psi,  c^* (\tau_x) c^* (\ph_y) c (\ph_y) c (\tau_x)  \Psi \rangle \leq C \rho^{2-42\eps-\delta} L^3 \,.
\end{equation} 
\end{lemma} 

\noindent {\bf Remark.} the r.h.s. of (\ref{eq:cccc-claim}) should be compared with the r.h.s. of (\ref{eq:kas}), (\ref{eq:kbs}), with $k=2$.

\begin{proof}
%Recall that 
%\[ c_x = a (\gamma_x) - a^* (\sigma_x) -\sqrt{\rho_0} \]
To prove this lemma, it will be convenient to express the operators $c, c^*$ in terms of the new fields 
\begin{equation}\label{eq:def-e} 
\begin{split} 
e_x &= c (\gamma^{-1}_x)  = a_x - a^* (\gamma^{-1} \ast\sigma_x) - \sqrt{\rho_0}  \\
e^*_x &= c^* (\gamma^{-1}_x)  = a^*_x - a (\gamma^{-1} \ast\sigma_x) - \sqrt{\rho_0} 
\end{split} 
\end{equation} 
(where we used the observation $\widehat{\gamma^{-1}} (0) = 1/ \widehat{\gamma} (0) = 1/ \sqrt{1-\widehat{\sigma}^2 (0)} = 1$) and to treat $e_x, e_x^*$ as a small perturbation of   
\begin{equation}\label{eq:def-g} \begin{split} 
g_x &= a_x - a^* (k_x) -\sqrt{\rho_0} , \\ g_x^* &= a_x^* - a (k_x) - \sqrt{\rho_0} \end{split} \end{equation}
with the kernel $k_x (y) = f_\ell (x-y) (\gamma^{-1} \ast\sigma ) (x-y)$. To control the difference between (\ref{eq:def-e}) and (\ref{eq:def-g}), we notice that  
\begin{equation}\label{eq:exfx} e_x = g_x + a^* (\xi_x), \end{equation} with $\xi_x (y) = \xi (x-y)$ and $\xi (y) = k (y) - (\gamma^{-1} \ast\sigma) (y)  = (f_\ell (y) - 1) (\gamma^{-1}\ast \sigma) (y)$.
From Lemma \ref{lm:fell}  and Lemma \ref{lm:eta}, we find 
\begin{equation}\label{eq:kxi-bd}  
\| k \|_2^2 \lesssim \rho^{3/2}, \quad \| \xi \|_2^{2} \lesssim \rho^{2+\delta} , \quad \| \xi \|_1 \lesssim \rho^{1-\delta}  \,.
\end{equation} 
The advantage of working with the operators (\ref{eq:def-g}) is clear from the identity 
\begin{equation}\label{eq:g-id}  g_x \Psi = \sqrt{\rho_0} (J(x) - 1) \Psi + \int dy \, k (x-y) (J(x) J(y) -1) \Psi \end{equation} 
which follows from (\ref{eq:id}) and from the intertwining relations (\ref{eq:pull}). The appearance of the differences $J(x) - 1$ in the first and $J(x) J(y) -1$ in the second term are crucial to get improved bounds, compared with the estimates in Lemma \ref{lm:Nm}. 

We show first how to bound (\ref{eq:cc-claim}), (\ref{eq:cccc-claim}) in terms of quantities defined through the fields (\ref{eq:def-g}). With $c_x = e (\gamma_x)$, using $\gamma * \gamma = 1 + \sigma * \sigma$ and the bound $\| \sigma * \sigma \|_1 \lesssim \| \sigma \|_1^2 \lesssim \rho^{-\eps}$ from Lemma \ref{lm:eta}, we obtain 
\[ \begin{split} \int  dx \langle \Psi , c_x^* c_x \Psi \rangle   &= \int dy dy' \, (\gamma * \gamma) (y-y') \langle \Psi, e_y^* e_{y'} \Psi \rangle \lesssim \rho^{-\eps} \int dx \langle \Psi, e_x^* e_x \Psi \rangle \,.
\end{split} \]
From (\ref{eq:exfx}), we have  
\[ \int dx \, \langle \Psi, e_x^* e_x \Psi \rangle \lesssim \int dx \langle \Psi, g_x^* g_x \Psi \rangle +  \int dx \langle \Psi, a (\xi_x) a^* (\xi_x) \Psi \rangle \,.\]
With (\ref{eq:kxi-bd}) and (\ref{eq:rmkN}), we can estimate
\[ \begin{split} \int dx \langle \Psi, a (\xi_x) a^* (\xi_x) \Psi \rangle &= \int dx \, \| \xi_x \|_2^2 + \int dy dy' \, (\xi * \xi) (y-y') \langle \Psi, a_y^* a_{y'} \Psi \rangle \\ &\lesssim \rho^{2+\delta} L^3 +  \rho^{2-2\delta} \int dy \, \langle \Psi , a_y^* a_y \Psi \rangle  \lesssim \rho^2 L^3 \end{split}  \]
if $\eps, \delta > 0$ are small enough. We conclude that 
\begin{equation} \label{eq:cg} \int dx \langle \Psi , c_x^* c_x \Psi \rangle \lesssim \rho^{-\eps} \int dx \langle \Psi, g_x^* g_x \Psi \rangle + \rho^{2-\eps} L^3\,. \end{equation} 
To express (\ref{eq:cccc-claim}) in terms of the fields $g_x, g_x^*$ we proceed similarly. First, we write 
\[ \begin{split}  \int_{|x-y| \leq R} &dx dy \, \| c_x c_y \Psi \|^2 \\ &= \int_{|x-y| \leq R} dx dy \, \langle \Psi, c_x^* e^* (\gamma_y) e(\gamma_y) c_x \Psi \rangle \\ &\lesssim \int_{|x-y| \leq R} dx dy \, \langle \Psi, c_x^* e_y^* e_y c_x \Psi \rangle \\ &\hspace{.4cm} + \int_{|x-y| \leq R} dx dy dz dz'  (\gamma -\mathbbm{1})(y-z) (\gamma-\mathbbm{1}) (y-z')  \langle \Psi, c_x^* e^*_z e_{z'}  c_x \Psi \rangle \,. \end{split}  \]
Using the decay of the kernel $(\gamma-\mathbbm{1})$, as established in (\ref{eq:decay}), we can restrict the integrals to $|x-z|, |x-z'| \leq 2R$. Contributions from $|x-z| > 2R$ or $|x-z'| > 2R$ (which imply that $|y-z| > R$ or $|y-z'| > R$) can be controlled similarly as the term (\ref{eq:ex-decay}) (here, we additionally need to express $c_x^*, e_z^*, e_{z'} , c_x$ in terms of the fields $a^*, a$, recalling the definitions (\ref{eq:def-cc}), (\ref{eq:def-e})), to be able to apply Lemma \ref{lm:Nm} to estimate the resulting expectations). We find  
\[ \begin{split} &\int_{|x-y| \leq R} dx dy \, \| c_x c_y \Psi \|^2 \\ \lesssim \; &\rho^2 L^3 + \int_{|x-y| \leq R} dx dy \langle \Psi, c_x^* e_y^* e_y c_x \Psi \rangle \\ &+ \int_{|x-y| \leq R, |x-z|, |x-z'| \leq 2R} dx dy dz dz'  (\gamma - \mathbbm{1}) (y-z) (\gamma - \mathbbm{1}) (y-z')  \langle \Psi, c_x^* e^*_z e_{z'}  c_x \Psi \rangle \\ 
%= \; &\rho^2 L^3 + \int_{|x-y| \leq R} dx %dy \,  \langle \Psi, c_x^* e^* (\chi_x 
%\gamma_y)  e (\chi_x \gamma_y)  c_x \Psi 
%\rangle\\ 
\lesssim  \; &\rho^2 L^3+ \rho^{-3\eps}  \int_{|x-z| \leq R}  dx dz \langle \Psi, c_x^* e^*_z e_z c_x \Psi \rangle  \end{split} \]
where we used Cauchy-Schwarz and the estimate $\| (\gamma - \mathbbm{1}) * (\gamma -\mathbbm{1} )\|_1 \lesssim \| \gamma -\mathbbm{1} \|_1^2 \lesssim \rho^{-3\eps}$. Noticing that $[c_x, e_z] = 0$, we can move $c_x^*, c_x$ in the middle of the product and repeat the same argument. We conclude that 
\begin{equation}\label{eq:ctoe} \int_{|x-y| \leq R} dx dy \, \| c_x c_y \Psi \|^2 \lesssim \rho^2 L^3 + \rho^{-6\eps}  \int_{|z-w| < R} dzdw \, \| e_z e_w \Psi \|^2\,. \end{equation} 
To switch to the $g, g^*$ operators, we recall (\ref{eq:exfx}). We find 
\[\begin{split} &\int_{|x-y| \leq R} dx dy \, \| e_x e_y \Psi \|^2 \\ & \lesssim \int_{|x-y| \leq R} dx dy \langle \Psi, e_x^*  g_y^* g_y e_x \Psi \rangle +  \int_{|x-y| \leq R} dx dy \langle \Psi, e_x^*  a (\xi_y) a^* ( \xi_y)  e_x \Psi \rangle \,.\end{split}  \]
With $[e_x^* , g_y^*] = - \langle \xi_x, k_y \rangle$, $[e_x^* , a^* (\xi_y)] = - \langle \gamma^{-1}\ast \sigma_x, \xi_y \rangle$, $[ a(\xi_y) , a^* (\xi_y) ] = \| \xi_y \|^2$, with (\ref{eq:kxi-bd}) and since, additionally,   
\begin{equation}\label{eq:kxi-bd2}  |\langle \xi_x, k_y \rangle | ,  |\langle \gamma^{-1}\ast \sigma_x, \xi_y \rangle | \lesssim \rho^{7/4+\delta/2} \end{equation} 
we conclude that 
\begin{equation}\label{eq:eeg}   \begin{split} &\int_{|x-y| \leq R} dx dy \, \| e_x e_y \Psi \|^2 \\ &\lesssim \int_{|x-y| \leq R} dx dy \langle \Psi,  g_y^* e_x^* e_x g_y \Psi \rangle + \int_{|x-y| < R} dx dy \langle \Psi , a^* (\xi_y) e_x^* e_x a(\xi_y) \Psi \rangle \\ &\hspace{.4cm} + \rho^{1/2-9\eps+\delta} \int dx \langle \Psi, e_x^* e_x \Psi \rangle + \rho^{2-9\eps+\delta}  L^3 \,.\end{split} \end{equation}
To bound the first term, we apply (\ref{eq:exfx}) again. We find
\[ \begin{split} \int_{|x-y| \leq R} dx dy \langle &\Psi,  g_y^* e_x^* e_x g_y \Psi \rangle \\ &\lesssim \int_{|x-y| \leq R} dx dy \langle \Psi,  g_y^* g_x^* g_x g_y \Psi \rangle + \int_{|x-y| \leq R} dx dy \langle \Psi,  g_y^* a (\xi_x) a^* (\xi_x)  g_y \Psi \rangle \,. \end{split} \]
Proceeding as we did above, using again the bounds (\ref{eq:kxi-bd}), (\ref{eq:kxi-bd2}), and combining with (\ref{eq:eeg}), we obtain 
\begin{equation}\label{eq:eegg}  \begin{split} 
&\int_{|x-y| \leq R} dx dy \, \| e_x e_y \Psi \|^2 \\ &\lesssim  \int_{|x-y| \leq R} dx dy \langle \Psi,  g_y^* g_x^* g_x g_y \Psi \rangle + \int_{|x-y| \leq R} dx dy \, \langle \Psi, a^* (\xi_x) \big( e_y^* e_y + g_y^* g_y \big) a (\xi_x) \Psi \rangle \\ &\hspace{.4cm} + \rho^{1/2-9\eps+\delta} \int dy \langle \Psi, g_y^* g_y \Psi \rangle + \rho^{2-9\eps+\delta} L^3\,. \end{split} \end{equation} 
In the second term, we estimate 
\begin{equation}\label{eq:eyey} e_y^* e_y   \lesssim a_y^* a_y + a^* (\gamma^{-1}\ast \sigma_y ) a (\gamma^{-1} \ast\sigma_y) +  \rho \end{equation} 
and similarly for $g_y^* g_y$ (with $\gamma^{-1} \ast\sigma$ replaced by $k$). To bound the resulting contribution
\[ \begin{split}  \int_{|x-y| \leq R} dx dy \, &\langle \Psi, a^* (\xi_x) a_y^* a_y a (\xi_x) \Psi \rangle \\ &= \int_{|x-y| \leq R} dx dy dz dz' \, \xi (x-z) \xi (x-z') \langle \Psi, a_z^* a_y^* a_y a_{z'} \Psi \rangle \end{split} \]
we use the fast decay of $\xi$, to restrict the integrals over $z,z'$ to $|x-z|, |x-z'| \leq R$, with only small errors. For example, consider the contribution 
\[ \begin{split} \text{A} &= \int_{|x-y| \leq R, |x-z'| \leq R, |x-z| \geq R} dx dy dz dz' \, |\xi (x-z)| |\xi (x-z')| \, \langle \Psi, a^*_z a^*_y a_y a_z \Psi \rangle \\ &\lesssim \rho^{1+\delta} \sum_{n \in \bZ^3 : |n| \geq 2}  \int_{\substack{|x-y| \leq R, |x-z'| \leq R, \\ |x-z - n R/2| \leq R/2}} dx dy dz dz' \frac{\rho^{3k\eps/2}}{|n|^k} \| a_z a_y \Psi \| \| a_{z'} a_y \Psi \| \\ &\lesssim \rho^{1+3\eps k/2+\delta} \sum_{|n| \geq 2}  |n|^{-k} \int dx \langle \Psi, \cN_R (x) \cN_{R/2} (x- n R/2) \Psi \rangle^{1/2} \langle \Psi,  \cN^2_R (x) \Psi \rangle^{1/2} \,.\end{split} \]
By the translation invariance of $\Psi$ and Lemma \ref{lm:Nm}, we find that 
\[ \begin{split} 
|\text{A}| &\lesssim \rho^{1+ 3\eps k/2+\delta} L^3  \langle \Psi, \cN_R (0)^2 \Psi \rangle \lesssim \rho^{1+3\eps k/2+\delta} \rho^{-1-24 \eps} L^3 \lesssim \rho^2 L^3 \end{split} \]
if $k$ is chosen large enough. We conclude, with Lemma \ref{lm:a-bds} and recalling (\ref{eq:kxi-bd}), that 
\[ \begin{split} &\int_{|x-y| \leq R} dx dy \, \langle \Psi, a^* (\xi_x) a_y^* a_y a (\xi_x) \Psi \rangle \\ &\lesssim \rho^2 L^3 + \int_{\substack{|x-y| \leq R \\ |x-z| \leq R \\ |x-z'| \leq R}} dx dy dz dz' \, \xi (x-z) \xi (x-z') \, \langle \Psi, a^* _x a_y^* a_y a_x \Psi \rangle \\ &\lesssim \rho^2 L^3 + \int dydzdz' \, \chi( |y-z'| \leq 2R) (|\xi| * |\xi|) (z-z') \chi (|y-z| \leq 2R) \langle \Psi, a_y^* a_z^* a_{z'} a_y \Psi \rangle \\ &\lesssim \rho^2 L^3 + \||\xi | * |\xi | \|_1 \int_{|y-z| \leq 2R} dy dz \langle \Psi, a_z^* a_y^* a_y a_z \Psi \rangle \lesssim \rho^2 L^3 + \rho^{2-2\delta} \rho^{1/2-15\eps} L^3 \lesssim \rho^2 L^3 \end{split} \]  
if $\eps , \delta > 0$ are small enough. The other contributions from (\ref{eq:eyey}) can be bounded similarly. Taking into account (\ref{eq:ctoe}), we arrive at 
\[  \begin{split} 
\int_{|x-y| \leq R} &dx dy \, \| c_x c_y \Psi \|^2 \\ &\lesssim  \rho^{-6\eps} \int_{|x-y| \leq R} dx dy \langle \psi,  g_y^* g_x^* g_x g_y \psi \rangle  + \rho^{1/2-15\eps+\delta} \int dy \langle \Psi, g_y^* g_y \Psi \rangle + \rho^{2-15\eps+\delta} L^3 \,.\end{split} \]

With (\ref{eq:cg}) and the last equation, to complete the proof of the lemma it is enough to show 
\begin{equation}\label{eq:goalsc}\begin{split} 
\int dx \langle \Psi, g_x^* g_x \Psi \rangle &\lesssim \rho^{2-15\eps-2\delta} L^3 \\ 
\int_{|x-y| \leq R} dx dy \langle \Psi, g_x^* g_y^* g_y g_x \Psi \rangle &\lesssim \rho^{2-24\eps-\delta} L^3\,. \end{split} \end{equation} 

To this end, we use (\ref{eq:g-id}) to estimate 
\[ \begin{split}  &\int  dx \, \langle \Psi, g_x^* g_x \Psi \rangle \\ &\lesssim \rho \int dx\, \| (J(x)-1) \Psi \|^2 \\ &\hspace{.4cm} + \int dx dzdz' k (x-z) k (x-z') \langle a_z^* (J(x) J(z) -1) \Psi , a_{z'}^* (J(x) J(z') -1) \Psi \rangle = \text{T}_1 + \text{T}_2 \,.
\end{split}  \]
To bound the first term, we observe that, from (\ref{eq:1-J}),  
\begin{equation}\label{eq:Jx-} (1- J(x))^2 \leq 1-J(x) \lesssim \int dy \, \o_\ell (x-y) \,  a_y^* a_y \,. \end{equation} 
With Lemma \ref{lm:fell} and \eqref{eq:rmkN}, we obtain 
\[ \text{T}_1 \lesssim \rho^{2-3\eps-2\delta} L^3\,. \]
To control $\text{T}_2$, we commute $a_z$ and $a^*_{z'}$. We find 
\[ \begin{split} \text{T}_2 \lesssim \; &\int dx dz dz' \, k (x-z) k (x-z') \langle a_{z'} (J(x) J(z) -1) \Psi, a_z (J(x) J(z') -1) \Psi \rangle  \\  &+\int dx dz \, |k (x-z)|^2 \| (J(x) J(z) -1) \Psi \|^2 = \text{T}_{21} + \text{T}_{22}  \,.\end{split} \]
Writing $J(x) J(z) - 1 = (J(x) -1) J(z) + (J(z) -1)$, we find, since $\| k \|_2^2 \lesssim \rho^{3/2}$, 
\[ \text{T}_{22} \leq \int dx dz |k(x-z)|^2 \, \| (J(x) -1) \Psi \|^2 \lesssim \rho^{5/2-3\eps-2\delta} L^3 \lesssim \rho^2 L^3 \] 
if $\eps, \delta > 0$ are small enough. Furthermore, using the fast decay (\ref{eq:decay}) of $k (x)$ for $|x| > R$ and applying Cauchy-Schwarz, we obtain 
\[ \begin{split} \text{T}_{21} &\lesssim \rho^2 L^3 + \int_{|x-z| \leq R}  dx dz dz' |k(x-z')|^2 \| a_z (J(x) J(z') - 1) \Psi \|^2 \\ &\lesssim \rho^2 L^3 +  \int_{|x-z| \leq R}  dx dz dz' |k(x-z')|^2 \| (J(x) J(z') f_\ell(x-z) f_\ell(z-z') - 1) a_z \Psi \|^2\,. \end{split} \]
Decomposing 
\begin{equation}\label{eq:JJ-deco} \begin{split} J(x) J(z') &f_\ell(x-z) f_\ell(z-z') - 1 \\ &= J(x) J(z') (f_\ell(x-z) f_\ell(z-z') -1) + (J(z') -1) J(x) + (J(x) -1)\end{split} \end{equation} we can bound the integral on the r.h.s. of $\text{T}_{21}$ by the sum of several terms. Let us focus on the term $\text{T}_{211}$, proportional to $J(x) -1$. With (\ref{eq:Jx-}), we find 
\[ \begin{split} \text{T}_{211} \lesssim \int_{|x-z| \leq R}  dx dz & dz' |k(x-z')|^2 \| (J(x) -1) a_z \Psi \|^2\\  &\lesssim \rho^{3/2} \int_{|x-z| \leq R}  dx dz dw \, \o_\ell (x-w) \langle \Psi, a_w^* a_z^* a_z a_w \Psi \rangle \\ &\lesssim \rho^{3/2} \ell^2 \int_{|z-w| \leq 2R} dz dw \, \langle \Psi, a_w^* a_z^* a_z a_w \Psi \rangle  \leq \rho^{2 -15 \eps -2\delta} L^3 \end{split} \]
where we applied Lemma \ref{lm:a-bds}. The other contributions from (\ref{eq:JJ-deco}) can be handled similarly. This leads to the first bound in (\ref{eq:goalsc}). 

To prove the second bound, we go back to (\ref{eq:g-id}) and we apply $a_y$. With (\ref{eq:pull}) and $[a_y, a_z^*] = \delta (y-z)$, we obtain 
\[ \begin{split} a_y g_x \Psi = \; & \sqrt{\rho_0} \, (J(x) - 1) a_y \Psi + \sqrt{\rho_0} \, (f_\ell(x-y) - 1) J(x) a_y \Psi  + k(x-y) (J(x) J (y) -1) \Psi \\ &+ \int dz \, k(x-z) a_z^* (J(x) J(z) -1) a_y \Psi\\ &
+ \int dz \,   k (x-z) (f_\ell(x-y) f_\ell(z-y) - 1) a_z^* J(x) J(z) a_y \Psi \,. \end{split} \]
On the r.h.s., we replace $a_y = g_y+ a^* (k_y) + \sqrt{\rho_0}$. Afterwards, we use (\ref{eq:pull}) to commute $a^* (k_y)$ to the left of the operators $J(x), J(z)$ (after expanding $a^* (k_y) = \int dz \, k (y-z) a_z^*$). At this point, we subtract factors $a^* (k_y) g_x \Psi$ and $\sqrt{\rho_0} g_x \Psi$ (using (\ref{eq:g-id}), we get four terms, canceling contributions in $a_y g_x \Psi$). We find 
\[ \begin{split} 
g_y g_x \Psi = \; &\sqrt{\rho_0} \, (J(x) - 1) g_y \Psi + \sqrt{\rho_0}  \, \int dz \, k (y-z) (f_\ell(x-z) -1) a_z^* J(x) \Psi \\ &+ \sqrt{\rho_0}  \, (f_\ell(x-y) - 1) J(x) g_y \Psi \\ &+ \sqrt{\rho_0}  \, (f_\ell(x-y) - 1) \int dz  \, k (y-z) f_\ell(z-x) a_z^* J(x) \Psi \\ &+ \rho_0 \, (f_\ell(x-y) - 1) J(x) \Psi + k(x-y) (J(x) J(y) - 1) \Psi \\ &+ \int dz  \, k(x-z) a_z^* (J(x) J(z) - 1) g_y \Psi \\ &+ \int dz dw  \, k(x-z) k(y-w) (f_\ell(w-x) f_\ell(w-z) - 1) a_z^* a_w^* J(x) J(z) \Psi \\ &+ \int dz  \, k (x-z) (f_\ell(x-y) f_\ell(z-y) - 1) a_z^* J(x) J(z) g_y \Psi 
\\ &+ \int dz dw  \, k (x-z) k(y-w) f_\ell(x-w) f_\ell(z-w) (f_\ell(x-y) f_\ell(z-y) - 1) \\ &\hspace{8cm} \times  a_z^* a_w^* J(x) J(z) \Psi \\ &+ \sqrt{\rho_0} \int dz  \, k(x-z) (f_\ell(x-y) f_\ell(z-y) - 1) a_z^* J(x) J(z) \Psi \,.\end{split} \]
We apply again (\ref{eq:g-id}) to compute $g_y \Psi$ and, again, we commute creation operators to the left of all $J$-factors. After a long but straightforward computation, we arrive at 
\begin{equation} \label{eq:fxfy}  \begin{split} 
g_y g_x \Psi = \; &\rho_0 \, (J(x) - 1) (J(y) -1) \Psi + \rho_0 \, (f_\ell(x-y) - 1) J(x) J(y) \Psi \\ &+ k (x-y) (J(x) J(y) - 1) \Psi \\ &+ \sqrt{\rho_0} \int dz \, k(y-z) (f_\ell(x-y) f_\ell(x-z) -1) a_z^* J(x) J(y) J(z) \Psi \\ &+ \sqrt{\rho_0} \int dz \, k(x-z) (f_\ell(x-y) f_\ell(y-z) - 1) a_z^* J(x) J(y) J(z) \Psi \\ &+ \sqrt{\rho_0} \int dz \, k (y-z) a_z^* (J(x) - 1) (J(y) J(z) -1) \Psi \\ &+ \sqrt{\rho_0} \int dz \, k (x-z) a_z^* (J(x) J(z) -1) (J(y) -1)  \Psi  \\ &+ \int dz dw \, k (x-z) k(y-w) (f_\ell(x-w) f_\ell(z-w) f_\ell(x-y) f_\ell(z-y) - 1) \\ &\hspace{7cm} \times a_z^* a_w^* J(x) J(z) J(y) J(w) \Psi \\ &+ \int dz dw \, k(x-z) k(y-w) a_z^* a_w^* (J(x) J(z) - 1) (J(y) J(w) - 1) \Psi \\ =: \, &\sum_{j=1}^9 A^{(j)}_{x,y} \Psi\,. \end{split} \end{equation} 
To bound the first contribution, we observe that, similarly to (\ref{eq:Jx-}), 
\begin{equation} \label{eq:JxJy-}  \begin{split} 
(J(x) &- 1)^2 (J(y) -1)^2 \\ &\leq \int  dz dw \, \o_\ell (x-z) \o_\ell (y-w) \, a_z^* a_w^* a_w a_z + \int dz \, \o_\ell (x-z) \o_\ell (y-z)  \, a_z^* a_z\,.  \end{split} \end{equation} 
Thus,
\[ \begin{split}  
\int_{|x-y| \leq R} dx dy \,  \| A_{x,y}^{(1)} \Psi \|^2 \leq \; &\rho^2 \int_{|x-y| \leq R} dx dy \int dz dw \, \o_\ell (x-z)  \o_\ell (y-w) \, \langle \Psi, a_z^* a_w^* a_w a_z \Psi \rangle \\ &+ \rho^2 \int_{|x-y| \leq R} dx dy  dz \, \o_\ell (x-z)  \o_\ell (y-z)  \, \langle \Psi, a_z^* a_z \Psi \rangle\,. \end{split} \]
Exchanging the order of the integrals, we obtain, with Lemma \ref{lm:a-bds}, 
\[ \begin{split}  \int_{ |x-y| \leq R} &dx dy \, \| A_{x,y}^{(1)} \Psi \|^2  \\ &\lesssim  \rho^2 \ell^4 \int_{|z-w| \leq 2R} dz dw \, \langle \Psi , a_z^* a_w^* a_w a_z \Psi \rangle + \rho^2 \ell^4  \int dz \langle \Psi , a_z^* a_z \Psi \rangle \lesssim \rho^{2} L^3\end{split} \]
recalling that $\ell = C\rho^{-\delta}$ and choosing $\delta > 0$ small enough.

Since $0\leq 1- f_\ell(x-y) \lesssim \chi (|x-y| \leq \ell)/|x-y|$, we immediately obtain 
\[ \int_{|x-y| \leq R} dx dy \, \| A_{x,y}^{(2)} \Psi \|^2 \lesssim \rho^{2-\delta} L^3\,. \]
With (\ref{eq:kxi-bd}), $(J(x) J(y) - 1)^2 \lesssim (J(x)-1)^2 + (J(y)-1)^2$ and exchanging $x$ and $y$, we can bound (using also Lemma \ref{lm:a-bds})   
\[ \begin{split}  \int_{|x-y| \leq R}dx dy \| A_{x,y}^{(3)} \Psi \|^2 &\lesssim \int_{|x-y| \leq R} dx dy \, |k (x-y)|^2 \| (J(x) - 1) \Psi \|^2 \\ &\lesssim  \int dx dy dz \, |k (x-y)|^2 \o_\ell (x-z) \langle \Psi, a_z^* a_z \Psi \rangle \\ &\lesssim \rho^{3/2} \ell^2 \int dz \langle \Psi , a_z^* a_z \Psi \rangle \lesssim \rho^2 L^3 \end{split} \]
if $\eps , \delta > 0$ are small enough. 

To estimate $A^{(4)}_{x,y}$, we write $f_\ell(x-y) f_\ell(y-z) - 1 = (f_\ell (z-y) - 1) f_\ell(x-y) + (f_\ell(x-y) - 1)$ and we focus on the contribution proportional to $(f_\ell(x-y) - 1)$, which we denote by $A^{(4,1)}_{x,y}$ (the other contribution can be controlled similarly). We have    
\[ \begin{split}  \int_{|x-y| \leq R} dx dy \, \| A_{x,y}^{(4,1)} \Psi \|^2 &\lesssim \rho \int dx dy dz dz' \, k(y-z) k(y-z') \o_\ell^2 (x-y) \\ &\hspace{2cm} \times \langle J(x) J (y) J (z') \Psi, a_{z'} a_{z}^* J(x) J (y) J (z) \Psi \rangle \\
&= \rho \int dx dy dz \, |k(y-z)|^2 \o_\ell^2 (x-y) \| J(x) J (y) J (z) \Psi \|^2 \\ &\hspace{.4cm} + 
 \rho \int dx dy dz dz' \, k(y-z) k(y-z') \o_\ell^2 (x-y) \\ &\hspace{2cm} \times \langle a_z J(x) J (y) J (z') \Psi, a_{z'} J(x) J (y) J (z) \Psi \rangle \,.
  \end{split} \]
 Estimating $0 \leq J(x) \leq 1$, using the decay (\ref{eq:decay}) of $k$ to restrict the integrals to $|y-z|, |y-z'| \leq R$, applying Cauchy-Schwarz and Lemma \ref{lm:Nm}, we find 
\[ \begin{split} 
& \int_{|x-y| \leq R} dx dy \, \| A_{x,y}^{(4,1)} \Psi \|^2 \\ &\lesssim \rho^{5/2} \ell L^3 + \rho \int_{|y-z| \leq R}  dx dy dz dz' \, |k(y-z')|^2 \o_\ell^2 (x-y)  \| a_{z} \Psi \|^2  \lesssim \rho^{2-12\eps-\delta} L^3 \,.\end{split} \]
 The term $A_{x,y}^{(5)}$ can be handled like $A^{(4)}_{x,y}$ (exchanging $x$ and $y$). As for $A_{x,y}^{(6)}$, we write $J(y) J(z) - 1 = (J(y) - 1) + J(y) (J(z) - 1)$ and we focus on the contribution $A^{(6,1)}_{x,y}$ proportional to $(J(y) -1)$ (the others can be bounded similarly). We compute 
 \[\begin{split}  \int_{|x-y| \leq R} &dx dy \, \| A^{(6,1)}_{x,y} \Psi \|^2 \\ = \; &\rho_0 \int_{|x-y| \leq R}  dx dy dz dz' \, k (y-z) k(y-z')  \\ &\hspace{2cm} \times \langle (J(x) - 1) (J(y) -1) \Psi , a_z a_{z'}^* (J(x) -1) (J(y) -1) \Psi \rangle \\ = \;&\rho_0 \int_{|x-y| \leq  R}  dx dy dz \, |k (y-z)|^2 \| (J(x) - 1) (J(y) -1) \Psi \|^2 \\ &+ \rho_0 \int_{|x-y| \leq R} dx dy dz dz' \, k(y-z) k(y-z') \\ &\hspace{2cm} \times \langle a_{z'} (J(x) - 1) (J(y) -1) \Psi , a_z (J(x) - 1) (J(y) -1) \Psi \rangle\,. \end{split} \]
To bound the first term, we use (\ref{eq:JxJy-}). As for the second term, we apply Cauchy-Schwarz. We obtain
\[ \begin{split}  &\int_{|x-y| \leq R} dx dy \, \| A^{(6,1)}_{x,y} \Psi \|^2 \\ &\lesssim \rho \int_{|x-y| \leq R} dx dy dz dv dw \, |k (y-z)|^2 \o_\ell (x-v) \o_\ell (y-w)  \langle \Psi, a_v^* a_w^* a_w a_v\Psi \rangle \\ &\hspace{.4cm} + \rho \int_{|x-y| \leq R} dx dy dz dw \, |k (y-z)|^2  \o_\ell (x-w) \o_\ell (y-w) \langle \Psi, a_w^* a_w \Psi \rangle \\ &\hspace{.4cm} + \rho \int_{|x-y| \leq R, |y-z| \leq R}  dx dy dz dz' \, |k (y-z')|^2 \| a_{z} (J(x) - 1) (J(y) -1) \Psi \|^2 + \rho^2 L^3  \end{split} \]
where we used (\ref{eq:decay}) to restrict the last integral to $|y-z| \leq R$, producing the small error $\rho^2 L^3$. The first two contributions can be easily estimated, using Lemma \ref{lm:a-bds}. In the last term, we integrate over $z'$ and we commute $a_{z}$ through the operator $(J(x) - 1) (J(y) - 1)$. We find, for $\eps , \delta > 0$ small enough,   
\[ \begin{split} 
 \int_{|x-y| \leq R} &dx dy \, \| A^{(6,1)}_{x,y} \Psi \|^2 \\  \lesssim \; &\rho^2  L^3 + \rho^{5/2} \int dx dy dz \,  \o_\ell^2 (x-z) \o_\ell^2 (y-z) \| J(x) J(y) a_{z} \Psi \|^2 \\ &+ \rho^{5/2} \int dx dy dz \, \chi (|y-z| < R) \o_\ell^2 (x-z) \| J(x) (J(y) -1) a_{z} \Psi \|^2 \\ &+ \rho^{5/2} \int dx dy dz \, \chi (|x-z|, |y-z| < 2R)  \| (J(x) - 1) ( J(y) -1) a_{z} \Psi \|^2\,.  \end{split} \] 
Using $0\leq \o_\ell (x) \lesssim \chi (|x| \leq \ell )/ |x|$, $J(x) \leq 1$ and applying (\ref{eq:Jx-}), (\ref{eq:JxJy-}), we arrive at 
 \[ \begin{split} 
 \int_{|x-y| \leq R} &dx dy \, \| A^{(6,1)}_{x,y} \Psi \|^2 \\  \lesssim \; &\rho^2  L^3  + \rho^{5/2} \ell^3 \int_{|w-z| \leq 2R}  dw dz \,  \langle \Psi, a_w^* a_{z}^* a_{z} a_w \Psi \rangle \\ &+ \rho^{5/2} \ell^4 \int_{|w-z|, |v-z| \leq 3 R}  dw dv dz \,  \langle \Psi, a_w^* a_v^* a_{z}^* a_{z} a_v a_w \Psi \rangle \lesssim \rho^2 L^3 
 \end{split} \] 
 if $\eps , \delta > 0$ are chosen small enough. In the last inequality, we used Lemma \ref{lm:a-bds}. The term $A^{(7)}_{x,y}$ can be handled analogously.  
  
 To bound the term $A^{(8)}_{x,y}$, we write 
 \[ \begin{split} f_\ell(x-w) &f_\ell(z-w) f_\ell(x-y) f_\ell(z-y) - 1 \\ = \; &(f_\ell(x-w) - 1) + f_\ell(x-w) (f_\ell(z-w) -1) + f_\ell(x-w) f_\ell(z-w) (f_\ell(x-y) - 1) \\ &+ f_\ell(x-w) f_\ell(z-w) f_\ell (x-y) (f_\ell (z-y) -1) \,.\end{split} \]
We focus on the contribution $A^{(8,1)}_{x,y}$ proportional to $f_\ell(x-w) -1 = \o_\ell (x-w)$. We estimate it by
\[ \begin{split} &\int_{|x-y| \leq R} dx dy \, \| A^{(8,1)}_{x,y} \Psi \|^2  \\ &\lesssim  \int_{|x-y| \leq R}  dx dy dz dw dz' dw' \, k(x-z) k(x-z') k(y-w) k(y-w')   \o_\ell (x-w) \o_\ell (x-w')   \\ & \hspace{4cm} \times  \langle J(x) J(y) J(z) J(w) \Psi, a_z a_w a_{w'}^* a_{z'}^* J(x) J(y) J(z') J(w') \Psi \rangle \,.\end{split} \]
We write 
\begin{equation} \label{eq:aaa*a*}  \begin{split} a_z a_w a^*_{w'} a^*_{z'} = \; &\delta (z-z') \delta (w-w') + \delta (z-w') \delta (w-z') + \delta (z-z') a_{w'}^* a_w \\ &+ \delta (w-w') a_{z'}^* a_z + \delta (z-w') a_{z'}^* a_w + \delta (w-z') a_{w'}^* a_z + a_{z'}^* a_{w'}^* a_w a_z\,. \end{split} \end{equation} 
Consider first the fully contracted contribution $A^{(8,1,1)}_{x,y}$ associated with $\delta (z-z') \delta (w-w')$. With $\| J(x) \| \leq 1$, we can bound it by 
\[ \begin{split} \int_{|x-y| \leq R} &dx dy \| A^{(8,1,1)}_{x,y} \Psi \|^2 \\ &\leq  \int dx dy dz dw \, |k(x-z)|^2  |k(y-w)|^2 \o^2_\ell (x-w)   \| \Psi \|^2  \lesssim \rho^3 \ell L^3 \lesssim \rho^2 L^3 
%\\ &\lesssim \rho^4 \int dx dy dz dw \frac{\chi (|x-z| \leq R)}{|x-z|^2} \frac{\chi (|y-w| \leq R)}{|y-w|^2} 
%\chi (|x-w| \leq C) \leq  \rho^3 L^3 
\end{split} \]
if $\delta > 0$ is small enough. The second fully contracted contribution can be handled similarly (using Cauchy-Schwarz). Let us now consider the contribution $A^{(8,1,7)}_{x,y}$ arising from the normally ordered term $a_{z'}^* a_{w'}^* a_w a_z$. Also here, we apply Cauchy-Schwarz. We find, using (\ref{eq:decay}) to cut the integrals to $|y-w|, |y-w'|, |x-z|, |x-z'| \leq R$ and applying Lemma \ref{lm:a-bds},  
\[ \begin{split} 
\int_{|x-y| \leq R} &\| A^{(8,1,7)}_{x,y} \Psi \|^2 \\ \lesssim \; &\rho^2 L^3 + \int dx dy dz dw dz'dw' \, |k(x-z)|^2  |k(y-w)|^2 \o^2_\ell(x-w) \\ & \hspace{.3cm} \times \chi (|x-y| \leq R) \chi (|y-w'| \leq R) \chi (|x-z'| <R)  \| a_{z'} a_{w'} J(x) J(y) J(z) J(w) \Psi \|^2  \\ \lesssim \; &\rho^2 L^3 + \rho^3 \ell R^3 \int_{|z'-w'| \leq 3R} dz' dw' \langle \Psi, a^*_{z'} a^*_{w'} a_{w'} a_{z'} \Psi \rangle \leq \rho^{2-24\eps-\delta}  L^3 \,.\end{split} \]
The contribution of the other, partially contracted, terms in (\ref{eq:aaa*a*}) can be handled similarly. 

Finally, we consider the term $A^{(9)}_{x,y}$, on the r.h.s. of (\ref{eq:fxfy}). 
As usual, we write $(J(x) J(z) - 1) = (J(x) -1) + J(x) (J(z) -1)$, $(J(y) J(w) -1) = (J(y) -1) + J(y) (J(w) -1)$ and we focus, for example, on the contribution $A^{(9,1)}_{x,y}$ arising from $(J(x) -1) (J(y) -1)$. We find 
\[ \begin{split} \int_{|x-y| \leq R} dxdy \, \| A_{x,y}^{(9,1)} \Psi \|^2 \leq \; &\int dx dy dz dz' dw dw' \, k(x-z) k(x-z') k(y-w) k(y-w') \\ &\times  \langle (J(x) -1) (J(y) -1) \Psi , a_z a_w a_{w'}^* a_{z'}^* (J(x) -1) (J(y) -1) \Psi \rangle \,.\end{split} \]
We use again (\ref{eq:aaa*a*}). The fully contracted contribution $A^{(9,1,1)}_{x,y}$ associated with $\delta (z-z') \delta (w-w')$ is given by 
\[ \begin{split} \int_{|x-y| \leq R} dxdy \, &\| A_{x,y}^{(9,1,1)} \Psi \|^2 \\ \leq \; &\int_{|x-y| < R} dx dy dz dw \, |k(x-z)|^2  |k(y-w)|^2 \| (J(x) -1) (J(y) -1) \Psi \|^2 \\ \lesssim \; &\rho^3 \int_{|x-y| \leq R} dx dy dv du \, \o_\ell (x-v) \o_\ell (y-u) \langle \Psi, a_v^* a_u^* a_v a_u \Psi \rangle \\ \lesssim \; &\rho^3\ell^4 \int_{|u-v| \leq 2R}  dv du  \langle \Psi, a_v^* a_u^* a_v a_u \Psi \rangle \lesssim \rho^{2} L^3 \,. \end{split} \]
The other fully contracted contribution can be bounded analogously. To bound the contribution $A^{(9,1,7)}_{x,y}$ proportional to $a_{z'}^* a_{w'}^* a_w a_z$, we apply Cauchy-Schwarz. We find, using (\ref{eq:decay}) to restrict to $|x-z|, |x-z'|, |y-w|, |y-w'| \leq R$, 
\begin{equation}\label{eq:A917} \begin{split}   &\int_{|x-y| \leq R} dx dy \, \| A_{x,y}^{(9,1,7)} \Psi \|^2  \\  &\lesssim \rho^2 L^3 + \int_{|x-y| < R} dx dy dz dz' dw dw' \, |k(x-z)|^2  |k(y-w)|^2 \\ &\hspace{.8cm} \times  \chi (|x-z'| < R) \chi (|y-w'| <R) \| a_{z'} a_{w'} (J(x) -1) (J(y) -1) \Psi \|^2 \\  &\lesssim \rho^2 L^3 +  \rho^3 \int_{|x-y| < R}  dx dy dz' dw'  \, \chi (|x-z'| < R) \chi (|y-w'| <R) \\ &\hspace{6cm} \times  \| a_{z'} a_{w'} (J(x) -1) (J(y) -1) \Psi \|^2 \,. \end{split} \end{equation} 
We write 
\[ \begin{split} a_{z'} a_{w'} (J(x) - 1) &= (f_\ell(x-z') f_\ell(x-w') J(x) - 1) a_{z'} a_{w'} \\ &= (f_\ell (x-z') f_\ell (x-w') -1) J(x) a_{z'} a_{w'} + (J(x) -1)  a_{z'} a_{w'} \end{split} \]
and we proceed similarly to pass $a_{z'} a_{w'}$ through $(J(y) -1)$. Let us analyze the contribution $A^{(9,1,7,1)}_{x,y}$ to the integral on the r.h.s. of (\ref{eq:A917}) arising from the term proportional to $(J(x) -1) (J(y) -1)$; factors like $(f_\ell (x-z') f_\ell (x-w') -1)$ force $|x-z'| \leq \ell$ or $|x-w'| \leq \ell$ and are easier to control.  We have,  by (\ref{eq:JxJy-}) and Lemma \ref{lm:a-bds},  
\[ \begin{split}  &\int_{|x-y| \leq R} dx dy \, \| A_{x,y}^{(9,1,7,1)} \Psi \|^2  \\  &\lesssim \rho^3 \int_{|x-y| < R}  dx dy dz' dw'  \, \chi (|x-z'| < R) \chi (|y-w'| <R) \| (J(x) -1) (J(y) -1) a_{z'} a_{w'}  \Psi \|^2 \\ &\lesssim \rho^3 \ell^4 \int_{|z'-w'|, |z'- u|, |z'-v| < 4R} dz'dw' du dv \langle \Psi, a_u^* a_v^* a_{z'}^* a_{w'}^* a_{w'} a_{z'} a_v a_u \Psi \rangle \\ &\hspace{.4cm} +\rho^3 \ell^4 \int_{|z'-u|, |w'- u| < 2R} dz'dw' du \langle \Psi, a_u^* a_{z'}^* a_{w'}^* a_{w'} a_{z'} a_u \Psi \rangle  \\ &\lesssim  \rho^{2} L^3 
\end{split} \]
if $\eps , \delta > 0$ are small enough. This concludes the proof of the second bound in (\ref{eq:goalsc}) and thus of (\ref{eq:cccc-claim}). 

To show (\ref{eq:cc-gh}), we expand  
\[ \begin{split} \int dx \, &\langle \Psi, c^* (\tau_x) c(\tau_x) \Psi \rangle \\ &= 
\int dx dz dz' \, \tau (x-z) \tau (x-z') \, \langle \Psi, c^*_z c_{z'} \Psi \rangle = \int dz dz' (\tau * \tau) (z-z') \langle \Psi, c^*_z c_{z'} \Psi \rangle \,. \end{split} \]
From Lemma \ref{lm:eta}, we have $\| \tau *\tau \|_1 \leq \| \tau \|_1^2 \lesssim \rho^{-6\eps}$ and therefore, with (\ref{eq:cc-claim}), 
\[  \int dx \langle \Psi, c^* (\tau_x) c(\tau_x) \Psi \rangle \lesssim \rho^{-6\eps} \int dz \, \langle \Psi, c_z^* c_z \Psi \rangle \lesssim \rho^{2-22\eps-2\delta} L^3\,. \]
As for (\ref{eq:cccc-gh}), we proceed similarly, writing 
\[ \begin{split} \int_{|x-y| \leq R}  dx dy  \, &\langle \Psi, c^* (\tau_x) c^* (\ph_y) c (\ph_y) c(\tau_x) \Psi \rangle \\ &= \int_{|x-y| \leq R} dx dy dz dz' \, \ph (y-z) \ph (y-z') \, \langle \Psi, c^* (\tau_x) c^*_z c_{z'} c (\tau_x)  \Psi \rangle\,.  \end{split} \]
Up to small errors, we can restrict the $z,z'$ integrals to $|y-z| , |y- z'| \leq R$ using (\ref{eq:decay}). Hence, we obtain 
\[  \begin{split}  \int_{|x-y| \leq R} &dx dy \, \langle \Psi, c^* (\tau_x) c^* (\ph_y) c (\ph_y) c(\tau_x) \Psi \rangle \\ &\lesssim \rho^2 L^3 +  \int_{|x-z| , |x-z'| \leq 2R} dx dz dz' \, (\ph * \ph) (z-z')  \, \langle \Psi, c^* (\tau_x) c^*_z c_{z'} c(\tau_x) \Psi \rangle \\ &\lesssim  \rho^2 L^3 +  \| \ph \|_1^2 \int_{|x-z| \leq 2R}  dx dz \, \langle \Psi, c^* (\tau_x) c^*_z c_z c(\tau_x) \Psi \rangle\\ &\lesssim  \rho^2 L^3 +  \rho^{-6\eps} \int_{|x-z| \leq 2R} dx dz \, \langle \Psi, c^* (\tau_x) c^*_z c_z c(\tau_x) \Psi \rangle \,.\end{split} \]
Repeating the same argument for $c^* (\tau_x) c (\tau_x)$, we arrive at 
\[ \begin{split}  \int_{|x-y| \leq R} &dx dy \, \langle \Psi, c^* (\tau_x) c^* (\ph_y) c (\ph_y) c(\tau_x) \Psi \rangle \\ &\lesssim  \rho^2 L^3 + \rho^{-12\eps} \int_{|w-z| \leq 3R} dz dw \, \langle \Psi, c_z^*c_w^* c_w c_z \Psi \rangle \lesssim \rho^{2-42\eps -\delta} L^3 \end{split} \]
by (\ref{eq:cccc-claim}). 
\end{proof}

Lemma \ref{lm:c's} can be used to derive improved bounds on the number of particles and on the number of excitations localized in balls with size of the order $\ell \simeq \rho^{-\delta}$ (which are much smaller than the balls of radius $R \simeq \rho^{-1/2-3\eps}$ considered in Lemma \ref{lm:a-bds}), which will be very useful in the next section, to estimate the energy of $\Psi$. 
\begin{lemma} \label{lm:aaaaC} 
For $0 < \delta < \eps$, $\eps > 0$ small enough and $\rho > 0$ small enough, we have 
\begin{equation}\label{eq:aaaaC} \int_{|x-y| < C \ell} dx dy \, \langle \psi, b_x^* b_y^* b_y b_x \psi \rangle , \; \int_{|x-y| < C \ell} dx dy \, \langle \psi, a_x^* a_y^* a_y a_x \psi \rangle  \lesssim \rho^{2-36\eps-\delta} L^3\,. \end{equation} 
Moreover, recalling $R = C \rho^{-1/2-3\eps}$, we find  
\begin{equation}\label{eq:aaaCR} \int_{|x-y| \leq C \ell, |x-z| \leq R} dx dy dz \langle \Psi, a_x^* a_y^* a_z^* a_z a_y a_x \Psi \rangle \lesssim \rho^{3/2-45\eps-\delta} L^3\,. \end{equation}
\end{lemma} 
\begin{proof} 
We start with (\ref{eq:aaaaC}). Recalling the definition \eqref{eq:bb}, we can write 
\begin{equation}\label{eq:atob} a_z^* a_z = (b_z^* + \sqrt{\rho_0}) ( b_z + \sqrt{\rho_0}) \lesssim b_z^* b_z + \rho\,. \end{equation} 
With $\ell = C \rho^{-\delta}$, we find 
\begin{equation}\label{eq:4a} \int_{|x-y| < C\ell} dx dy \, a_x^* a_y^* a_y a_x  \lesssim \rho^{2-3\delta} L^3 + \rho^{1-3\delta} \int dx \, b_x^* b_x + \int_{|x-y| < C\ell} dx dy \, b_x^* b_y^* b_y b_x \,.\end{equation} 
From Lemma \ref{lm:a-bds}, we obtain
\[ \int_{|x-y| < C\ell} dx dy \, \langle \Psi, a_x^* a_y^* a_y a_x \Psi \rangle \lesssim \rho^{2-3\delta} L^3 + \int_{|x-y| < C\ell} dx dy \, \langle \Psi, b_x^* b_y^* b_y b_x\Psi \rangle  \, \]
if $\delta, \eps > 0$ are small enough.
%Inverting the relations  
%\eqref{eq:def-cc}, we write 
%\[ b_x = c (\gamma_x) + c^* %(\sigma_x)\,. \]
%Therefore, we obtain 
%\[ \int dx \, b_x^* b_x = \int dx 
%\, (c^* (\gamma_x) + c (\sigma_x)) %(c(\gamma_x) + c^* (\sigma_x))\,. \]
%With 
%\begin{equation}\label{eq:CCR-c} 
%\big[ \, c (\sigma_x) , c^* %(\sigma_x) \big] = \| \sigma_x \|%_2^2 \end{equation} 
%and using $\gamma* \gamma = 
%\mathbbm{1} + \sigma * \sigma$, we %find 
%\[ \int dx \, b_x^* b_x \lesssim 
%\int dx \, \| \sigma_x \|^2 + \int %dz \, c_z^* c_z  + \int dz dz' \, %(\sigma * \sigma) (z-z') c_z^* %c_{z'}\,. \]
%With $\| \sigma_x \|_2^2 \lesssim 
%\rho^{3/2}$ and $\| \sigma * \sigma %\|_1 \leq \| \sigma \|_1^2 \lesssim %\rho^{-\eps}$, we conclude by %Cauchy-Schwarz that 
%\begin{equation}\label{eq:2b}  \int %dx \, \langle \Psi, b_x^* b_x \Psi 
%\rangle \lesssim \rho^{3/2} L^3  + 
%\rho^{-\eps} \int dz \, \langle 
%\Psi, c_z^* c_z \Psi \rangle 
%\lesssim \rho^{3/2} L^3 
%\end{equation} 
%where, in the last step, we applied %Lemma \ref{lm:c's}. 
To control the last term, we invert (\ref{eq:def-cc}), writing  
\[ b_y = c (\gamma_y) + c^*(\sigma_y)\, , \quad b_y^* = c^*  (\gamma_y) + c (\sigma_y) \,. \]
With Cauchy-Schwarz, we find 
\[ \begin{split}  \int_{|x-y| < C\ell} dx dy \, b_x^* &b_y^* b_y b_x \\ &= \int_{|x-y| < C\ell} dx dy \, b_x^* (c^* (\gamma_y) + c (\sigma_y)) (c (\gamma_y) + c^* (\sigma_y)) b_x \\ &\lesssim \int_{|x-y|< C\ell} dx dy \, b_x^* c^* (\gamma_y) c(\gamma_y) b_x +  \int_{|x-y|< C\ell} dx dy \, b_x^* c (\sigma_y) c^* (\sigma_y) b_x\,. \end{split} \]
Recalling $\| \sigma_y \|^2_2 \lesssim \rho^{3/2}$, we obtain, applying Lemma \ref{lm:a-bds} (with $k=1$), 
\begin{equation} \label{eq:bbbb}  \begin{split}  \int_{|x-y| < C\ell} dx dy \, &\langle \Psi, b_x^* b_y^* b_y b_x \Psi \rangle \\ &\lesssim  \sum_{\tau = \gamma,\sigma} \int_{|x-y|< C\ell} dx dy \, b_x^* c^* (\tau_y) c(\tau_y) b_x + \rho^{3/2-3\delta} \int dx  \, \langle \Psi , b_x^* b_x \Psi \rangle \\ &\lesssim \rho^{3-\eps-3\delta} L^3 + \sum_{\tau = \gamma,\sigma} \int_{|x-y|< C\ell} dx dy \, b_x^* c^* (\tau_y) c(\tau_y) b_x \,.\end{split} \end{equation} 
We have  
\[ \begin{split} &\int_{|x-y| < C\ell} dx dy \, b_x^* c^* (\tau_y ) c (\tau_y) b_x \\
%= \int_{|x-y| < C\ell} dx dy \, (c^* (\gamma_x) + c (\sigma_x))  c^* (\gamma_y ) c (\gamma_y) (c 
%(\gamma_x) + c^* (\sigma_x)) 
&\lesssim 
 \int_{|x-y| < C \ell} dx dy \, c^* (\gamma_x)   c^* (\tau_y ) c (\tau_y) c (\gamma_x) + \int_{|x-y| < C\ell} dx dy \, c (\sigma_x) c^* (\tau_y ) c (\tau_y) c^* (\sigma_x) \\ 
 &\lesssim  \int_{|x-y| < C\ell} dx dy \, c^* (\gamma_x)   c^* (\tau_y ) c (\tau_y) c (\gamma_x) + \int_{|x-y| < C\ell} dx dy \, c^* (\tau_y ) c (\sigma_x)  c^* (\sigma_x) c (\tau_y) \\ &\hspace{.5cm} + \int_{|x-y| < C\ell} dx dy \, |\langle \sigma_x , \tau_y \rangle |^2 \end{split} \]
 and therefore
 \begin{equation} \label{eq:bccb} \begin{split} &\int_{|x-y| < C\ell} dx dy \, b_x^* c^* (\tau_y ) c (\tau_y) b_x
 \\ &\hspace{.3cm} \lesssim  \int_{|x-y| < C\ell} dx dy \, c^* (\gamma_x)   c^* (\tau_y ) c (\tau_y) c (\gamma_x) + \int_{|x-y| < C\ell} dx dy \, c^* (\tau_y )  c^* (\sigma_x)  c (\sigma_x) c (\tau_y) \\ &\hspace{.7cm} + \int_{|x-y| < C\ell} dx dy \, \| \sigma_x \|_2^2 \, c^* (\tau_y) c (\tau_y) +  \int_{|x-y| < C\ell} dx dy \, |\langle \sigma_x , \tau_y \rangle |^2\,. \end{split} \end{equation} 
 From Lemma \ref{lm:eta} (the $L^2$-norm of $\sigma$, restricted to $|x| \leq C \ell$, can be estimated through (\ref{eq:est-combi})), we find \[ \int_{|x-y| < C\ell} dx dy \, |\langle \sigma_x , \tau_y \rangle|^2 \lesssim \rho^{2-\delta} L^3 \]
and $\| \sigma * \sigma \|_1 \leq \| \sigma \|^2_1 \lesssim \rho^{-\eps}$. Hence 
 \[ \int_{|x-y| < C\ell} dx dy\,  \| \sigma_x \|_2^2 \, c^* (\tau_y) c (\tau_y) \lesssim \rho^{3/2-3\delta} \int dy \, c^* (\tau_y) c (\tau_y) \lesssim \rho^{3/2-\eps-3\delta} \int dz \, c^*_z c_z\,. \]
 In the quartic terms on the r.h.s. of (\ref{eq:bccb}), we can use (\ref{eq:decay}) to cutoff the kernels $\tau (z), \sigma (z), \mathbbm{1}-\gamma (z)$ to $|z| \leq R$, producing only small errors. With the estimates $\| \sigma * \sigma \|_1 \lesssim \rho^{-\eps}$ and $\| (\gamma-\mathbbm{1}) * (\gamma-\mathbbm{1}) \|_1 \lesssim \rho^{-3\eps}$, we conclude that 
 \[ \begin{split} \int_{|x-y| < C\ell} dx dy \, b_x^* c^* &(\tau_y ) c (\tau_y) b_x \\ &\lesssim \rho^{-6\eps} \int_{|z-w| < 3R} dz dw \, c_z^* c_w^* c_w c_z + \rho^{3/2-\eps-3\delta} \int dz \, c^*_z c_z + \rho^{2-\delta} L^3 \,.\end{split}  \]
Inserting in (\ref{eq:bbbb}) and using Lemma \ref{lm:c's}, we obtain 
\begin{equation}\label{eq:bbbbC} \begin{split}  \int_{|x-y| < C\ell} &dx dy \, \langle \Psi , b_x^* b_y^* b_y b_x \Psi \rangle   \\ &\lesssim \rho^{-6\eps} \int_{|z-w| \leq CR} dz dw \, \langle \Psi , c_z^* c_w^* c_w c_z \Psi \rangle  + \rho^{3/2-\eps-3\delta} \int dz \, \langle \Psi, c_z^* c_z \Psi \rangle + \rho^{2-\delta} L^3 \\ &\lesssim \rho^{2-36\eps-\delta} L^3  \,.\end{split} \end{equation}
This proves the inequality (\ref{eq:aaaaC}) for the $b$-operators and, from (\ref{eq:4a}), it also implies that 
\[ \int_{|x-y| < C\ell} dx dy \, \langle \Psi, a_x^* a_y^* a_y a_x \Psi \rangle  \lesssim \rho^{2-36\eps-\delta}  L^3\, .  \]

To show (\ref{eq:aaaCR}), we iterate (\ref{eq:atob}) to prove that 
\[ \begin{split} 
&\int_{|x-y| \leq C \ell , |x-z| \leq R} dx dy dz \, \langle \Psi, a_x^* a_y^* a_z^* a_z a_y a_x \Psi \rangle \\ &\lesssim \rho^3 \ell^3 R^3 L^3 + \rho^2 \ell^3 R^3 \int dx \langle \Psi, b_x^* b_x \Psi \rangle + \rho R^3  \int_{|x-y| \leq C\ell} dx dy \langle \Psi, b_x^* b_y^* b_y b_x \Psi \rangle \\ &\hspace{.4cm} + \rho \ell^3 \int_{|x-z| \leq R} dx dz \langle \Psi, b_x^* b_z^* b_z b_x \Psi \rangle + \int_{|x-y| \leq C \ell , |x-z| \leq R} dx dy dz \langle \Psi, b_x^* b_y^* b_z^* b_z b_y b_x \Psi \rangle\,. \end{split} \]
Using (\ref{eq:kbs}) (with $k=1,2,3$) to bound the second, fourth and fifth contributions and (\ref{eq:bbbbC}) to control the third term, we obtain 
\[ \int_{|x-y| \leq C \ell , |x-z| \leq R} dx dy dz \langle \Psi, a_x^* a_y^* a_z^* a_z a_y a_x \Psi \rangle \lesssim \rho^{3/2 - 45\eps -\delta} L^3\,. \]
\end{proof} 

We can also apply Lemma \ref{lm:c's} to improve the rough estimate (\ref{eq:rmkN}) on the expectation of the total number of particles $\cN$ and to show Prop. \ref{prop:N_on_psi}. 
\begin{proof}[Proof of Prop. \ref{prop:N_on_psi}] 
Recalling (\ref{eq:bb}), we write 
\begin{equation} \label{eq:a_to_b}
a_x = \sqrt{\rho_0} + b_x\,, \qquad a_x^* = \sqrt{\rho_0} + b_x^*\,.
\end{equation} We find 
\[ \langle\Psi,\mathcal{N} \Psi\rangle = \rho_0 L^3 +  2\sqrt{\rho_0} \, \text{Re } \int dx \langle \Psi, b_x \Psi \rangle + \int dx \langle \Psi, b_x^* b_x \Psi \rangle\,. \]
In the last term, we apply (\ref{eq:def-cc}). We obtain 
%\[ b_x = c (\gamma_x) - c^* (\sigma_x), 
%\quad b_x^* = c^* (\gamma_x) - c %(\sigma_x)\,. \]
%We obtain [\textbf{Since this comes after %the long prior proof, we can maybe quote %the definition of $b$'s and $c's$ and get %directly to this formula}]
\[ \begin{split}  \langle\Psi,\mathcal{N} \Psi\rangle =\;&(\rho_0 + \| \sigma \|_2^2) L^3 +  2 \sqrt{\rho_0} \, \text{Re } \int dx \langle \Psi, b_x \Psi \rangle \\ &+ \int dx\, \langle\Psi, \Big[ c^*(\gamma_x) c(\gamma_x)+ c^*(\sigma_x) c(\sigma_x)+c^*(\gamma_x) c^*(\sigma_x)+c(\gamma_x) c(\sigma_x)\Big]\Psi\rangle.
\end{split}\]
The error term on the first line can be bounded with Lemma \ref{lm:1b}. Contributions in the second line can be estimated with Lemma \ref{lm:c's}, in particular (\ref{eq:cc-gh}). For example, by Cauchy-Schwarz, we find 
\[ \begin{split} \int dx \langle &\Psi, c^* (\gamma_x) c^* (\sigma_x) \Psi \rangle \\ &\leq \Big[\int dx \langle \Psi, c^* (\gamma_x) c (\gamma_x) \Psi \rangle \Big]^{1/2} \Big[ \| \sigma \|_2^2 L^3 + \int  dx \langle \Psi, c^* (\sigma_x) c (\sigma_x) \Psi \rangle \Big]^{1/2} \\ &\lesssim \rho^{7/4-11\eps-\delta} L^3  \end{split}\]
if $\eps, \delta > 0$ are small enough. 
\end{proof}

\section{Computation of the energy} 
 \label{sec:energycomputation}

The goal of this Section is to show Proposition \ref{prop:Psi-energy}.  Using \eqref{eq:id} we have
\begin{equation*}
	\nabla_x a_x \Psi= \Phi_1(x)+\Phi_2(x)+\Phi_3(x)
\end{equation*}
with
\begin{equation} \label{eq:Phi_123}
	\begin{split}
		\Phi_1(x)=\;&\sqrt{\rho_0}\, J(x)\, {d}\Gamma\bigg(\frac{\nabla f_\ell(x-\cdot)}{f_\ell(x-\cdot)}\bigg) \Psi\\
		\Phi_2(x)=\;&J(x)\int dy\,\nabla \nu(x-y) \, a^*_y J(y)\Psi\\
		\Phi_3(x)=\;& {d}\Gamma\bigg(\frac{\nabla f_\ell(x-\cdot)}{f_\ell(x-\cdot)}\bigg) J(x)\int dy\,\nu(x-y) a^*_y J(y) \Psi,
	\end{split}
\end{equation}
where we recall the notation $\nu = \gamma^{-1} * \sigma$, and $d\Gamma ((\nabla f_\ell / f_\ell)_x)$ is defined as in (\ref{eq:1-J}). Hence, noticing that $\Phi_1, \Phi_2, \Phi_3$ are real valued,  
\be \label{eq:cK-split}
\begin{split}
 & \langle \Psi,\mathcal{K} \Psi\rangle \\
 &=  \sum_{j=1}^3\int \| \Phi_j(x) \|^2_2 \,dx + 2 \int \langle \Phi_1(x), \Phi_2 (x) + \Phi_3 (x) \rangle \, dx + 2 \int \langle \Phi_2 (x),  \Phi_3(x) \rangle  \,dx  \,.
 \end{split}
\ee

In the next propositions, whose proof is deferred to Subsections \ref{subsec:1-1} - \ref{subsec:2-3}, we evaluate the six terms on the r.h.s. of (\ref{eq:cK-split}). Prop. \ref{prop:Psi-energy} follows directly  combining these estimates. 
\begin{proposition} \label{prop:1-1}
    We have
    \begin{equation}
       L^{-3} \int dx \,\|\Phi_1(x)\|_2^2 \le\rho_0(\rho_0+\|\sigma\|_2^2)\|\nabla f_\ell \big\|_2^2+ C \rho^{5/2+\delta} 
    \end{equation}
    if $\eps,  \delta > 0$ are small enough. 
\end{proposition}

\begin{proposition} \label{prop:2-2}
We have
    \begin{equation}
        L^{-3}\int dx \,\|\Phi_2(x)\|_2^2 \le \big\| f_\ell \nabla\sigma\big\|_2^2 +  C \rho^{5/2+\delta}     \end{equation}
   if $\eps,  \delta > 0$ are small enough.      
\end{proposition}

\begin{proposition} \label{prop:3-3}
 We have
    \begin{equation*}
        \begin{split}
            L^{-3}\int dx\,\|\Phi_3(x)\|_2^2\le\;& \int dx\, |\sigma(x)|^2 |\nabla f_\ell(x)|^2 + \rho_0  \|\nabla f_\ell\|_2^2 \|\sigma\|_2^2 + C \rho^{5/2+\delta}        \end{split}
    \end{equation*}
     if $\eps,  \delta > 0$ are small enough.      
\end{proposition}

\begin{proposition} \label{prop:1-2}
    We have 
    \begin{equation}
    \begin{split}
        2 L^{-3}  \int dx\, \langle \Phi_1(x) , \Phi_2(x) \rangle  \le\;&  2\rho_0 \int  f_\ell(x) \nabla f_\ell (x) \cdot \nabla \s (x)  dx + C \rho^{5/2+\delta}
    \end{split}
    \end{equation}
     if $\eps,  \delta > 0$ are small enough.      
\end{proposition}

\begin{proposition} \label{prop:1-3}
We have 
    \begin{equation}
    \begin{split}
        2 L^{-3}  \int dx\,\langle \Phi_1(x) , \Phi_3(x) \rangle  \le\;&  2\rho_0 \int |\nabla f(x)|^2 \big[(\g \ast \s)(x) + (\s \ast \s)(x)\big] + C \rho^{5/2+\delta}
    \end{split}
    \end{equation}
       if $\eps,  \delta > 0$ are small enough.      
\end{proposition}

\begin{proposition} \label{prop:2-3}
We have 
    \begin{equation}
    \begin{split}
        2L^{-3}  \int dx\, \langle \Phi_2 (x), \Phi_3 (x) \rangle \le\;& 2 \int dx \, \nabla \sigma (x)  \sigma (x) \nabla f_\ell (x) f_\ell (x) + C \rho^{5/2+\delta}     \end{split}
    \end{equation}
       if $\eps,  \delta > 0$ are small enough.      
\end{proposition}

\subsection{Proof of Proposition \ref{prop:1-1}}
\label{subsec:1-1} 

Writing 
\[ \Phi_1 (x) = \sqrt{\rho_0} \int dy \frac{\nabla f_\ell (x-y)}{f_\ell (x-y)} J(x)\,  a_y^* a_y \Psi \]
and noticing that, with (\ref{eq:pull}), 
\begin{equation}\label{eq:Jsquare} a_z J(x)^2 a^*_y = f^2_\ell (x-y) f_\ell^2 (x-z) a^*_y J(x)^2 a_z + \delta (z-y) f^2_\ell (x-y) J(x)^2 \end{equation}
we obtain 
\[ \begin{split} \int dx \| \Phi_1 (x) \|^2 \leq \; &\rho_0 \int dx dy \,  |\nabla f_\ell (x-y)|^2 \langle \Psi, a_y^* a_y \Psi \rangle \\ &+ \rho_0 \int dx dy dz \, |\nabla f_\ell (x-y)| | \nabla f_\ell (x-z)| \langle \Psi, a_y^* a_z^* a_z a_y \Psi \rangle \,.\end{split}  \]
Thus
\[  \int dx \| \Phi_1 (x) \|^2 \leq \rho_0 \| \nabla f_\ell \|_2^2  \, \langle \Psi , \cN \Psi \rangle + C \rho_0 \int_{|y-z| \leq C \ell} dy dz \langle \Psi ,  a_y^* a_z^* a_z a_y \Psi \rangle\,. \]
Using \eqref{eq:cN_ub} and  \eqref{eq:aaaaC}, we conclude that 
\[  \int dx \| \Phi_1 (x) \|^2 \leq \rho_0 (\rho_0 +\|\sigma\|_2^2) L^3\|\nabla f_\ell \big\|_2^2+ C \rho^{5/2+\delta} L^3 \]
if $\eps , \delta > 0$ are small enough. 
\qed

\subsection{Proof of Proposition \ref{prop:2-2}}

Using  \eqref{eq:pull} and \eqref{eq:Jsquare} we write 
\begin{equation*}
	\begin{split}
		\int dx\,\|\Phi_2(x)\|_2^2 =\;& \int dxdydz\,\nabla\nu(x-y)\nabla\nu(x-z)\big\langle \Psi, J(y) a_y J(x)^2 a^*_{z}J(z)\Psi\rangle \\
		=\;&\int dxdy \,|\nabla\nu(x-y)|^2 f_\ell(x-y)^2 \langle \Psi, J(y)^2 J(x)^2\Psi\rangle\\
		&+\int dxdydz\,\nabla\nu(x-y)\nabla\nu(x-z) f_\ell (x-y)^2 f_\ell(x-z)^2\\
		&\qquad\qquad\qquad\times f_\ell(y-z)^2 \langle \Psi, a^*_z J(x)^2 J(y) J(z) a_y\Psi\rangle\\
		\le\;&\big\| f_\ell \nabla\nu\big\|_2^2\, L^3+F_1+F_2
	\end{split}
\end{equation*}
with
\begin{equation*}
	\begin{split}
		F_1=\;&\int dydz\,\big((f^2_\ell\nabla\nu)*(f^2_\ell\nabla\nu)\big)(y-z)  f^2_\ell(y-z) \langle \Psi, a^*_z a_y\Psi\rangle\\
		F_2=\;&\int dxdydz\,(f_\ell^2\nabla\nu)(x-y)(f_\ell^2\nabla\nu)(x-z) f^2_\ell(y-z) \, \big\langle \Psi, a^*_z \big(J(x)^2 J(y) J(z)-1\big) a_y\Psi\big\rangle.
	\end{split}
\end{equation*}
We analyze the two terms separately. With \eqref{eq:def-cc} we have  
\begin{equation} \label{eq:a_to_c}
		\begin{split}
			a_x= c(\gamma_x)+c^*(\sigma_x)+\sqrt{\rho_0}\,,\qquad a_x^*= c^*(\gamma_x)+c(\sigma_x)+\sqrt{\rho_0}\,.
		\end{split}
	\end{equation}
Inserting into $F_1$ yields
\begin{equation*}
	\begin{split}
		F_1=F_{11}+F_{12}+F_{13}+F_{14}
	\end{split}
\end{equation*} 
with  
\begin{equation*}
	\begin{split}
		F_{11}=\;&\rho_0\, L^3\int dy\,\big((f^2_\ell\nabla\nu)*(f^2_\ell\nabla\nu)\big)(y) f_\ell(y)^2\\
		F_{12}=\;&2\sqrt{\rho_0} \int dydz\,\big((f^2_\ell\nabla\nu)*(f^2_\ell\nabla\nu)\big)(y-z) f_\ell(y-z)^2 \langle \Psi, (a_z-\sqrt{\rho_0})\Psi\rangle \\
		F_{13}=\;&L^3\int dy\,\big((f^2_\ell\nabla\nu)*(f^2_\ell\nabla\nu)\big)(y) f_\ell(y)^2 (\sigma*\sigma)(y)\\
		F_{14}=\;& \int dydz \big((f^2_\ell\nabla\nu)*(f^2_\ell\nabla\nu)\big)(y-z) f^2_\ell(y-z)\\
		&\qquad\quad\times \big\langle \Psi,  \big[ c^*(\gamma_z) c(\gamma_y)+c^*(\sigma_y)c(\sigma_z) + c^*(\gamma_z) c^*(\sigma_y)+c(\sigma_z)c(\gamma_y) \big] \Psi \big\rangle .
	\end{split}
\end{equation*}
Writing $f_\ell=1-\o_\ell$ and using $\int \nabla\nu (y) dy =0$, we find 
\begin{equation}\label{eq:F11} |F_{11}| \leq C \rho L^3 \| \o_\ell \|_1 \|  \nabla \nu \|_2^2 + C \rho L^3 \| \o_\ell \nabla \nu \|_1^2 \leq C  \rho^{5/2+ \delta} L^3 \end{equation} 
if $\eps , \delta > 0$ are small enough; here we estimated $\| \o_\ell \nabla \nu \|_1 \lesssim \| \o_\ell \|_1 \| \nabla \nu \|_\infty \lesssim \rho$ and $\| \nabla \nu \|_2^2 \lesssim \rho^2$, from Lemma \ref{lm:eta}. Similarly, using again Lemma \ref{lm:eta} to bound $\| \nabla \nu \|_1 \lesssim \rho^{1/2 - 5 \eps}$, we obtain   
\begin{equation*}
	\begin{split}
		\Big|F_{13}-L^3\big(\nabla\nu&*\nabla\nu*\sigma*\sigma\big)(0)\Big| \\ &\leq L^3 \int dy dz \, |\nabla \nu (y-z)| |\nabla \nu (z)| | (\sigma * \sigma )(y)| \big( |\o_\ell (y)| + |\o_\ell (z)| + |\o_\ell (y-z)| \big) \\ &\leq  3 \| \sigma * \sigma \|_\infty \| \nabla \nu \|_1 \| \nabla \nu \|_\infty \| \o_\ell \|_1 L^3  \leq C \rho^{5/2+\delta} L^3 \end{split} \]
if $\eps, \delta > 0$ are small enough. Combining (\ref{eq:F11}) with \eqref{eq:1bb}, we immediately obtain $|F_{12}| \leq C \rho^{5/2+\delta} L^3$, if $\eps, \delta > 0$ are small enough. 

In order to estimate $F_{14}$, in particular the off-diagonal contribution proportional to $c (\sigma_z) c (\gamma_y) + c^* (\gamma_y)c^* (\sigma_z)$, we need to restrict the integral to $|y-z| \leq R$, with $R = \rho^{-1/2-3\eps}$ as in Lemma \ref{lm:Nm}. To this end, we are going to control the contribution from $|y-z| > R$, using the fast decay of $\nabla \nu (x)$, for $|x| \geq CR$. From Lemma \ref{lm:eta} (in particular, with  (\ref{eq:decay})) and keeping in mind that $\sigma (x) = 0$ for $|x| > C \ell_0$, we have 
\begin{equation}\label{eq:nnu-decay} \begin{split}  |\nabla \nu (x)| &= |(\gamma^{-1} * \nabla \sigma) (x)| \\ &\leq |\nabla \sigma (x)| + |((\gamma^{-1} - \mathbbm{1}) * \nabla \sigma ) (x)| \\ &\lesssim \frac{\rho}{\ell} \| \nabla \sigma \|_1 \frac{1}{(\rho^{1/2+3\eps/2} |x|)^m} \lesssim  \frac{\rho^{3/2-\eps+\delta} |\log \rho |}{(\rho^{1/2+3\eps/2}|x|)^m} \end{split} \end{equation}
for all $|x| > R/2$, $m \in \bN$. Hence
\[ \big| (f_\ell^2 \nabla \nu) * (f_\ell^2 \nabla \nu) (y-z) \big| \lesssim \| \nabla \nu \|_1  \frac{\rho^{3/2-\eps+\delta} |\log \rho|}{(\rho^{1/2+3\eps/2} |y-z|)^m}\leq \frac{\rho^{2-5\eps}}{(\rho^{1/2+3\eps/2} |y-z|)^m}\]
for all $|y-z| > R$, $m \in \bN$. We conclude that 
\[\begin{split}  \int_{|y-z|>R} dy dz \, \big|  &(f_\ell^2 \nabla \nu) * (f_\ell^2 \nabla \nu) (y-z)\big| \big| \langle \Psi, c (\sigma_z) c (\gamma_y) \Psi \rangle \big| \\&\lesssim \rho^{2-5\eps} \int dy dz \frac{\chi ( |y-z| >R)}{(\rho^{1/2+3\eps/2} |y-z|)^m} \| c^* (\sigma_z) \Psi \| \| c (\gamma_y) \Psi \| \, .  \end{split} \]
With $\| c^* (\sigma_z) \Psi \| \leq \| \sigma \|_2 + \| c (\sigma_z) \Psi \|$ and with Cauchy-Schwarz, we can estimate
\[ \begin{split}  \int_{|y-z|>R} dy dz \, \big|  &(f_\ell^2 \nabla \nu) * (f_\ell^2 \nabla \nu) (y-z)\big| \big| \langle \Psi, c (\sigma_z) c (\gamma_y) \Psi \rangle \big| \\\lesssim \; &\rho^{2-5\eps} \int dy dz \frac{\chi ( |y-z| >R)}{(\rho^{1/2+3\eps/2} |y-z|)^m} \Big[ \rho^{3/2} + \| c (\gamma_y) \Psi \|^2 + \| c (\sigma_z) \Psi \|^2 \Big] \\ 
\lesssim \; &\rho^{1/2 -5\eps +3\eps/2 (m-6)} \int dy \big[ \rho^{3/2} + \| c (\gamma_y) \Psi \|^2 + \| c (\sigma_y) \Psi \|^2 \big]   \\ \lesssim \; &\rho^{5/2 -14\eps +3m\eps/2} L^3 \lesssim \rho^{5/2+\delta} L^3 
   \end{split} \]
if $\eps, \delta > 0$ are small and $m \in \bN$ is large enough. In the last step, we applied (\ref{eq:cc-gh}) in Lemma \ref{lm:c's}. Thus, applying Cauchy-Schwarz to handle the diagonal terms $c^* (\gamma_z) c (\gamma_y)$ and $c^* (\sigma_y) c (\sigma_z)$, we can bound 
\begin{equation*}
    \begin{split}
        |F_{14}|\le\;& \|\nabla\nu\|_1^2 \int dx \Big[\langle \Psi,c^*_x c_x \rangle\Psi+\langle \Psi,c^*(\sigma_x) c(\sigma_x) \rangle\Psi+\langle \Psi,c^*((\gamma-\mathbbm{1})_x) c(\gamma-\mathbbm{1})_x) \rangle\Psi\Big]\\
        &+L^{3/2} \| f_\ell^2 \nabla\nu* f_\ell^2\nabla\nu\|_2\bigg( \int_{|y-z|\le C R} \, dy dz \, \|c(\gamma_z)c(\sigma_y)\Psi\|^2\bigg)^{1/2} + C \rho^{5/2+\delta} L^3 \\
        \le\;& C \rho^{5/2+\delta} L^3 
    \end{split}
\end{equation*}
if $\eps, \delta > 0$ are small enough. Here we used Lemma \ref{lm:eta} to estimate $\| \nabla \nu \|_1^2 \leq C\rho^{1-10\eps}$ and $\|  f_\ell^2 \nabla\nu* f_\ell^2\nabla\nu\|_2 \leq C \| \nabla \nu \|_{4/3}^2 \leq C \rho^{7/4-8\eps}$ and we applied \eqref{eq:cccc-gh}. This shows that
\begin{equation}
    \Big|F_1- \big(\nabla\nu*\nabla\nu*\sigma*\sigma\big)(0) L^3\Big| \le C\rho^{5/2+\delta} L^3 \end{equation}
if $\eps, \delta > 0$ are small enough. To bound $F_2$, we apply Cauchy-Schwarz, we decompose $a_y = \sqrt{\rho_0} + b_y, a_y^* = \sqrt{\rho_0} + b^*_y$ as in (\ref{eq:a_to_b}) and we use  
\begin{equation*}
    0\le 1-J(x)^2 J(y)J(z) \le d\Gamma(2(\o_\ell)_x+(\o_\ell)_y+(\o_\ell)_z) \,.
\end{equation*}
We find
\begin{equation*}
\begin{split} |F_2| &\le \int dx dy dz \, |\nabla \nu (x-y)| |\nabla \nu (x-z)| \\ &\hspace{2cm} \times \Big[ \rho \, \langle \Psi, (1- J^2 (x) J(y) J(z) )\Psi \rangle + \langle \Psi, b_y^* (1- J^2 (x) J(y) J(z)) b_y \Psi \rangle \Big] \\ &\le \rho \int dx dy dz dw \, |\nabla \nu (x-y)| |\nabla \nu (x-z)| \\ &\hspace{4cm} \times \big( |\o_\ell (w-x)| + |\o_\ell (w-y)| + |\o_\ell (w-z)| \big) \langle \Psi, a_w^* a_w \Psi \rangle \\ &\hspace{.4cm} + \int dx dy dz dw \, |\nabla \nu (x-y)| |\nabla \nu (x-z)| \\ &\hspace{4cm} \times \big( |\o_\ell (w-x)| + |\o_\ell (w-y)| + |\o_\ell (w-z)| \big) \langle \Psi, b_y^* a_w^* a_w b_y \Psi \rangle \\ &=: F_{21} + F_{22} \,.
\end{split} 
\end{equation*}
With  Lemma \ref{lm:eta} and (\ref{eq:cN_ub}), we have 
\[ F_{21} \leq C \rho \| \nabla \nu \|_1^2 \| \o_\ell \|_1 \langle \Psi, \cN \Psi \rangle \leq C \rho^{5/2+\delta} L^3 \]
if $\eps, \delta > 0$ are small enough. Writing $a_w = \sqrt{\rho_0} + b_w, a_w^* = \sqrt{\rho_0} + b_w^*$, we obtain 
\[ \begin{split} F_{22} \leq \; &\rho \int dx dy dz dw \, |\nabla \nu (x-y)| |\nabla \nu (x-z)| \\ &\hspace{3cm} \times \big( |\o_\ell (w-x)| + |\o_\ell (w-y)| + |\o_\ell (w-z)| \big) \langle \Psi, b_y^*  b_y \Psi \rangle \\ &+ \int dx dy dz dw \,  |\nabla \nu (x-y)| |\nabla \nu (x-z)| \\ &\hspace{3cm} \times \big( |\o_\ell (w-x)| + |\o_\ell (w-y)| + |\o_\ell (w-z)| \big) \langle \Psi, b_y^*  b_w^* b_w b_y \Psi \rangle \,.
\end{split} \]
In the second term, the contribution proportional to $\omega_\ell (y-w)$ is automatically restricted to $|y-w| \lesssim \ell$. The contributions proportional to $\omega_\ell (w-x), \omega_\ell (w-z)$, on the other hand, can be restricted to $|w-y| \lesssim R$, up to a negligible error, using the fast decay (\ref{eq:nnu-decay}). We find 
\[\begin{split} 
F_{22} \leq \; &C \rho \| \nabla \nu \|_1^2 \| \o_\ell \|_1 \int dy \langle \Psi, b_y^* b_y \Psi \rangle + C \| \nabla \nu \|_1^2 \int_{|y-w| \leq 2\ell} dy dz \, \langle \Psi, b_y^*  b_w^* b_w b_y \Psi \rangle \\ &+ C \| \o_\ell \|_1 \| \nabla \nu \|_\infty \| \nabla \nu \|_1 \int_{|y-w| \leq R} dy dz \, \langle \Psi, b_y^*  b_w^* b_w b_y \Psi \rangle + C \rho^{5/2+\eps} L^3 \\ \leq \; &C \rho^{5/2+\eps}L^3 \end{split} \]
if $\eps, \delta > 0$ are small enough, from Lemma \ref{lm:eta}, Lemma \ref{lm:a-bds} and Lemma \ref{lm:aaaaC}.

It follows that
\begin{equation}\label{eq:phi2-2f}  \int dx\,\|\Phi_2(x)\|_2^2 \le \big\| f_\ell \nabla\nu\big\|_2^2L^3+ \big(\nabla\nu*\nabla\nu*\sigma*\sigma\big)(0) L^3 + C \rho^{5/2+\delta} L^3 \,.\end{equation} 
To conclude the proof of the proposition, we observe that (writing $\nu = \sigma + (\gamma^{-1} -\mathbbm{1})*\sigma$) 
\begin{equation*}
    \begin{split}
        \int dx&\,(f_\ell^2(x)-1) |\nabla\nu(x)|^2\\
        =\;& \int dx\,(f_\ell^2(x)-1) |\nabla\sigma(x)|^2 + \int dx\,(f_\ell^2(x)-1) ((\gamma^{-1}-\mathbbm{1} )* \nabla \sigma )(x)\nabla\nu(x) \\
        &+\int dx\,(f_\ell^2(x)-1) \nabla \sigma(x) ((\gamma^{-1}-\mathbbm{1}) * \nabla \sigma) (x) \\
        \le\;&\int dx\,(f_\ell^2(x)-1) |\nabla\sigma(x)|^2 + C \| \o_\ell \|_1 \| (\gamma^{-1} - \mathbbm{1})*\nabla \sigma \|_\infty \big[ \| \nabla \nu \|_\infty + \| \nabla \sigma \|_\infty \big] \\
        \le\;&\int dx\,(f_\ell^2(x)-1) |\nabla\sigma(x)|^2 + C \rho^{5/2+\delta} 
    \end{split}
\end{equation*}
if $\eps, \delta > 0$ are small enough (because, from Lemma \ref{lm:eta}, $\| (\gamma^{-1} - \mathbbm{1}) * \nabla \sigma \|_\infty \leq \| \gamma^{-1} - \mathbbm{1} \|_2 \| \nabla \sigma \|_2 \leq C \rho^{7/4}$, $\| \nabla \nu \|_\infty , \| \nabla \sigma \|_\infty \leq C \rho /\ell^2$). Noticing that 
\[ \begin{split}  \| \nabla \nu \|^2_2 &+ (\nabla\nu*\nabla\nu*\sigma*\sigma\big)(0) \\ &= \frac{1}{|\Lambda|} \sum_{k \in \Lambda_+^*} k^2 |\widehat{\nu}_k|^2 (1 + |\widehat{\sigma}_k|^2) = \frac{1}{|\Lambda|} \sum_{k \in \Lambda_+^*} k^2 |\widehat{\nu}_k|^2 |\widehat{\gamma}_k|^2 =  \frac{1}{|\Lambda|} \sum_{k \in \Lambda^*_+} k^2 |\widehat{\sigma}_k|^2 = \| \nabla \sigma \|_2^2 \end{split} \]
we arrive, from (\ref{eq:phi2-2f}), at
\[ \int dx \| \Phi_2 (x) \|^2 \leq \| f_\ell \nabla \sigma \|_2^2 L^3 + C \rho^{5/2+\delta} L^3\,. \]
\qed

\subsection{Proof of Proposition \ref{prop:3-3}}

Recall that 
\[ \Phi_3 (x) = \int dy dw \, \nabla f_\ell (x-w) \nu (x-y) a_w^* J(x) a_w a_y^* J(y) \Psi.  \]
This implies that 
\begin{equation}\label{eq:phi3-3i}  \begin{split} \int dx \, \| &\Phi_3 (x) \|^2 \\ = \; &\int dx dy dz dw dt \, \nabla f_\ell (x-w) \cdot \nabla f_\ell (x-t) \nu (x-y) \nu (x-z) \\ &\hspace{4cm} \times  \langle a_w a_y^* J(y) \Psi ,  J(x) a_w a_t^* J(x) a_t a_z^* J(z) \Psi \rangle \\ = \; &\int dx dy dz dw \, |\nabla f_\ell (x-w)|^2 \nu (x-y) \nu (x-z) \langle a_w a_y^* J(y) \Psi ,  J^2 (x)  a_w a_z^* J(z) \Psi \rangle \\ &+ \int dx dy dz dw dt \, \nabla f_\ell (x-w) \cdot \nabla f_\ell (x-t) \nu (x-y) \nu (x-z) \\ &\hspace{4cm} \times \langle a_w a_y^* J(y) \Psi ,  J(x) a_t^* a_w  J(x) a_t a_z^* J(z) \Psi \rangle  \\ \leq \; &G_1 + G_2 \end{split} \end{equation} 
with 
\begin{equation} \label{eq:G_1_G_2}
    \begin{split}
        G_1=\;&\int dxdydzdw \, |\nabla f_\ell (x-w)|^2 \nu(x-y) \nu(x-z)  \big\langle \Psi, J(y) a_y a_w^* a_w   a^*_z J(z) \Psi\big\rangle\\
        G_2=\;& \int dxdydzdtdw\,|\nabla f_\ell(x-t)| |\nabla f_\ell(x-w)| \nu(x-y) \nu(x-z) \\
        &\hspace{6cm} \times\big\langle \Psi, J(y) a_y a^*_ta^*_w a_wa_t a^*_z J(z)\Psi\big\rangle\,.
    \end{split}
\end{equation}
Rearranging in normal order 
\[ a_y a_w^* a_w a_z^* = a_w^* a_z^* a_y a_w + \delta (w-z) a_z^* a_y + \delta (w-y) a_z^* a_y + \delta (z-y) a_w^* a_w + \delta (w-y) \delta (y-z) \]
we can write 
\begin{equation} \label{eq:G_1_split_part}
    G_1 \leq  
     L^3 \int dx\, |\nu(x)|^2 |\nabla f_\ell(x)|^2 +G_{11}+G_{12}+G_{13}
\end{equation}
with
\begin{equation*}
    \begin{split}
        G_{11}=\;& \int dxdydw\, |\nu(x-y)|^2\, |\nabla f_\ell(x-w)|^2 \langle \Psi, J(y) a^*_w a_w J(y)\Psi\rangle\\
        G_{12}=\;& 2 \int dxdydz\,\nu(x-y) \nu(x-z) |\nabla f_\ell(x-y)|^2 \langle \Psi, J(y) a^*_z a_y J(z)\Psi\rangle\\
        G_{13}=\;& \int dxdydzdw\,\nu(x-y) \nu(x-z) |\nabla f_\ell(x-w)|^2 \langle \Psi, J(y) a^*_z a^*_w a_w a_y J(z)\Psi\rangle.
    \end{split}
\end{equation*}
Observing with Lemma \ref{lm:fell} and Lemma \ref{lm:eta} that 
\begin{equation*}
\begin{split}
    \int dx\,\big|&|\nu(x)|^2-|\sigma(x)|^2\big| |\nabla f_\ell(x)|^2\\
    \le\;& C \big(\|\sigma\|_\infty+\|\nu\|_\infty\big) \big\|(\gamma^{-1}-\delta)*\sigma\big\|_\infty \|\nabla f_\ell\|_2^2  \le C \rho \ell^{-1} \|\gamma^{-1}-\delta\|_2\, \|\sigma\|_2 \leq C \rho^{5/2+\delta} 
\end{split}
\end{equation*}
we find 
\begin{equation}\label{eq:G1-2} G_1 \leq  L^3 \int dx\, |\sigma (x)|^2 |\nabla f_\ell(x)|^2 +G_{11}+G_{12}+G_{13} + C \rho^{5/2+\delta}\,. \end{equation} 
Next, we estimate the terms $G_{11}, G_{12}, G_{13}$. With (\ref{eq:cN_ub}), we easily find 
\begin{equation}\label{eq:G11f}  G_{11} \leq \| \nu\|_2^2  \|\nabla f_\ell\|_2^2 \, \langle \Psi,\mathcal{N}\Psi\rangle \leq  \| \nu \|_2^2 \| \nabla f_\ell \|_2^2 \, \rho_0 L^3 + C \rho^3 L^3 \,.\end{equation} 
Switching to 
\[ G_{12} = 2 \int dx dy dz \, \nu (x-y) \nu (x-z) |\nabla f_\ell (x-y) |^2 f^2_\ell (y-z) \langle \Psi, a_z^* J(z) J(y) a_y \Psi \rangle \]
we write  $1-J(y) J(z) = 1- J(y) + J(y) (1- J(z))$. For the contribution containing $1-J(y)$ we restrict the integration to the region $|x-z| \leq R$, since the fast decay \eqref{eq:decay} of $\nu (x-z)$ implies that the remainder is estimated by
\begin{equation*}
    \begin{split}
        \Big|&\int_{|x-z|>R}dxdydz\,\nu(x-y) \nu(x-z) |\nabla f_\ell(x-y)|^2 f_\ell^2(y-z) \langle \Psi, a^*_z (1-J(y)) a_y \Psi\rangle\Big|\\
        \lesssim\;&\frac{\rho}{\ell}\|\nu\|_\infty  \int_{|x-z|>R}  \frac{1}{\big(\rho^{1/2+3\varepsilon/2}|x-z|)^m}|\nabla f_\ell(x-y)|^2 \big( \|a_z\Psi\|^2 + \|a_y\Psi\|^2\big)\\
        \lesssim\;& \frac{\rho}{\ell} \|\nu\|_\infty \|\nabla f_\ell\|_2^2 \langle \Psi,\mathcal{N}\Psi\rangle \int_{|x|>R} dx\frac{1}{\big(\rho^{1/2+3\varepsilon/2}|x|)^m}\le C \rho^{5/2+\delta} L^3
    \end{split}
\end{equation*}
if $\varepsilon,\delta>0$ are small enough and $m\in\mathbb{N}$ is large enough. We thus have
\[ \begin{split} \Big| \int &dx dy dz \, \nu (x-y) \nu (x-z) |\nabla f_\ell (x-y)|^2 f^2_\ell (y-z) \langle \Psi, a_z^* (1 - J(y)) a_y \Psi \rangle \Big| \\ \leq \; & \Big[ \int_{|x-z| \leq R} dx dy dz dw \, |\nu (x-y)| |\nabla f_\ell (x-y)|^2 \o_\ell (y-w) \langle \Psi, a_z^* a_w^* a_w a_z \Psi \rangle \Big]^{1/2} \\ &\hspace{.5cm} \times \Big[ \int dx dy dz dw \, |\nu (x-y)| |\nu (x-z)|^2 |\nabla f_\ell (x-y)|^2 \o_\ell (y-w) \langle \Psi, a_y^* a_w^* a_w a_y \Psi \rangle \Big]^{1/2}\\ &+C \rho^{5/2+\delta} L^3. \end{split} \]
With Lemma \ref{lm:eta}, \eqref{eq:kas} and (\ref{eq:aaaaC}), we conclude that   
\[ \begin{split} \Big| \int &dx dy dz \, \nu (x-y) \nu (x-z) |\nabla f_\ell (x-y)|^2 f^2_\ell (y-z) \langle \Psi, a_z^* (1 - J(y)) a_y \Psi \rangle \Big| \\ \leq \; &C \rho^{5/2+\delta} L^3 + C \| \nu \|_\infty \| \nu \|_2 \| \nabla f_\ell \|_2^2 \| \o_\ell \|_1^{1/2} \\ &\hspace{.5cm} \times \Big[ \int_{|z-w| \leq 2R} dz dw |\langle \Psi, a_z^* a_w^* a_w a_z \Psi \rangle \Big]^{1/2} \Big[ \int_{|w-y| \leq C \ell} dy dw \, \langle \Psi, a_y^* a_w^* a_w a_y \Psi \rangle \Big]^{1/2} \\ \leq \; & C \rho^{5/2+\delta} L^3 \end{split} \]
for $\eps, \delta > 0$ sufficiently small. The contribution proportional to $1-J(z)$ can be controlled similarly. Hence  
\[ G_{12} \leq 2 \int dx dy dz \, \nu (x-y) \nu (x-z) |\nabla f_\ell (x-y) |^2 f^2_\ell (y-z) \langle \Psi, 
a_z^* a_y \Psi \rangle + C \rho^{5/2+\delta} L^3 \,. \]
With $|f^2_\ell (y-z) - 1| \leq 2 \o_\ell (y-z)$ and Cauchy-Schwarz, we can estimate  
\[ \begin{split} 
\Big| \int dx dy dz \, \nu (x-y) \nu (x-z) |\nabla f_\ell (x-y) &|^2 (f^2_\ell (y-z)-1) \langle \Psi, a_z^* a_y \Psi \rangle \Big| \\ &\leq C \| \nu \|^2_\infty \| \nabla f_\ell \|_2^2 \| \o_\ell \|_1 \langle \Psi, \cN \Psi \rangle \leq  C \rho^{5/2+\delta} L^3  \end{split}  \]
for $\eps , \delta > 0$ small enough. Therefore, using (\ref{eq:bb}) to decompose $a_y = \sqrt{\rho_0} + b_y, a_z^* = \sqrt{\rho_0} + b_z^*$ and recalling that $\int \nu (x) dx = 0$ we arrive at
\[ G_{12} \leq 2 \int dx dy dz \, \nu (x-y) \nu (x-z) |\nabla f_\ell (x-y) |^2 \langle \Psi, b_z^*  b_y \Psi \rangle + C \rho^{5/2+\delta} L^3\,. \]
Inverting (\ref{eq:def-cc}), we obtain 
\begin{equation}\label{eq:G12f} \begin{split} G_{12} \leq \; &2 L^3 \int  dx\, \nu (x) (\nu * \sigma * \sigma) (x) |\nabla f_\ell (x)|^2  \\ 
   &+  2 \int dxdydz\,\nu(x-y) \nu(x-z) |\nabla f_\ell(x-y)|^2\\
        &\qquad\qquad\times\Big\langle \Psi, \Big[ c^*(\gamma_z) c(\gamma_y) + c^*(\sigma_y) c(\sigma_z)+ c^*(\gamma_z) c^*(\sigma_y)+c(\sigma_z) c(\gamma_y)  \Big]\Psi\Big\rangle \\
        &+C \rho^{5/2+\delta} L^3\,.
    \end{split}\end{equation} 
To bound the first term, we observe that 
\[ \nu * \sigma * \sigma = (\gamma-\mathbbm{1})*\sigma + (\mathbbm{1} - \gamma^{-1}) *\sigma \]
 and therefore $\| \nu * \sigma * \sigma \|_\infty \leq (\| \gamma -\mathbbm{1} \|_2 + \| \gamma^{-1} - \mathbbm{1} \|_2 )\| \sigma \|_2 \leq C \rho^{3/2}$. Hence
 \[ L^3 \int  \nu (x) (\nu * \sigma * \sigma) (x) |\nabla f_\ell (x)|^2 \leq C \rho^{3/2} \| \nu \|_\infty \| \nabla f_\ell \|_2^2 \, L^3 \leq C \rho^{5/2+\delta} L^3\,. \] 
As for the second term on the r.h.s. of (\ref{eq:G12f}), we apply Cauchy-Schwarz together with \eqref{eq:cc-claim},  \eqref{eq:cc-gh}, \eqref{eq:cccc-gh}. We obtain $G_{12} \leq C \rho^{5/2+\delta}$. 

Next, we bound $G_{13}$. With $1- J(y) J(z) \lesssim d\Gamma (\omega_{\ell,y} + \o_{\ell,z})$ and Cauchy-Schwarz, we find 
\begin{equation}\label{eq:G1312} \begin{split} \Big| \int &dx dy dz dw\, \nu(x-y) \nu(x-z) |\nabla f_\ell(x-w)|^2 f_\ell(y-z)^2 \\
    &\qquad\qquad\qquad\times f_\ell(y-w) f_\ell(z-w) \langle\Psi, a^*_z a^*_w (J(y)J(z)-1) a_wa_y\Psi\rangle \Big| \\ 
    &\leq \Big[ \int dx dy dz dw \, |\nu (x-y)| |\nu (x-z)| |\nabla f_\ell (x-w)|^2 \langle \Psi, a_z^* a_w^* d\Gamma (\o_{\ell,y}) a_w a_z \Psi \rangle \Big]^{1/2} \\ &\hspace{.4cm} \times \Big[ \int dx dy dz dw \, |\nu (x-y)| |\nu (x-z)| |\nabla f_\ell (x-w)|^2 \langle \Psi, a_y^* a_w^* d\Gamma (\o_{\ell,y}) a_w a_y \Psi \rangle \Big]^{1/2}\,.\end{split} \end{equation} 
 On the one hand, writing $a_z^\sharp = b_z^\sharp + \sqrt{\rho_0}, a_w^\sharp = b_w^\sharp + \sqrt{\rho_0}$, we have
 \begin{equation*}
    \begin{split}
     \int dx dy &dz dw \, |\nu (x-y)| |\nu (x-z)| |\nabla f_\ell (x-w)|^2 \langle \Psi, a_z^* a_w^* d\Gamma (\o_{\ell,y}) a_w a_z \Psi \rangle \\ 
    \leq \; &C \rho^2 \int dxdydzdw\, |\nu(x-y)| |\nu(x-z)| |\nabla f_\ell(x-w)|^2  \langle\Psi, d\Gamma(\o_{\ell,y})\Psi\rangle\\
        &+C \rho  \int dxdydzdw\, |\nu(x-y)| |\nu(x-z)| |\nabla f_\ell(x-w)|^2 \langle\Psi, b^*_z d\Gamma(\o_{\ell,y}) b_z\Psi\rangle\\
        &+C \rho  \int dxdydzdw\, |\nu(x-y)| |\nu(x-z)| |\nabla f_\ell(x-w)|^2 \langle\Psi, b^*_w d\Gamma(\o_{\ell,y}) b_w\Psi\rangle\\
        &+C   \int dxdydzdw\, |\nu(x-y)| |\nu(x-z)| |\nabla f_\ell(x-w)|^2 \langle\Psi, b^*_z b^*_w d\Gamma(\o_{\ell,y}) b_w b_z\Psi\rangle.
    \end{split}
\end{equation*}
Making use of the fast decay (\ref{eq:decay}) of the kernel $\nu$ to restrict the integrals to $|x-y|, |x-z| < R$ producing only negligible errors, we obtain, with \eqref{eq:kbs}, 
\begin{equation*}
    \begin{split}
  \int dx dy dz dw \, &|\nu (x-y)| |\nu (x-z)| |\nabla f_\ell (x-w)|^2 \langle \Psi, a_z^* a_w^* d\Gamma (\o_{\ell,y}) a_w a_z \Psi \rangle \\  \le\;& C \rho^3 L^3 + C \rho^2 \|\nu\|_1^2 \|\nabla f_\ell\|_2^2 \|\o_\ell\|_1 \langle \Psi,\mathcal{N}\Psi\rangle\\
        &+C \rho \|\nu\|_1 \|\nu\|_\infty \|\nabla f_\ell\|_2^2\|\o_\ell\|_1 \int_{|w-t|\le R} dwdt\,\langle \Psi, b^*_w a^*_t a_tb_w\Psi\rangle \\
        &+C \|\nu\|_\infty ^2 \|\nabla f_\ell\|_2^2 \|\o_\ell\|_1 \int_{\substack{|z-w|\le R \\|z-t|\le R}} dzdwdt\,\langle\Psi, b^*_z b^*_w a^*_t a_t b_zb_w \Psi\rangle\\
        \le\;& C \rho^{3-c(\eps+\delta)} L^3 
    \end{split}
\end{equation*}
for some $c> 0$.  Similarly, with (\ref{eq:aaaCR}) we find 
\[ \begin{split} 
 \int &dx dy dz dw \, |\nu (x-y)| |\nu (x-z)| |\nabla f_\ell (x-w)|^2 \langle \Psi, a_y^* a_w^* d\Gamma (\o_{\ell,y}) a_w a_y \Psi \rangle \\  = \; & \int dx dy dz dw dt \, |\nu (x-y)| |\nu (x-z)| |\nabla f_\ell (x-w)|^2 \o_\ell  (y-t)  \langle \Psi, a_y^* a_w^* a_t^* a_t a_w a_y \Psi \rangle \\ \leq\; &C \rho^{5/2} L^3 + C \| \nu \|_1 \| \nu \|_\infty \| \nabla f_\ell \|_2^2 \int_{\substack{|y-t| \leq \ell \\ |y-w| \leq R}} dy dw dt \,  \langle \Psi, a_y^* a_w^* a_t^* a_t a_w a_y \Psi \rangle \\ \leq \; & C \rho^{5/2-c (\eps + \delta)} L^3 \end{split} \]
for some $c >0$. Inserting the last two bounds in (\ref{eq:G1312}) we have  
\[ \begin{split} 
\Big| \int &dx dy dz dw\, \nu(x-y) \nu(x-z) |\nabla f_\ell(x-w)|^2 f_\ell(y-z)^2 \\
    &\qquad\qquad\qquad\times f_\ell(y-w) f_\ell(z-w) \langle\Psi, a^*_z a^*_w (J(y)J(z)-1) a_wa_y\Psi\rangle \Big| \leq C \rho^{5/2+\delta} L^3 \end{split} \]
if $\eps , \delta >0$ are small enough. Hence   
\begin{equation} \label{eq:G_13_split2}
    \begin{split}
        G_{13}\le\;&\int dxdydzdw\, \nu(x-y) \nu(x-z) |\nabla f_\ell(x-w)|^2 f_\ell(y-z)^2 \\
    &\hspace{4cm} \times f_\ell(y-w) f_\ell(z-w) \langle\Psi, a^*_z a^*_w a_wa_y\Psi\rangle\\
    &+C \rho^{5/2+\delta} L^3 \\
    =\;&\int dxdydzdw\, \nu(x-y) \nu(x-z) |\nabla f_\ell(x-w)|^2 f_\ell(y-z)^2 \\
    &\hspace{4cm} \times f_\ell(y-w) f_\ell(z-w) \langle\Psi, b^*_z a^*_w a_wb_y\Psi\rangle\\
    &+G_{131}+G_{132}+C \rho^{5/2+\delta} L^3 
    \end{split}
\end{equation} 
where now
\begin{equation*}
    \begin{split}
        G_{131}=\;&\rho_0\int dxdydzdw\, \nu(x-y) \nu(x-z) |\nabla f_\ell(x-w)|^2 f_\ell(y-z)^2 \\
        &\qquad\qquad\qquad\times f_\ell(y-w) f_\ell(z-w) \langle\Psi,  a^*_w a_w\Psi\rangle\\
        G_{132}=\;&2\rho_0^{1/2} \int dxdydzdw\, \nu(x-y) \nu(x-z) |\nabla f_\ell(x-w)|^2 f_\ell(y-z)^2 \\
        &\qquad\qquad\qquad\times f_\ell(y-w) f_\ell(z-w) \langle\Psi, a^*_w a_wb_y\Psi\rangle.
    \end{split}
\end{equation*}
Using $\int\nu (x) dx =0$ we find
\begin{equation*}
    \begin{split}
        G_{131}=\;&\rho_0\int dxdydzdw\, \nu(x-y) \nu(x-z) |\nabla f_\ell(x-w)|^2  \o_\ell(y-w) \o_\ell(z-w) \langle\Psi,  a^*_w a_w\Psi\rangle\\
        &+\rho_0\int dxdydzdw\, \nu(x-y) \nu(x-z) |\nabla f_\ell(x-w)|^2 (f_\ell(y-z)^2-1) \\
        &\hspace{7cm} \times f_\ell(y-w) f_\ell(z-w) \langle\Psi,  a^*_w a_w\Psi\rangle
    \end{split}
\end{equation*}
and therefore
\begin{equation*}
    \begin{split}
        |G_{131}|\le\;&C \rho \| \nabla f_\ell \|_2^2 \big[  \|\nu\|_\infty^2 \|\o_\ell\|_1^2 + \|\nu\|_\infty \|\nu\|_1 \|\o_\ell \|_1 \big] \langle \Psi,\mathcal{N}\Psi\rangle \leq C \rho^{5/2+\delta} L^3\,.
    \end{split}
\end{equation*}
Moreover,
\begin{equation*}
    \begin{split}
        G_{132}=\;&-2\rho_0^{1/2}\int dxdydzdw\, \nu(x-y) \nu(x-z) |\nabla f_\ell(x-w)|^2  \\ &\hspace{6cm} \times f_\ell(y-w) \o_\ell(z-w) \langle\Psi, a^*_w a_wb_y\Psi\rangle \\
        &+2\rho_0^{1/2} \int dxdydzdw\, \nu(x-y) \nu(x-z) |\nabla f_\ell(x-w)|^2 (f_\ell(y-z)^2-1) \\
        &\hspace{6cm} \times f_\ell(y-w) f_\ell(z-w) \langle\Psi, a^*_w a_wb_y\Psi\rangle\\
        \le\;& C \rho^{5/2+\delta} \\ &+ C \rho^{1/2} \|\nu\|_\infty \|\nu\|_2 \|\nabla f_\ell\|_2^2 \|\o_\ell\|_1 \langle \Psi,\mathcal{N}\Psi\rangle^{1/2} \Big( \int_{|y-w|\le CR} dydw\,\langle \Psi,b^*_y a^*_w a_w b_y\Psi\rangle\Big)^{1/2}\\
        \le\;& C \rho^{5/2+\delta} L^3
    \end{split}
\end{equation*}
where we used the decay (\ref{eq:decay}) of $\nu$ to cut the integral to $|y-w| \leq C R$ (producing only a negligible error), we estimated $a_w^* a_w \lesssim \rho + b_w^* b_w$ and we applied Lemma \ref{lm:a-bds}. 
Plugging the last two bounds into \eqref{eq:G_13_split2} yields
\begin{equation} \label{eq:G_13_split3}
\begin{split}
    G_{13} \le\;&\rho_0  \int dxdydzdw\, \nu(x-y) \nu(x-z) |\nabla f_\ell(x-w)|^2 \\ &\hspace{4cm} \times f_\ell(y-z)^2 f_\ell (y-w) f_\ell (z-w) \langle\Psi, b^*_z b_y\Psi\rangle \\
    &+G_{133}+G_{134} + C \rho^{5/2+\delta} L^3 
    \end{split}
\end{equation}
with 
\begin{equation*}
    \begin{split}
      %  G_{135}=\;&\rho_0\int dxdydzdw\, \nu(x-y) \nu(x-z) |\nabla f_\ell(x-w)|^2 f_\ell(y-z)^2 \\
      %  &\qquad\qquad\qquad\times f_\ell(y-w) f_\ell(z-w) \langle\Psi, b^*_z b_y\Psi\rangle\\
        G_{133}=\;&\int dxdydzdw\, \nu(x-y) \nu(x-z) |\nabla f_\ell(x-w)|^2 f_\ell(y-z)^2 \\
        &\qquad\qquad\qquad\times f_\ell(y-w) f_\ell(z-w) \langle\Psi, b^*_zb^*_w b_w b_y\Psi\rangle\\
        G_{134}=\;&2\rho_0^{1/2} \int dxdydzdw\, \nu(x-y) \nu(x-z) |\nabla f_\ell(x-w)|^2 f_\ell(y-z)^2 \\
        &\qquad\qquad\qquad\times f_\ell(y-w) f_\ell(z-w) \langle\Psi, b^*_z b_wb_y\Psi\rangle.
    \end{split}
\end{equation*}
Applying Lemma \ref{lm:a-bds} to estimate 
\begin{equation*}
    \begin{split}
        |G_{133}| &\le C \rho^{5/2+\delta} L^3 + \int_{|w-z| \leq R} dxdydzdw\,|\nu(x-y)|^2|\nabla f_\ell(x-w)|^2 \langle \Psi, b^*_z b^*_w b_w b_z\Psi\rangle \\ &\le C \rho^{5/2+\delta} L^3
    \end{split}
\end{equation*}
and 
\begin{equation*}
    \begin{split}
        |G_{134}| \le\; &C \rho^{5/2+\delta} L^3 \\ &+ C\rho^{1/2} \bigg(\int_{|y-w| \leq R}  dxdydzdw\, |\nu(x-z)|  |\nabla f_\ell(x-w)|^2 \langle \Psi, b^*_y b^*_w b_w b_y\Psi\rangle\bigg)^{1/2}\\
        &\times\bigg( \int dxdydzdw\,|\nu(x-y)|^2|\nu(x-z)| |\nabla f_\ell(x-w)|^2 \langle \Psi, b^*_z b_z\Psi\rangle\bigg)^{1/2}\\
        \le\; &C \rho^{5/2+\delta} L^3 + C \rho^{1/2} \|\nu\|_1 \|\nu\|_2 \|\nabla f_\ell\|_2^2 \\ &\hspace{1cm} \times \Big( \int_{|z-w| \leq R}  dz dw \langle \, \Psi, b_w^* b_z^* b_z b_w \Psi \rangle \Big)^{1/2} \Big( \int dz \langle \Psi, b_z^* b_z \Psi \rangle \Big)^{1/2} \\ \le\; & C \rho^{5/2+\delta} L^3 
            \end{split}
\end{equation*}
we arrive at
\[ \begin{split} G_{13} \le\;&\rho_0  \int dxdydzdw\, \nu(x-y) \nu(x-z) |\nabla f_\ell(x-w)|^2 \\ &\hspace{4cm} \times f_\ell(y-z)^2 f_\ell (y-w) f_\ell (z-w) \langle\Psi, b^*_z b_y\Psi\rangle \\
 &+ C \rho^{5/2+\delta} L^3\,.  \end{split} \]
Finally, we invert (\ref{eq:def-cc}). We obtain 
\[ 
    \begin{split}
        G_{13} \leq \;&\rho_0\int dxdydzdw\, \nu(x-y) \nu(x-z) |\nabla f_\ell(x-w)|^2 f_\ell(y-z)^2  \\
        &\qquad\qquad\qquad\times  f_\ell(y-w) f_\ell(z-w) (\sigma*\sigma)(y-z)\\
        &+\rho_0\int dxdydzdw\, \nu(x-y) \nu(x-z) |\nabla f_\ell(x-w)|^2 f_\ell(y-z)^2 f_\ell(y-w)  f_\ell(z-w) \\
        &\hspace{.3cm}  \times \big\langle \Psi, \big[ c^*(\gamma_z) c(\gamma_y) + c^*(\sigma_y) c(\sigma_z) + c^*(\gamma_z) c^*(\sigma_y) + c(\sigma_y) c(\gamma_z)\big]\Psi\big\rangle + C \rho^{5/2+\delta} L^3\,. \end{split} \] 
We can control the second term by Cauchy-Schwarz, together with \eqref{eq:cc-claim},  \eqref{eq:cc-gh}, \eqref{eq:cccc-gh}. As for the first term, we use 
\[ 0\leq 1- f^2_\ell (y-z) f_\ell (y-w) f_\ell (z-w)\leq C \big( \o_\ell (y-z) + \o_\ell (y-w) + \o_\ell (z-w) \big) \]
to estimate 
\[  \begin{split} 
\big| \rho_0 \int dx dy dz dw \, \nu(x-y) &\nu(x-z) |\nabla f_\ell(x-w)|^2  (\sigma*\sigma)(y-z) \\ & \hspace{3cm} \times ( f_\ell(y-z)^2 f_\ell(y-w) f_\ell(z-w) -1) \big| \\ &\leq C \rho \| \nu \|_\infty  \| \sigma * \sigma \|_\infty \| \nu \|_1 \| \nabla f_\ell \|_2^2 \| \o_\ell \|_1 L^3 \leq C \rho^{5/2+\delta} L^3 \,.\end{split} \]
We conclude that 
\[ G_{13} \leq \rho_0L^3  \| \nabla f_\ell \|_2^2 \int dx \, \nu (x) (\nu * \sigma * \sigma) (x) + C \rho^{5/2+\delta} L^3. \]
Combining this with (\ref{eq:G1-2}), (\ref{eq:G11f}) and the estimate $G_{12} \leq C \rho^{5/2+\delta} L^3$, we obtain 
\begin{equation}\label{eq:G1ff} \begin{split}  G_1 \leq \; &L^3 \int dx \, |\sigma (x)|^2 |\nabla f_\ell (x)|^2 + \rho_0L^3 \| \nu \|_2^2 \| \nabla f_\ell \|_2^2 \\ &+ \rho_0 L^3  \| \nabla f_\ell \|_2^2 \int dx \, \nu (x) (\nu * \sigma * \sigma) (x) + C \rho^{5/2+\delta} L^3  \\ =\; &L^3 \int dx \, |\sigma (x)|^2 |\nabla f_\ell (x)|^2 + \rho_0L^3 \| \sigma \|_2^2 \| \nabla f_\ell \|_2^2 + C \rho^{5/2+\delta} L^3 \end{split} \end{equation} 
where we recalled $\nu = \gamma^{-1} * \sigma$ and $\sigma * \sigma = \gamma * \gamma -1$ to  compute $(\nu * \nu * \sigma * \sigma) (0) = (\sigma * \sigma)(0) - (\nu * \nu) (0) = \| \sigma \|_2^2 - \| \nu \|_2^2$.  

We are left with estimating $G_2$ from \eqref{eq:G_1_G_2}. Rearranging 
\[ \begin{split} a_y &a_t^* a_w^* a_w a_t a_z^* \\ = \; &a_z^* a_t^* a_w^* a_w a_t a_y + (\delta (t-z) + \delta (t-y)) a_z^* a_w^* a_w a_y + (\delta (w-z)+ \delta (w-y)) a_z^* a_t^* a_t a_y \\ &+ \delta (y-z) a_t^* a_w^* a_w a_t  + (\delta (y-t) \delta (z-w) + \delta (y-w) \delta (z-t)) a_z^* a_y \\ &+ \delta (y-t) \delta (z-t) a_w^* a_w + \delta (y-w) \delta (z-w) a_t^* a_t  \end{split} \]
we write 
\begin{equation*}
    G_2=G_{21}+G_{22}+G_{23}+G_{24}+G_{25}
\end{equation*}
with
\begin{equation*}
    \begin{split}
        G_{21}=\;& \int dxdydzdtdw\,\nu(x-y) \nu(x-z) |\nabla f_\ell(x-t)| |\nabla f_\ell(x-w)|\\
        &\qquad\qquad\qquad\qquad\times \big\langle \Psi, J(y) a^*_t a^*_w a^*_z a_y a_w a_t J(z)\Psi\big\rangle\\
        G_{22}=\;&4\int dxdydzdw\,\nu(x-y) \nu(x-z) |\nabla f_\ell(x-y)| |\nabla f_\ell(x-w)|\\
        &\qquad\qquad\qquad\qquad\times \big\langle \Psi, J(y)  a^*_w a^*_z a_y a_w J(z)\Psi\big\rangle\\
        G_{23}=\;&\int dxdydtdw\,|\nu(x-y)|^2 |\nabla f_\ell(x-t)| |\nabla f_\ell(x-w)|\\
        &\qquad\qquad\qquad\qquad\times \big\langle \Psi, J(y) a^*_t a^*_w  a_w a_t J(y)\Psi\big\rangle\\
        G_{24}=\;&2\int dxdydw\,|\nu(x-y)|^2 |\nabla f_\ell(x-w)| |\nabla f_\ell(x-y)| \langle \Psi, J(y) a^*_w a_w J(y)\Psi\rangle\\
        G_{25}=\;&2 \int dxdydz \,\nu(x-y)\nu(x-z) |\nabla f_\ell(x-y)| |\nabla f_\ell(x-z)| \\
        &\qquad\qquad\qquad\qquad\times \langle\Psi,J(y) a^*_z a_y J(z)\Psi\rangle.
    \end{split}
\end{equation*} 
Using (\ref{eq:decay}) to cut the integral to $|z-w| \leq R$ and applying \eqref{eq:pull}, Cauchy-Schwarz and  \eqref{eq:aaaCR}, we obtain 
\begin{equation*}
    \begin{split}
        |G_{21}|\le\;&C \rho^{5/2+\delta} L^3 \\ &+ C \int_{\substack{|t-w|\le C\ell\\|z-w|\le R}} dxdydzdtdw\,|\nu(x-y)|^2 |\nabla f_\ell(x-t)|^2 \langle \Psi, a^*_t a^*_w a^*_z a_z a_w a_t\Psi\rangle\\
        \le\;& C \rho^{5/2+\delta} L^3 \,.
    \end{split}
\end{equation*}
% \begin{equation*}
%     \begin{split}
%         |G_{21}|\le\;& C \rho^3 \|\gamma^{-1}\nu\|_1^2 \|\nabla f_\ell\|_1^2 L^3+ C \int_{\substack{|w-z|\le\ell_0\\ |w-t|\le\ell_0}} dxdydzdtdw\,|\nu(x-y)|^2|\nabla f_\ell(x-t)|^2\\
%         &\qquad\qquad\qquad\qquad\qquad\qquad\qquad\qquad\times   \big\langle \Psi, J(y) b^*_t b^*_w b^*_z b_z b_w b_t J(z)\Psi\big\rangle\\
%         \le\;& C \rho^3 \mathfrak{a}^2 \ell^2 L^3 \textcolor{red}{(\rho\mathfrak{a}^3)^{-4\delta}}.
%     \end{split}
% \end{equation*}
Similarly, from  \eqref{eq:kas} (with $k=2$) and  \eqref{eq:aaaaC}, we find 
\begin{equation*}
    \begin{split}
        |G_{22}|\le\;& C \rho^{5/2+\delta} L^3 \\ &+ C \|\nu\|_\infty \bigg( \int_{|z-w|\le R} dxdydzdw\,|\nabla f_\ell (x-y)| |\nabla f_\ell(x-w)|^2 \langle \Psi, a^*_w a^*_z a_z a_w\Psi\rangle\bigg)^{1/2}\\
        &\qquad\times \bigg( \int_{|y-w|\le C\ell} dxdydzdw \,|\nu(x-z)|^2|\nabla f_\ell(x-y)| \langle \Psi, a^*_w a^*_y a_y a_w\Psi\rangle\bigg)^{1/2}\\
        \le\;& C \rho^{5/2+\delta} L^3\,.
    \end{split}
\end{equation*}
We continue with
\begin{equation*}
    \begin{split}
        |G_{23}|\le\;& C \|\nu\|_2^2 \|\nabla f_\ell\|_2^2 \int_{|t-w|\le C\ell} \langle\Psi, a^*_t a^*_w a_w a_t\Psi\rangle\le C \rho^{5/2+\delta} L^3
    \end{split}
\end{equation*}
and
\begin{equation*}
    |G_{24}|, |G_{25}|\le C \|\nu\|_\infty^2  \|\nabla f_\ell\|_1^2 \langle \Psi, \mathcal{N}\Psi\rangle \le C \rho^{5/2+\delta}  L^3.
\end{equation*}
Hence $|G_2| \leq C \rho^{5/2+\delta} L^3$, if $\eps, \delta > 0$ are small enough. Together with (\ref{eq:phi3-3i}) and \eqref{eq:G1ff}, this completes the proof of the proposition.
\qed

\subsection{Proof of Proposition \ref{prop:1-2}}

From \eqref{eq:Phi_123}, we find 
\begin{equation}\label{eq:P1P2} \begin{split} 2 \int &dx \, \langle \Phi_1 (x) , \Phi_2 (x) \rangle \\ &= 2 \sqrt{\rho_0} \int dx dy dz \, \nabla \nu (x-y) \cdot \frac{\nabla f_\ell (x-z)}{f_\ell (x-z)} \langle J(x)^2 a_z^* a_z \Psi, a_y^* J(y) \Psi \rangle \\ &= 2\sqrt{\rho_0} \int dx dy dz \, \nabla \nu (x-y) \cdot \nabla f_\ell (x-z) f_\ell (x-z) \langle J(x)^2 a_z \Psi, a_z a_y^* J(y) \Psi \rangle \\ &= P_1 + P_2  
 \end{split}  \end{equation} 
with 
\[ \begin{split} P_1 &= \sqrt{\rho_0} \int dx dy \, \nabla \nu (x-y) \nabla f_\ell^2 (x-y) \langle J(x)^2 a_y \Psi , J(y) \Psi \rangle  \\ 
P_2 &= \sqrt{\rho_0} \int dx dy dz \, \nabla \nu (x-y) \nabla f_\ell^2 (x-z) \langle a_y J(x)^2 a_z \Psi , a_z J(y) \Psi \rangle \,. \end{split} \]
We write 
\[ \begin{split} P_1 = \; &\sqrt{\rho_0}\int dx dy \, \nabla \nu (x-y) \nabla f_\ell^2 (x-y) \langle a_y \Psi , \Psi \rangle \\ &+ \sqrt{\rho_0} \int dx dy \, \nabla \nu (x-y) \nabla f_\ell^2 (x-y) \langle a_y \Psi , (J(x)^2 J(y) - 1) \Psi \rangle  \\ = \; & P_{11} + P_{12} \,. \end{split} \]
By Cauchy-Schwarz, we can estimate 
\[  \begin{split} |P_{12}| \leq C \sqrt{\rho} &\Big[ \int dx dy |\nabla \nu (x-y)| |\nabla f_\ell^2 (x-y)| \langle \Psi, a_y^* d\Gamma (\o_{\ell,x} + \o_{\ell,y}) a_y \Psi \rangle \Big]^{1/2} \\ &\hspace{2cm} \times \Big[ \int dx dy |\nabla \nu (x-y)| |\nabla f_\ell^2 (x-y)| \langle \Psi,  d\Gamma (\o_{\ell,x} + \o_{\ell,y}) \Psi \rangle \Big]^{1/2}\,. \end{split} \]
We have, from Prop. \ref{prop:N_on_psi}
\[ \begin{split} \int dx dy \, & |\nabla \nu (x-y)| |\nabla f_\ell^2 (x-y)| \langle \Psi,  d\Gamma (\o_{\ell,x} + \o_{\ell,y})  \Psi \rangle \\ &= 2 \int dx dy dz \,  |\nabla \nu (x-y)| |\nabla f_\ell^2 (x-y)| \o_\ell (x-z) \langle \Psi, a_z^* a_z  \Psi \rangle  \\ &\leq C \| \nabla \nu \|_\infty \| \nabla f_\ell \|_1 \| \o_\ell \|_1 \langle \Psi, \cN \Psi \rangle \leq C \rho^{2-\delta} L^3 \end{split} \]
and, with Lemma \ref{lm:aaaaC}, 
\[ \begin{split}  \int dx dy &|\nabla \nu (x-y)| |\nabla f_\ell^2 (x-y)| \langle \Psi, a_y^* d\Gamma (\o_{\ell,x} + \o_{\ell,y}) a_y \Psi \rangle \\ &\leq \int dx dy dz  |\nabla \nu (x-y)| |\nabla f_\ell^2 (x-y)| \big( \o_\ell (x-z) + \o_\ell (y-z) \big)  \langle \Psi, a_y^* a_z^* a_z a_y \Psi \rangle \\ &\leq C \| \nabla \nu \|_\infty \| \nabla f_\ell \|_1  \int_{|y-z| \leq C \ell} dy dz \, \langle \Psi, a_y^* a_z^* a_z a_y \Psi \rangle \leq C \rho^{3-c(\eps + \delta)}  L^3  \end{split} \] for some $c>0$.
This implies that $|P_{12}| \leq C \rho^{5/2+\delta}L^3$, if $\eps , \delta > 0$ are small enough. Moreover, with Lemma~\ref{lm:1b} we find 
\[ \Big| P_{11} - \rho_0 L^3  \int dx \nabla \nu (x) \nabla f_\ell^2 (x) \Big| \leq C \rho^{5/2+\delta} L^3 \]
 if $\eps, \delta > 0$ are small enough. Hence,
 \begin{equation}\label{eq:P1fin}  \Big| P_{1} - \rho_0 \int dx \nabla \nu (x) \nabla f_\ell^2 (x) \Big| \leq C \rho^{5/2+\delta}L^3\,. \end{equation} 
On the other hand, we can write 
\[ \begin{split} 
P_2 = \; &\sqrt{\rho_0} \int dx dy dz \, \nabla \nu (x-y) \nabla f_\ell^2 (x-z) f_\ell^2 (x-y) f_\ell (z-y) \langle a_y a_z \Psi , J(x)^2 J(y) a_z \Psi \rangle \\ =\; & \sqrt{\rho_0} \int dx dy dz \, \nabla \nu (x-y) \nabla f_\ell^2 (x-z) \big( f_\ell^2 (x-y) f_\ell (z-y) -1 \big)  \\ &\hspace{7cm} \times \langle a_y a_z \Psi , J(x)^2 J(y) a_z \Psi \rangle \\ &+\sqrt{\rho_0} \int dx dy dz \, \nabla \nu (x-y) \nabla f_\ell^2 (x-z)  \langle a_y a_z \Psi , (J(x)^2 -1) J(y) a_z \Psi \rangle \\
&+\sqrt{\rho_0} \int dx dy dz \, \nabla \nu (x-y) \nabla f_\ell^2 (x-z) \langle a_y a_z \Psi , (J(y)-1) a_z \Psi \rangle \\
&+ \sqrt{\rho_0} \int dx dy dz \, \nabla \nu (x-y) \nabla f_\ell^2 (x-z)  \langle a_y a_z \Psi , a_z \Psi \rangle \\ =\; &P_{21}+ P_{22} + P_{23} + P_{24}\,.  \end{split} \]
With Lemma \ref{lm:eta} and Lemma \ref{lm:aaaaC}, we bound 
\[ \begin{split} |P_{21}| &\leq C \sqrt{\rho} \int dx dy dz \, |\nabla \nu (x-y)| |\nabla f_\ell (x-z)| \big(\o_\ell (x-y) + \o_\ell (y-z)\big) \| a_y a_z \Psi \| \| a_z \Psi \| \\ 
&\leq C \sqrt{\rho} \, \| \nabla \nu \|_\infty \|\nabla f_\ell\|_1 \|\o_\ell\|_2 
\Big[  \int_{|y-z| \leq C \ell}  dy dz \, \| a_y a_z \Psi \|^2 \Big]^{1/2} \Big[ \int dz \, \| a_z \Psi \|^2 \Big]^{1/2}\\
& \leq C \rho^{5/2+\delta} L^3  \end{split} \] 
if $\eps, \delta > 0$ are small enough. Again by Lemma \ref{lm:eta} and Lemma \ref{lm:aaaaC}, we get (as above, we use (\ref{eq:decay}) to restrict the integral to $|x-y|< R$, with a negligible error)
\[ \begin{split} |P_{22}| \leq \; & C \rho^{5/2+\delta} L^3 + C \sqrt{\rho} \Big[ \int_{|x-y| \leq R}  dx dy dz dw \, |\nabla f_\ell (x-z)| \o_\ell (x-w) \| a_y a_z a_w \Psi \|^2 \Big]^{1/2} \\ &\hspace{1.5cm} \times \Big[ \int dx dy dz dw \, |\nabla \nu (x-y)|^2 |\nabla f_\ell (x-z)| \o_\ell (x-w)| \| a_z a_w \Psi \|^2 \Big]^{1/2} \\ \leq \; &C \rho^{5/2+\delta} L^3 + C \|\nabla f_\ell\|_1 \| \nabla \nu \|_2 \Big[ \int_{\substack{|z-w| \leq C\ell \\ |y-w| \leq R}} dy dzdw \, \| a_y a_z a_w \Psi \|^2 \Big]^{1/2} \\ &\hspace{6cm} \times  \Big[ \int_{|z-w| \leq  C \ell} dz dw \, \| a_z a_w \Psi \|^2 \Big]^{1/2}\\ \leq\; &C \rho^{5/2+\delta} L^3 \end{split} \]
if $\eps, \delta > 0$ are sufficiently small. As for $P_{23}$, we decompose
    \begin{equation*}
        \begin{split}
        P_{23} = P_{231}+ P_{232} 
        \end{split}
    \end{equation*}
    with
    \begin{equation*}
        \begin{split}
   P_{231}=\;&\rho_0 \int dxdydz\,\nabla\nu(x-y)  \nabla f_\ell^2(x-z) \langle \Psi, a^*_y b^*_z(J(y)-1)\Psi\rangle  \\
            P_{232}=\;&\sqrt{\rho_0}\int dxdydz\,\nabla\nu(x-y)  \nabla f_\ell^2(x-z) \langle \Psi, a^*_y a^*_z(J(y)-1)b_z\Psi\rangle \,. 
        \end{split}
    \end{equation*}
With Cauchy-Schwarz, $(1- J(y))^2 \leq 1- J(y) \lesssim d\Gamma (\o_{\ell,y})$ and Lemma \ref{lm:a-bds}, we find
       \begin{equation*}
        \begin{split}
            |P_{231}| \le\;& C \rho^{5/2+\delta} L^3+ C \rho \|\nabla f_\ell\|_1 \|\nabla\nu\|_2 \|\o_\ell \|_1^{1/2} \langle\Psi,\mathcal{N}\Psi\rangle^{\frac{1}{2}}\Big( \int_{|y-z|\le R} dydz\langle \Psi,b^*_z a^*_y a_y b_z\Psi\rangle\Big)^{\frac{1}{2}}\\
            \le\;& C \rho^{5/2+\delta} L^3
        \end{split}
    \end{equation*}
    and, also with Lemma \ref{lm:aaaaC}, 
    \begin{equation*}
        \begin{split}
            |P_{232}| \le\;&C \rho^{5/2+\delta} L^3 + C \rho^{1/2} \|\nabla f_\ell\|_1 \|\nabla\nu\|_\infty \|\o_\ell\|_1^{1/2} \bigg( \int_{|z-w|\le R} dzdw\,\langle\Psi, b^*_z a^*_w a_w b_z\Psi\rangle \bigg)^{1/2}\\
            &\hspace{3cm} \times \bigg(\int_{\substack{|y-w|\le C \ell \\ |y-z|\le R}} dydzdw\,\langle \Psi, a^*_y a^*_z a^*_w a_w a_z a_y\Psi\rangle\bigg)^{1/2}\\
            \le\;& C \rho^{5/2+\delta} L^3        \end{split}
    \end{equation*}
    if $\eps , \delta > 0$ are small enough. 
    Thus, $|P_{23}| \leq C \rho^{5/2+\delta }L^3$. Finally, since $\int \nabla \nu (x) dx = \int \nabla f_\ell^2 (x) dx =0$, we can write $P_{24} =  P_{241} + P_{242}$, with
      \begin{equation*}
    \begin{split}
        P_{241}=\;&\rho_0\int dxdydz\, \nabla\nu(x-y) \nabla f_\ell^2(x-z) \langle b_y b_z \Psi,  \Psi\rangle \\
        P_{242} =\;&\sqrt{\rho_0}\int dxdydz\, \nabla\nu(x-y) \nabla f_\ell^2(x-z)  \langle b_y a_z \Psi, b_z \Psi\rangle \,.
    \end{split}
    \end{equation*}
    With Lemma \ref{lm:a-bds}, we estimate
        \begin{equation} \label{eq:bound_P_11}
        \begin{split}
            | P_{241}| \le \;& C \rho^{5/2+\delta} L^3 + C \rho \, \|\nabla f_\ell\|_1 \|\nabla\nu\|_2 L^{3/2} \bigg( \int_{|y-z|\le R} dydz\,\langle \Psi, b^*_x b^*_y b_y b_x \Psi\rangle\bigg)^{1/2} \\ \leq \; &C \rho^{5/2+\delta} L^3 
        \end{split}
    \end{equation}
       and 
    \begin{equation*}
        \begin{split}
            |P_{242}|\le\;& C \rho^{5/2+\delta} L^3 \\ &+C \rho^{1/2} \|\nabla f_\ell\|_1 \| \nabla \nu \|_2  \bigg( \int_{|y-z|\le R} dydz \,\langle \Psi, b^*_y a^*_z a_z b_y\Psi\rangle \bigg)^{1/2} \bigg( \int dz \langle \Psi, b^*_z b_z\Psi\rangle\bigg)^{1/2}\\
            \le\;& C \rho^{5/2+\delta} L^3 + C \rho^{1/2} \|\nabla f_\ell\|_1 \|\nabla\nu\|_2\bigg( \int dz\,\langle \Psi, b^*_z b_z \Psi\rangle\bigg)^{1/2}\\
            &\qquad\qquad\times\bigg( \rho R^3 \int dy\,\langle \Psi, b^*_y b_y\Psi\rangle + \int_{|y-z| \leq R} dydz\,\langle \Psi, b^*_y b^*_z b_z b_y\Psi\rangle \bigg)^{1/2}\\
            \le\;& C \rho^{5/2+\delta} L^3 
                   \end{split}
    \end{equation*}
if $\eps , \delta > 0$ are small enough. This shows that $|P_2| \leq C \rho^{5/2+\delta} L^3$, if $\eps, \delta> 0$ are sufficiently small and, together with (\ref{eq:P1P2} and (\ref{eq:P1fin}), it shows that 
\[ 2 \int dx \langle \Phi_1 (x), \Phi_2 (x) \rangle \leq 2 \rho_0 L^3 \int dx \, f_\ell (x) \nabla f_\ell (x) \nabla \nu (x) + C \rho^{5/2+\delta} L^3\,. \]
Since, from Lemma \ref{lm:eta}, 
\[ \begin{split}  \rho_0 \Big| \int dx f_\ell (x) &\nabla f_\ell (x) \big( (\gamma^{-1} -1) * \nabla \sigma \big) (x) \Big|  \\ &\leq C \rho \|\nabla f_\ell\|_1 \| (\gamma^{-1} - 1) * \nabla \sigma \|_\infty \leq C \rho \ell \| \gamma^{-1}-1\|_2 \| \nabla \sigma \|_2 \leq C \rho^{5/2+\delta} \end{split} \]
if $\eps, \delta > 0$ are small enough, this concludes the proof of the proposition. 
\qed

\subsection{Proof of Proposition \ref{prop:1-3}}

From (\ref{eq:Phi_123}), we have 
\[ \begin{split} 
 \int &dx  \langle \Phi_1 (x) , \Phi_3 (x) \rangle \\ =\; & \sqrt{\rho_0} \int dx dy dz dw \, \nu (x-y) \frac{\nabla f_\ell (x-z)}{f_\ell (x-z)} \frac{\nabla f_\ell (x-w)}{f_\ell (x-w)} \langle J(x)^2 a_z^* a_z \Psi , a_w^* a_w a_y^* J(y) \Psi \rangle  \\ =\; & \sqrt{\rho_0} \int dx dy dz \, \nu (x-y) |\nabla f_\ell (x-z)|^2 \langle J(x)^2 a_z \Psi , a_z a_y^* J(y) \Psi \rangle \\ &+ \frac{\sqrt{\rho_0}}{4} \int dx dy dz dw \, \nu (x-y) \nabla f^2_\ell (x-z) \nabla f^2_\ell (x-w) \langle J(x)^2 a_w a_z \Psi , a_w a_z a_y^* J(y) \Psi \rangle \end{split} \]
where we used the identity (\ref{eq:pull}) to pass $a_z^* , a_w$ through $J(x)^2$. Commuting also $a_y^*$ to the left, we arrive at 
\begin{equation}\label{eq:Q1-Q4} \begin{split} 
\int &dx  \langle \Phi_1 (x) , \Phi_3 (x) \rangle \\ =\; & \sqrt{\rho_0} \int dx dy \, \nu (x-y) |\nabla f_\ell (x-y)|^2 \langle a_y \Psi , J^2 (x) J(y) \Psi \rangle \\ &+ \sqrt{\rho_0} \int dx dy dz \, \nu (x-y) |\nabla f_\ell  (x-z)|^2 f^2_\ell (x-y) f_\ell (y-z) \langle a_y a_z \Psi, J^2 (x) J(y)  a_z \Psi \rangle \\ &+\frac{\sqrt{\rho_0}}{2} \int dx dy dz \, \nu (x-y) \nabla f_\ell^2 (x-z) \nabla f_\ell^2 (x-y) f_\ell (y-z) \langle a_y a_z \Psi, J^2 (x) J(y)  a_z \Psi \rangle \\ &+ \frac{\sqrt{\rho_0}}{4} \int dx dy dz dw \, \nu (x-y) \nabla f_\ell^2 (x-z) \nabla f_\ell^2 (x-w) f_\ell^2 (x-y) f_\ell (y-w) f_\ell (y-z) \ \\ &\hspace{7cm} \times \langle a_w a_y a_z \Psi, J^2 (x) J(y)  a_w a_z \Psi \rangle \\ =: \; &Q_1 + Q_2 + Q_3 + Q_4 \,.\end{split} \end{equation} 
The term $Q_4$ is negligible. In fact, with Lemma \ref{lm:aaaaC} and using (\ref{eq:decay}) to handle the region $|w-y| > R$, we find 
\begin{equation}\label{eq:Q4fin} \begin{split} 
|Q_4| \leq \; &C \rho^{5/2+\delta} L^3 +C \sqrt{\rho} \Big[ \int_{\substack{|w-z| \leq C \ell,\\ |w-y| \leq R}} dx dy dz dw |\nabla f_\ell (x-z)|^2 \| a_y a_z a_w \Psi \|^2 \Big]^{1/2} \\ &\hspace{2.5cm} \times \Big[ \int_{|w-z| \leq C \ell}  dx dy dz dw |\nu (x-y)|^2 |\nabla f_\ell (x-w)|^2 \| a_w a_z \Psi \|^2 \Big]^{1/2} \\ \leq \; &C \rho^{5/2+\delta} L^3 +C \sqrt{\rho} \,\| \nu \|_2 \| \nabla f_\ell \|_2^2 \Big[ \int_{\substack{|w-z| \leq C \ell, \\ |w-y| \leq R}} dy dz dw \, \| a_y a_z a_w \Psi \|^2 \Big]^{\frac{1}{2}} \\ &\hspace{4cm} \times \Big[ \int_{|w-z| \leq C \ell} dwdz \, \| a_w a_z \Psi \|^2 \Big]^{\frac{1}{2}} \\ \leq \; & C \rho^{5/2+\delta} L^3 \end{split} \end{equation}
if $\eps, \delta > 0$ are small enough. Also $Q_3$ is negligible, since
\begin{equation}\label{eq:Q3fin} \begin{split} |Q_3| \leq \; &C \sqrt{\rho} \int dx dy dz \, |\nu (x-y)| |\nabla f_\ell (x-z)| |\nabla f_\ell (x-y)| \| a_y a_z \Psi \| \| a_z \Psi \| \\ \leq \; &C \sqrt{\rho} \, \| \nu \|_\infty \| \nabla f_\ell \|_1 \| \nabla f_\ell \|_2 \Big[ \int_{|y-z| \leq C \ell} dy dz \, \| a_y a_z \Psi \|^2 \Big]^{1/2} \Big[ \int dz \|a_z \Psi \|^2 \Big]^{1/2} \\ \leq \; & C \rho^{5/2+\delta} L^3\,. \end{split} \end{equation}
As for the term $Q_1$, we decompose it as
\[ \begin{split} Q_1 =  \; & \rho_0 L^3  \int dx \, \nu (x) |\nabla f_\ell (x)|^2 + \sqrt{\rho_0} \int dx \, \nu (x) |\nabla f_\ell (x)|^2 \int dy  \langle b_y \Psi,\Psi \rangle \\ &+\sqrt{\rho_0} \int dx dy \, \nu (x-y) |\nabla f_\ell (x-y)|^2 \langle a_y \Psi, (J^2 (x) J(y) -1) \Psi \rangle  \\ = \; & 
\rho_0 L^3  \int dx \, \nu (x) |\nabla f_\ell (x)|^2 +Q_{12} + Q_{13}\,. \end{split} \] 
From Lemma \ref{lm:1b}, we have 
\[ |Q_{12}| \leq C \sqrt{\rho_0} \, \| \nu \|_\infty \| \nabla f_\ell \|_2^2 \Big| \int dy \, \langle b_y \Psi, \Psi \rangle \Big| \leq C \rho^{5/2+\delta} L^3\,. \]
With $0 \leq 1-J^2 (x)J(y) \lesssim d\Gamma (\o_{\ell,x} + \o_{\ell,y})$, we find 
\begin{equation}\label{eq:Q13} \begin{split} |Q_{13}| \leq \; &C \sqrt{\rho} \, \| \nu \|_\infty \Big[ \int dx dy dw \, |\nabla f_\ell (x-y)|^2 \big( \o_\ell (x-w) + \o_\ell (y-w) \big) \| a_y a_w \Psi \|^2 \Big]^{1/2} \\ &\times \Big[ \int dx dy dw  \, |\nabla f_\ell (x-y)|^2 \big( \o_\ell (x-w) + \o_\ell (y-w) \big) \| a_w \Psi \|^2 \Big]^{1/2} \\ \leq \:&C \sqrt{\rho} \, \| \nu \|_\infty \| \nabla f_\ell \|^2_2 \| \o_\ell \|_1^{1/2} \Big[ \int_{|y-w| \leq C \ell}  dy dw \, \| a_y a_w \Psi \|^2 \Big]^{1/2} \Big[ \int dw \, \| a_w \Psi \|^2 \Big]^{1/2} \\ \leq \; &C \rho^{5/2+\delta} L^3\,. \end{split} \end{equation} 
We conclude that 
\begin{equation}\label{eq:Q1fin} \Big| Q_1 - \rho_0 L^3 \int dx \, \nu (x) |\nabla f_\ell (x)|^2 \Big| \leq C \rho^{5/2+\delta} L^3 \end{equation} 
if $\eps, \delta > 0$ are small enough. Finally, we consider $Q_2$. Proceeding as in (\ref{eq:Q13}), we find  
\[ \begin{split} \sqrt{\rho_0}\, \Big| &\int dx dy dz \, \nu (x-y) |\nabla f_\ell  (x-z)|^2 \big( f^2_\ell (x-y) f_\ell (y-z) -1 \big)  \langle a_y a_z \Psi, J^2 (x) J(y)  a_z \Psi \rangle \Big| \\ &\leq C \sqrt{\rho} \, \| \nu \|_\infty  \int dx dy dz \, |\nabla f_\ell (x-z)|^2 \big( \o_\ell (x-y) + \o_\ell (y-z) \big) \| a_y a_z \Psi \| \| a_z \Psi \| \\ &\leq C \rho^{5/2+\delta} L^3 \,.\end{split} \]
Moreover, we have  
\[ \begin{split} \sqrt{\rho_0}\, &\Big| \int dx dy dz \,  \nu (x-y) |\nabla f_\ell  (x-z)|^2  \langle a_y a_z \Psi, (J^2 (x) -1) J(y)  a_z \Psi \rangle \Big| \\ \leq \; &C \rho^{5/2+\delta} L^3 + C \sqrt{\rho} \Big[ \int_{|x-y| \leq R} dx dy dz dw \,  |\nabla f_\ell  (x-z)|^2 \o_\ell (x-w) \| a_y a_z a_w \Psi \|^2 \Big]^{1/2} \\ &\hspace{3cm} \times \Big[ \int dx dy dz dw \, |\nu (x-y)|^2  |\nabla f_\ell  (x-z)|^2 \o_\ell (x-w) \| a_z a_w \Psi \|^2 \Big]^{1/2} \\ 
\leq \; &C \rho^{5/2+\delta} L^3 \\ &+C \sqrt{\rho} \| \nu \|_2 \| \nabla f_\ell \|_2^2 \Big[ \int_{\substack{|y-z| \leq CR ,\\ |w-z| \leq C \ell}} dy dz dw \, \| a_y a_z a_w \Psi \|^2 \Big]^{\frac{1}{2}} \Big[ \int_{|z-w| \leq C \ell} dz dw \| a_z a_w \Psi \|^2 \Big]^{\frac{1}{2}} \\ \leq \; &C \rho^{5/2+\delta} L^3 \,.\end{split} \]
Thus,
\[ \Big| Q_2 -  \sqrt{\rho_0}  \int dx dy dz \,  \nu (x-y) |\nabla f_\ell  (x-z)|^2  \langle a_y a_z \Psi, J(y) a_z \Psi \rangle \Big| \leq C \rho^{5/2+\delta} L^3\,. \]
Using \eqref{eq:a_to_b} and $\int \nu (x) dx =0$, we can further argue that
    \begin{equation} \label{eq:decomp_Q_121}
     \sqrt{\rho_0}  \int dx dy dz \,  \nu (x-y) |\nabla f_\ell  (x-z)|^2  \langle a_y a_z \Psi, J(y) a_z \Psi \rangle   =Q_{21}+Q_{22}+Q_{23}+Q_{24}
    \end{equation}
    with
    \begin{equation}\label{eq:Q21}
        Q_{21}=\sqrt{\rho_0} \int dxdydz\,|\nabla f_\ell(x-z)|^2 \nu(x-y)  \langle \Psi,  a^*_za^*_ya_z \Psi\rangle ,
    \end{equation}
    and
    \begin{equation*}
        \begin{split}
            Q_{22}=\;&{\rho_0} \int dxdydz\,|\nabla f_\ell(x-z)|^2 \nu(x-y)  \langle \Psi,  a^*_y (J(y)-1)b_z \Psi\rangle \\
            Q_{23}=\;&{\rho_0} \int dxdydz\,|\nabla f_\ell(x-z)|^2 \nu(x-y)  \langle \Psi,  b^*_z a^*_y (J(y)-1) \Psi\rangle \\
            Q_{24}=\;&\sqrt{\rho_0} \int dxdydz\,|\nabla f_\ell(x-z)|^2 \nu(x-y)  \langle \Psi,  b^*_z a^*_y (J(y)-1)b_z \Psi\rangle \, .
        \end{split}
    \end{equation*}
We claim that $Q_{22}, Q_{23}, Q_{24}$ are negligible. In fact, 
    \begin{equation*}
        \begin{split}
            |Q_{22}|\le\;&C \rho^{5/2+\delta} L^3 +  C \rho \bigg( \int dxdydz\,|\nabla f_\ell(x-z)|^2 |\nu(x-y)|^2\langle \Psi, a^*_y a_y\Psi\rangle\bigg)^{1/2}\\
            &\qquad\times \bigg( \int_{|x-y|\le R} dxdydzdw\,|\nabla f_\ell(x-z)|^2 \o_\ell(y-w) \langle \Psi, b^*_z a^*_w a_w b_z\Psi\rangle\bigg)^{1/2}\\
            \le\;& C \rho^{5/2+\delta} L^3 +C \rho \|\nabla f_\ell\|_2^2 \|\nu\|_2 \|\o_\ell\|_1^{1/2} \langle \Psi,\mathcal{N}\Psi\rangle^{1/2}\\
            &\qquad\times\bigg( \rho R^3 \int dz\,\langle\Psi, b^*_z b_z\Psi\rangle + \int_{|z-w|\le CR} dzdw\, \langle \Psi, b^*_z b^*_w b_w b_z\Psi\rangle \bigg)^{1/2}\\
            \le\;& C \rho^{5/2+\delta} L^3 
        \end{split}
    \end{equation*}
    and, analogously, $|Q_{23}| \leq C \rho^{5/2+\delta} L^3$, if $\eps, \delta> 0$ are small enough. 
        Moreover, by \eqref{eq:a_to_b} and  \eqref{eq:kbs} 
    \begin{equation*}
        \begin{split}
            |Q_{24}|\le\;&C \rho^{5/2+\delta} L^3\\ &+ \rho^{1/2} \| \nu \|_\infty \bigg( \int_{|y-x|\le R} dxdydzdw\,|\nabla f_\ell(x-z)|^2 \o_\ell(y-w) \langle \Psi, b^*_z a^*_y a^*_w a_w a_y b_z\Psi\rangle\bigg)^{\frac{1}{2}}\\
            &\times \bigg( \int_{|y-x| \leq R} dxdydzdw\,|\nabla f_\ell(x-z)|^2 \o_\ell(y-w) \langle \Psi, b^*_z a^*_w a_w b_z\Psi\rangle\bigg)^{1/2}\\
            \le\;&C \rho^{5/2+\delta} L^3 + C \rho^{1/2} \,\|\nu\|_\infty \|\nabla f_\ell\|_2^2 \,  \|\o_\ell\|_1^{1/2}\\
            &\quad\times \bigg( \rho R^3  \int_{|y-w| \leq \ell} dy dw \, \|a_y a_w \Psi \|^2 + \int_{\substack{|y-w| \leq \ell \\ |z-y| \leq C R}} dy dz dw \, \| a_z a_y a_w \Psi \|^2 \bigg)^{1/2} \\
 &\quad\times \bigg( \rho R^3 \int dz\,\langle\Psi,b^*_z b_z\Psi\rangle + \int_{|z-w|\le CR } dz dw\,\langle \Psi, b^*_z b^*_w b_w b_z\Psi\rangle\bigg)^{1/2}\\ 
            \le\;& C \rho^{5/2+\delta} L^3
        \end{split}
    \end{equation*}
    if $\eps, \delta > 0$ are small enough. We conclude that $|Q_2 - Q_{21}| \leq C \rho^{5/2+\delta} L^3$, with $Q_{21}$ defined as in (\ref{eq:Q21}). Using again $\int \nu (x) dx = 0$, we write 
        \begin{equation*}
        Q_{21}=Q_{211}+Q_{212}+Q_{213}
    \end{equation*}
    with
    \begin{equation}\label{eq:Q21123}
        \begin{split}
            Q_{211}=\;&{\rho_0} \int dxdydz\,|\nabla f_\ell(x-z)|^2 \nu(x-y)  \langle \Psi,  b^*_z b^*_y \Psi\rangle \\
            Q_{212}=\;&{\rho_0} \int dxdydz\,|\nabla f_\ell(x-z)|^2 \nu(x-y)  \langle \Psi,   b^*_y b_z \Psi\rangle \\
            Q_{213}=\;&\sqrt{\rho_0} \int dxdydz\,|\nabla f_\ell(x-z)|^2 \nu(x-y)  \langle \Psi,  b^*_z b^*_y b_z\Psi\rangle .
        \end{split}
    \end{equation}
    The last term satisfies 
    \begin{equation*}
        \begin{split}
            |Q_{213}|\le\;& C \rho^{5/2+\delta} L^3 \\ &+ C \rho^{1/2} \|\nabla f_\ell\|_2^2 \|\nu\|_2 \bigg(\int dz\,\langle\Psi, b^*_z b_z\Psi\rangle\bigg)^{1/2} \bigg( \int_{|y-z|\le R} dydz\, \langle \Psi, b^*_y b^*_z b_z b_y\Psi\rangle \bigg)^{1/2}\\
            \le\;& C \rho^{5/2+\delta} L^3
        \end{split}
    \end{equation*}
    which follows from \eqref{eq:rmkwtN} and \eqref{eq:kbs}. For $Q_{211}$ we recall from (\ref{eq:def-cc})  that $b_x^* = c^* (\gamma_x) + c (\sigma_x)$ for $x = z,w$. Using $\gamma * \nu = \sigma$, (\ref{eq:decay}) to restrict the integral to $|z-y| \leq R$ and Lemma \ref{lm:c's}, we find \begin{equation*}
        \begin{split}
            \bigg| Q_{211}- &\rho_0 L^3 \int dx\,|\nabla f_\ell(x)|^2  (\sigma*\sigma)(x)\bigg|\\
            \le\;& C \rho L^{3/2} \|\nabla f_\ell\|_2^2\|\nu\|_2 \bigg( \int_{|y-z|\leq R} dydz\,\Big( \| c(\gamma_z) c(\gamma_y)\Psi\|^2+ \|c(\sigma_z) c(\sigma_y)\Psi\|^2 \Big) \bigg)^{1/2}\\
            &+C \rho \|\nabla f_\ell\|_2^2 \|\nu\|_1 \bigg(\int dz\,\| c(\gamma_z)\Psi\|^2\bigg)^{1/2}\bigg(\int dz\,\| c(\sigma_y)\Psi\|^2\bigg)^{1/2} + C \rho^{5/2+\delta} L^3 \\
            \le\;& C \rho^{5/2+\delta} L^3 
        \end{split}
    \end{equation*}
    if $\eps, \delta> 0$ are small enough. Similarly,
    \begin{equation*}
        \begin{split}
            \bigg| Q_{212}- &\rho_0 L^3 \int dx\, |\nabla f_\ell(x)|^2 \big(\nu*\sigma*\sigma\big)(x) \bigg|\\
            \le\;&C \rho^{5/2+\delta} L^3 +  C \rho \|\nabla f_\ell\|_2^2 \|\nu\|_1 \int dy\, \Big( \| c(\gamma_y) \Psi \|^2 + \| c (\sigma_y) \Psi \|^2 \Big) \\
            &+ C \rho L^{3/2} \|\nabla f_\ell\|_2^2 \|\nu\|_2 \bigg( \int_{|y-z|\le R} dydz\, \big\|c(\gamma_y) c(\sigma_z)\Psi\big\|^2\bigg)^{1/2}\\
            \le\;& C \rho^{5/2+\delta} L^3        \end{split}
    \end{equation*}
    if $\eps, \delta > 0$ are small enough. Thus, 
    \begin{equation*}
    \begin{split}
        \Big|Q_{2}-&\rho_0 L^3 \int dx\,|\nabla f_\ell(x)|^2 (\sigma*\sigma)(x)- \rho_0 L^3 \int dx\, |\nabla f_\ell(x)|^2 \big(\nu*\sigma*\sigma\big)(x)\Big| \le C\rho^{5/2+\delta} L^3 
    \end{split}
    \end{equation*}
    if $\eps, \delta> 0$ are small enough. Together with \eqref{eq:Q1-Q4}, \eqref{eq:Q4fin}, \eqref{eq:Q3fin}, \eqref{eq:Q1fin} and with the observation that 
    \begin{equation*}
    \begin{split}
        \rho_0 L^3 \int& dx\,|\nabla f_\ell(x)|^2\nu(x)+\rho_0 L^3 \int dx\, |\nabla f_\ell(x)|^2 \big(\nu*\sigma*\sigma\big)(x)\\ &\hspace{4cm} 
        = \rho_0 L^3 \int dx\,|\nabla f_\ell(x)|^2(\gamma*\sigma)(x)
    \end{split}
    \end{equation*}
    this completes the proof of the proposition.
    \qed

\subsection{Proof of Proposition \ref{prop:2-3}}
\label{subsec:2-3} 

From \eqref{eq:Phi_123}, we have
    \begin{equation*}
    \begin{split}
    \int dx &\langle \Phi_2 (x) , \Phi_3 (x) \rangle \\ = \; &\int dx dy dz dw \, \nabla \nu (x-y) \nu (x-z) \frac{\nabla f_\ell (x-w)}{f_\ell (x-w)} \langle J(y) \Psi , a_y a_w^* a_w J(x)^2 a_z^* J(z) \Psi \rangle  \\ = \; & \int dx dy dz \, \nabla \nu (x-y) \nu (x-z) \frac{\nabla f_\ell (x-y)}{f_\ell (x-y)} \langle J(y) \Psi , a_y J(x)^2 a_z^* J(z) \Psi \rangle \\ &+ \int dx dy dz dw \, \nabla \nu (x-y) \nu (x-z) \frac{\nabla f_\ell (x-w)}{f_\ell (x-w)} \langle J(y) \Psi, a_w^* a_w a_y J (x)^2 a_z^* J(z) \Psi \rangle\,. \end{split} \end{equation*} 
    Using again the identity (\ref{eq:pull}) and the canonical commutation relations to move the operator $a_z^*$ to the left, we obtain 
        \[ \begin{split} 
    &\int dx \langle \Phi_2 (x) , \Phi_3 (x) \rangle \\ 
     &= \int dx dy \nabla \nu ( x-y) \nu (x-y) \nabla f_\ell (x-y) f_\ell (x-y) \langle \Psi, J (x)^2 J (y)^2 \Psi \rangle \\  &+\frac{1}{2} \int dx dy dz \, \nabla \nu (x-y) \nu (x-z) \nabla f^2_\ell (x-y) f^2_\ell (x-z) f^2_\ell (y-z) \langle a_z \Psi , J (x)^2 J(y) J(z) a_y \Psi \rangle \\
    &+ \frac{1}{2} \int dx dy dz \, \nabla \nu (x-y) \nu (x-y) \nabla f^2_\ell (x-z)  f^2_\ell (x-y) f^2_\ell (y-z) \langle a_z \Psi , J (x)^2 J (y)^2  a_z \Psi \rangle \\
    &+ \frac{1}{2} \int dx dy dz \, \nabla \nu (x-y) \nu (x-z) \nabla f^2_\ell (x-z) f^2_\ell (x-y) f^2_\ell (y-z) \langle a_z \Psi , J (x)^2 J(y) J(z) a_y \Psi \rangle \\
    &+ \frac{1}{2} \int dx dy dz dw  \, \nabla \nu (x-y) \nu (x-z) \nabla f^2_\ell (x-w) f^2_\ell (x-y) f^2_\ell (x-z) f^2_\ell (y-z) \\ &\hspace{5cm} \times f_\ell (w-y) f_\ell (w-z) \langle a_w a_z \Psi , J (x)^2 J(y) J(z) a_w a_y \Psi \rangle \\
    &=: R_1 + R_2 + R_3 + R_4 + R_5\,. \end{split} \]
With $1- J(x)^2 J(y)^2 \leq C d\Gamma (\o_{\ell,x} + \o_{\ell,y})$ and with the usual estimates from Lemma \ref{lm:eta} and Prop. \ref{prop:N_on_psi}, we can bound
\begin{equation}\label{eq:R1fin} \begin{split} \Big| &R_1 - L^3 \int dx \nabla \nu (x) \nu (x) \nabla f_\ell (x) f_\ell (x) \Big| \\ &\leq \int dx dy dw \, |\nabla \nu (x-y)| |\nu (x-y)| |\nabla f_\ell (x-y)| \big( \o_\ell (x-w) + \o_\ell (y-w) \big) \langle \Psi, a_w^* a_w \Psi \rangle \\ &\leq C \| \nabla \nu \|_\infty \| \nu \|_\infty \| \nabla f_\ell \|_1 \| \o_\ell \|_1 \langle \Psi, \cN \Psi \rangle \leq C \rho^{5/2+\delta} L^3\,. \end{split} \end{equation} 
To estimate $R_2$, we observe that 
\[ \begin{split} 
\Big|  \int &dx dy dz \, \nabla \nu (x-y) \nu (x-z) \nabla f^2_\ell (x-y) \big( f^2_\ell (x-z) f^2_\ell (y-z)-1 \big)  \\ &\hspace{7cm} \times \langle a_z \Psi , J (x)^2 J(y) J(z) a_y \Psi \rangle \Big| \\
\leq\; &C \int dx dy dz |\nabla \nu (x-y)| |\nu (x-z)| |\nabla f_\ell (x-y)| \big( \o_\ell (x-z) + \o_\ell (y-z) \big) \| a_z \Psi \| \| a_y \Psi \| \\
\leq\; &C \| \nabla \nu \|_\infty \| \nu \|_\infty \| \nabla f_\ell \|_1 \| \o_\ell \|_1 \langle \Psi, \cN \Psi \rangle \leq C \rho^{5/2+\delta} L^3 \,. 
\end{split} \]
Decomposing now \[ J(x)^2 J(y) J(z) - 1 = (J(x)^2 -1) J(y) J(z) + (J(y) - 1) J(z) + (J(z) -1) \] and focusing for example on the contribution proportional to $1- J(x)^2 \leq C d\Gamma (\o_{\ell,x})$, we can bound, by Cauchy-Schwarz and applying (\ref{eq:decay}) to estimate the contribution from $|z-w| > R$, 
\[ \begin{split} 
\Big|  \int &dx dy dz \, \nabla \nu (x-y) \nu (x-z) \nabla f^2_\ell (x-y)  \langle a_z \Psi , (J (x)^2 -1) J(y) J(z) a_y \Psi \rangle \Big| \\
\leq\; &C \rho^{5/2+\delta} L^3 + C \Big[ \int_{|z-w| \leq R} dx dy dz |\nabla \nu (x-y)| |\nabla f_\ell (x-y)|^2 \o_\ell (x-w) \| a_z a_w \Psi \|^2 \Big]^{1/2} \\ &\hspace{2cm} \times \Big[  \int_{|y-w| \leq C \ell} dx dy dz |\nabla \nu (x-y)| |\nu (x-z)|^2 \o_\ell (x-w) \| a_y a_w \Psi \|^2 \Big]^{1/2}   \\
\leq\; &C \rho^{5/2+\delta} L^3 + C \| \nabla \nu \|_\infty \| \nu \|_2 \| \nabla f_\ell \|_2 \| \o_\ell \|_1 \Big[ \int_{|z-w| \leq R} dzdw \, \| a_z a_w \Psi \|^2 \Big]^{\frac{1}{2}}\\ &\hspace{7cm} \times \Big[ \int_{|y-w| \leq C \ell} dydw \, \| a_y a_w \Psi \|^2 \Big]^{\frac{1}{2}} \\  \leq \; &C \rho^{5/2+\delta} L^3 
\end{split} \]
if $\eps, \delta > 0$ are small enough, as it follows from \eqref{eq:kas} and \eqref{eq:aaaaC}. The contributions proportional to $J(y)-1$ and $J(z) -1$ can be bounded similarly. Recalling that $\int \nu (x) dx = 0$, we obtain 
\[ \Big| R_2 - \frac{1}{2} \int dx dy dz \, \nabla \nu (x-y) \nu (x-z) \nabla f^2_\ell (x-y)  \langle \Psi , b_z^* b_y \Psi \rangle \Big| \leq C \rho^{5/2+\delta} L^3\,. \]
Recalling (\ref{eq:def-cc}), we write now $b^*_z = c^* (\gamma_z) + c (\sigma_z)$ and similarly for $b_y$. With Lemma~\ref{lm:c's}, we can control the resulting terms, except the commutator $[c (\sigma_z), c^* (\sigma_y)] = (\sigma * \sigma) (z-y)$, similarly as we did for the term $Q_{211}$ in (\ref{eq:Q21123}). We conclude that 
\begin{equation}\label{eq:R2fin} \Big| R_2 - L^3 \int dx \, \nabla \nu (x)  (\nu*\sigma * \sigma) (x) \nabla f_\ell (x) f_\ell (x)   \Big| \leq C \rho^{5/2+\delta} L^3\,. \end{equation}

The term $R_3$ is negligible, since, by Prop. \ref{prop:N_on_psi}, 
\[ \begin{split}  |R_3| &\leq C \int dx dy dz \, |\nabla \nu (x-y)| |\nu (x-y)| |\nabla f_\ell (x-z)| \| a_z \Psi \|^2 \\ &\leq \| \nabla \nu \|_2 \| \nu \|_2 \| \nabla f_\ell \|_1 \langle \Psi, \cN \Psi \rangle \leq C \rho^{5/2+\delta} L^3\end{split} \]
if $\eps, \delta > 0$ are small enough. 

Proceeding as we did above for $R_2$, we can show that 
\[ \Big| R_4 -  \frac{1}{2} \int dx dy dz \, \nabla \nu (x-y) \nu (x-z) \nabla f^2_\ell (x-z)  \langle \Psi , b_z^* b_y \Psi \rangle \Big| \leq C \rho^{5/2+\delta} L^3\,. \]
Since 
\[ \begin{split} \Big| \int dx dy dz \, &\nabla \nu (x-y) \nu (x-z) \nabla f^2_\ell (x-z)  \langle \Psi , b_z^* b_y \Psi \rangle \Big| \\ &\leq \; \| \nu \|_\infty  \int dx dy dz |\nabla \nu (x-y)| |\nabla f_\ell (x-z)| \| b_z \Psi \|^2 \\ &\leq \| \nu \|_\infty \| \nabla \nu \|_1 \| \nabla f_\ell \|_1 \int dx \| b_x \Psi \|^2 \leq C \rho^{5/2+\delta} L^3\end{split} \]
if $\eps , \delta > 0$ are small enough, we conclude that $|R_4| \leq  C \rho^{5/2+\delta} L^3$. 

Finally, we estimate $R_5$. Similarly as in the analysis of $R_2, R_4$, we first control the contribution of the difference 
\[ \begin{split} 1- f^2_\ell (x-y)  f^2_\ell (x-z)  &f^2_\ell (z-y)  f_\ell (w-y)  f_\ell (w-z) \\ &\leq C \big( \o_\ell (x-y) + \o_\ell (x-z)  + \o_\ell (z-y) + \o_\ell (w-y) +\o_\ell (w-z) \big)\,. \end{split} \]
Focussing for example on the term proportional to $\o_\ell (x-y)$, its contribution to $R_5$ can be estimated by 
\[ \begin{split}  \int &dx dy dz dw \, |\nabla \nu (x-y)| |\nu (x-z)| |\nabla f_\ell (x-w)| \o_\ell (x-y) \| a_w a_z \Psi \| \| a_w a_y\Psi \| \\ \leq \; & C \rho^{5/2+\delta} L^3 + \| \nabla \nu \|_\infty \Big[ \int_{|w-z| \leq R} dx dy dz dw \, |\nabla f_\ell (x-w)|^2  \o_\ell (x-y) \| a_w a_z \Psi \|^2 \Big]^{1/2} \\ &\hspace{1cm} \times \Big[  \int_{|w-y| \leq C \ell}  dx dy dz dw \,  |\nu (x-z)|^2 \o_\ell (x-y)  \| a_w a_y\Psi \| \Big]^{1/2} \\ \leq \; &C \rho^{5/2+\delta} L^3 + \| \nabla \nu \|_\infty \| \nu \|_2 \| \nabla f_\ell \|_2 \| \o_\ell \|_1 \Big[ \int_{|w-z| \leq R} dz dw \, \| a_w a_z \Psi \|^2 \Big]^{1/2}\\ &\hspace{5cm} \times \Big[ \int_{|w-y| \leq C\ell} dy dw \, \| a_w a_y \Psi \|^2 \Big]^{1/2} \\
\leq \; &C \rho^{5/2+\delta} L^3\,.
\end{split} \] 
Most of the other contributions can be bounded similarly. Only the term proportional to $\o_\ell (y-z)$ requires a different procedure. In this case, we decompose 
\[ \begin{split}  \int &dx dy dz dw \, |\nabla \nu (x-y)| |\nu (x-z)| |\nabla f_\ell (x-w)| \o_\ell (y-z) | \langle a_w a_z \Psi , a_w a_y \Psi \rangle |  \\ \leq \; &\rho  \int dx dy dz dw \, |\nabla \nu (x-y)| |\nu (x-z)| |\nabla f_\ell (x-w)| \o_\ell (y-z) \|  a_w \Psi \|^2 \\ &+ \sqrt{\rho}  \int dx dy dz dw \, |\nabla \nu (x-y)| |\nu (x-z)| |\nabla f_\ell (x-w)| \o_\ell (y-z) | \langle a_w \Psi , a_w b_y \Psi \rangle \\  &+ \sqrt{\rho}  \int dx dy dz dw \, |\nabla \nu (x-y)| |\nu (x-z)| |\nabla f_\ell (x-w)| \o_\ell (y-z) | \langle a_w b_z \Psi , a_w \Psi \rangle \\ &+ \int dx dy dz dw \, |\nabla \nu (x-y)| |\nu (x-z)| |\nabla f_\ell (x-w)| \o_\ell (y-z) | \langle a_w b_z \Psi , a_w b_y \Psi \rangle | \\ = \; &R_{51}+ R_{52} +R_{53} +R_{54}\,. \end{split} \] 
By Cauchy-Schwarz, we have 
\[ R_{51} \leq \rho \| \nabla \nu \|_2 \| \nu \|_2 \| \nabla f_\ell \|_1 \| \o_\ell \|_1 \langle \Psi, \cN \Psi \rangle \leq C \rho^{5/2+\delta}L^3\,. \]
Moreover, 
\[ \begin{split} R_{52} \leq\; &C \rho^{5/2+\delta} +  \sqrt{\rho} \, \| \nu \|_\infty 
\Big[ \int_{|w-y| \leq R}  dx dy dz dw \, |\nabla f_\ell (x-w)| \o_\ell (y-z) | \| a_w b_y \Psi \|^2 \Big]^{1/2} \\ &\hspace{1.5cm} \times 
\Big[ \int dx dy dz dw \, |\nabla \nu (x-y)|^2 |\nabla f_\ell (x-w)| \o_\ell (y-z) | \| a_w \Psi \|^2  \Big]^{1/2}  \\ 
\leq \; &C \rho^{5/2+\delta} L^3 + C \sqrt{\rho} \| \nu \|_\infty  \| \nabla \nu \|_2 \| \nabla f_\ell \|_1 \| \o_\ell \|_1 \langle \Psi, \cN \Psi \rangle^{\frac{1}{2}} \\ &\hspace{3cm} \times  \Big[ \rho R^3 \int dw \| b_w \Psi \|^2 + \int_{|w-y| \leq R} dwdy \, \| b_w b_y \Psi \|^2 \Big]^{\frac{1}{2}} \\ \leq \; &C \rho^{5/2+\delta} L^3 
 \end{split} \]
and similarly for $R_{53}$. As for $R_{54}$, we use Lemma \ref{lm:a-bds} to estimate
\[ \begin{split} R_{54} \leq  \; &C \rho^{5/2+\delta}L^3 + \| \nabla \nu \|_\infty \| \nu \|_\infty \| \nabla f_\ell \|_1 \| \o_\ell \|_1 \int_{|w-z| \leq R} dw dz \, \| b_z a_w \Psi \|^2 \\ \leq \; &C \rho^{5/2+\delta}L^3  + \| \nabla \nu \|_\infty \| \nu \|_\infty \| \nabla f_\ell \|_1 \| \o_\ell \|_1 \\ &\hspace{4cm} \times \Big[ \rho R^3 \int dz \| b_z \Psi \|^2 + \int_{|w-z| \leq R} dw dz \, \| b_z b_w \Psi \|^2 \Big] \\ \leq \; &C\rho^{5/2+\delta} L^3 \,. \end{split} \]
Thus, 
\[ \begin{split} \Big| R_5 - \frac{1}{2} \int dx dy dz dw  \, \nabla \nu (x-y) \nu (x-z) \nabla f^2_\ell (x-w)  \langle a_w a_z \Psi , J (x)^2 &J(y) J(z) a_w a_y \Psi \rangle \Big| \\ &\leq C \rho^{5/2+\delta}L^3\,. \end{split} \]
Next, we write $1- J(x)^2 J(y) J(z) = (1 - J(x)^2) J(y) J(z) + (J(y)-1) J(z) + J(z) -1$. We focus on the contribution proportional to $1- J(z) \leq C d\Gamma (\o_{\ell,z})$. Observing that $\int \nabla f_\ell^2 (x-w) dw = 0$, we find 
\[ \begin{split}  \int dx dy &dz dw  \, \nabla \nu (x-y) \nu (x-z) \nabla f^2_\ell (x-w)  \langle a_w a_z \Psi ,  (J(z) -1) a_w a_y \Psi \rangle \\ = \; &\sqrt{\rho}  \int dx dy dz dw  \, \nabla \nu (x-y) \nu (x-z) \nabla f^2_\ell (x-w)  \langle  a_z \Psi ,  (J(z) -1) b_w a_y \Psi \rangle \\ &+  \sqrt{\rho}  \int dx dy dz dw  \, \nabla \nu (x-y) \nu (x-z) \nabla f^2_\ell (x-w)  \langle  a_z b_w \Psi ,  (J(z) -1) a_y \Psi \rangle \\ &+  \int dx dy dz dw  \, \nabla \nu (x-y) \nu (x-z) \nabla f^2_\ell (x-w)  \langle  a_z b_w  \Psi ,  (J(z) -1) b_w a_y \Psi \rangle\,. \end{split} \] 
We estimate the last term; the other two can be bounded similarly. With Cauchy-Schwarz, we have  
\[ \begin{split}  &\Big| \int dx dy dz dw  \, \nabla \nu (x-y) \nu (x-z) \nabla f^2_\ell (x-w)  \langle  a_z b_w  \Psi ,  (J(z) -1) b_w a_y \Psi \rangle \Big| \\ \leq \; &C \rho^{5/2+\delta}L^3 + \| \nu \|_\infty \Big[  \int_{\substack{|z-t| \leq C \ell ,\\ |z-w| \leq R}} dx dy dz dw dt  \,  |\nabla \nu (x-y)|^2 |\nabla f_\ell (x-w)|   \|  a_z b_w a_t \Psi \|^2 \Big]^{1/2} \\ &\hspace{.5cm} \times \Big[  \int_{\substack{|y-w| \leq R , \\ |y-t| \leq R}} dx dy dz dw dt  \, |\nabla f_\ell (x-w)| \o_\ell (z-t)  \| a_y b_w a_t \Psi \|^2 \Big]^{1/2} \\
\leq \; &C \rho^{5/2+\delta}L^3 + C \| \nu \|_\infty \| \nabla \nu \|_2 \| \nabla f_\ell \|_1 \| \o_\ell \|_1^{1/2} \\ &\hspace{3cm} \times \Big[ \rho R^3 \int_{|z-t| \leq C \ell} dz dt \| a_z a_t \Psi \|^2 + \int_{\substack{|z-t| \leq C \ell ,\\ |z-w| \leq R}}  dz dt dw \| a_z a_t a_w \Psi \|^2 \Big]^{1/2} \\ 
& \times \Big[ \rho^2 R^6 \int dw \| b_w \Psi \|^2 + \rho R^3 \int_{|y-w| \leq R} \hskip -.5cm dw dy \| b_w b_y \Psi \|^2 + \int_{\substack{|y-w| \leq R , \\ |y-t| \leq R}} dw dy dt \| b_w b_y b_t \Psi \|^2 \Big]^{1/2}\\
\leq \; &C \rho^{5/2+\delta} L^3   \,.
 \end{split} \] 
 The contributions proportional to $1-J(y)$ and $1-J(x)^2$ can be treated similarly. Recalling that $\int \nu (x) dx = \int \nabla \nu (x) dx = 0$, we obtain 
 \[ \Big| R_5 - \frac{1}{2} \int dx dy dz dw  \, \nabla \nu (x-y) \nu (x-z) \nabla f^2_\ell (x-w)  \langle a_w b_z \Psi , a_w b_y \Psi \rangle \Big| \leq C \rho^{5/2+\delta} L^3 \,. \]
Estimating 
\[ \begin{split} \Big| \int dx &dy dz dw  \, \nabla \nu (x-y) \nu (x-z) \nabla f^2_\ell (x-w)  \langle a_w b_z \Psi , a_w b_y \Psi \rangle \Big| \\ \leq \; &C \rho^{5/2+\delta}L^3  + \Big[\int_{|z-w| \leq R} dx dy dz dw \, |\nabla \nu (x-y)|^2 |\nabla f_\ell (x-w)|  \| b_z a_w \Psi \|^2 \Big]^{1/2} \\ &\hspace{1cm} \times \Big[ \int_{|y-w| \leq R} dx dy dz dw \, |\nu (x-z)|^2 |\nabla f_\ell (x-w)| \| b_y a_w \Psi \|^2 \Big]^{1/2} \\ \leq \; &C \rho^{5/2+\delta}L^3  + \| \nabla \nu \|_2 \| \nu \|_2 \| \nabla f_\ell \|_1 \Big[ \rho R^3 \int dz \, \| b_z \Psi \|^2 + \int_{|z-w| \leq R} dz dw \, \| b_z b_w \Psi \|^2 \Big] \\ \leq \; &C \rho^{5/2+\delta} L^3 \end{split} \]
we conclude that $|R_5| \leq C \rho^{5/2+\delta} L^3$. 

Summarizing, we proved that 
\[ \Big| \int dx \langle \Phi_2 (x) , \Phi_3 (x) \rangle - L^3 \int dx \nabla \nu (x) \big[ \nu (x) + (\nu * \sigma * \sigma) (x) \big] \nabla f_\ell (x) f_\ell (x) \Big| \leq C \rho^{5/2+\delta} L^3 \,.\]
Recalling that $\nu = \gamma^{-1} * \sigma$ and that $\sigma * \sigma = \gamma * \gamma -1$, we find $\nu + \nu * \sigma * \sigma = \gamma * \sigma$. Hence, 
\begin{equation}\label{eq:phi23fin} 2 \int dx \langle \Phi_2 (x) , \Phi_3 (x) \rangle \leq 2 L^3  \int dx \, \nabla \nu (x) (\gamma * \sigma) (x) \nabla f_\ell (x) f_\ell (x)  + C \rho^{5/2+\delta} L^3\,. \end{equation} 
Next, we observe that  
\[ \begin{split} 
\Big| \int dx \, ((\gamma^{-1} - 1) * \nabla \sigma) &(x) (\gamma * \sigma) (x) \nabla f_\ell (x) f_\ell (x) \Big| \\ &\leq \| \nabla f_\ell \|_1 \| (\gamma^{-1} - 1) * \nabla \sigma \|_\infty \| \gamma * \sigma \|_\infty \leq C \rho^{5/2+\delta} \end{split} \]
if $\eps , \delta > 0$ are small enough. Here, we used Lemma \ref{lm:eta} to estimate $\| (\gamma^{-1} - 1) * \nabla \sigma \|_\infty \leq \| \gamma^{-1} - 1\|_2 \| \nabla \sigma \|_2 \leq C \rho^{7/4}$ and $\| \gamma * \sigma \|_\infty \leq \| \gamma - 1 \|_2 \| \sigma \|_2 + \| \sigma \|_\infty \leq C \rho$. Moreover, again by Lemma \ref{lm:eta}, we have 
\[ \begin{split} \Big| \int dx \,  \nabla \sigma (x) ((\gamma -1) * \sigma) (x) &\nabla f_\ell (x) f_\ell (x) \Big| \\ &\leq \| \nabla f_\ell \|_1 \| \nabla \sigma \|_\infty \| \gamma-\mathbbm{1} \|_2 \| \sigma \|_2 \leq C \rho^{5/2} \ell^{-1} \leq C \rho^{5/2+\delta} \,.\end{split}  \]
Hence, we conclude from (\ref{eq:phi23fin}) that 
\[ 2 \int dx \langle \Phi_2 (x) , \Phi_3 (x) \rangle \leq 2 L^3 \int dx \, \nabla \sigma (x) \sigma (x) \nabla f_\ell (x) f_\ell (x) + C \rho^{5/2+\delta} L^3. \]
\qed

\appendix

\section{Proof of Lemma \ref{lm:eta}} 
\label{sec:eta}

\begin{proof}
Given a sequence $\widehat{f} = \{ \widehat{f}_k \}_{k \in \Lambda^*}$, on the dual lattice $\Lambda^* = 2\pi \bZ^3 / L$, we define its discrete derivative $\delta_j \widehat{f}$ in the $e_j$ direction (for $j=1,2,3$), setting 
\[ (\delta_j \widehat{f})_k := \frac{L}{2\pi} \big( \widehat{f}_{k+ \frac{2\pi}{L} e_j} - \widehat{f}_k \big)\,. \]
Here $\{ e_1, e_2, e_3 \}$ is the canonical basis of $\bR^3$. 

From the elementary inequality $|e^{-iy} - 1| / |y| \geq 2/\pi$, for all $y \in [-\pi, \pi]$, we obtain the bound
   \begin{equation} \label{Eq:Fourier_position_multiplication}
        |x_j^m f(x)|\le \Big(\frac{\pi}{2}\Big)^m\Big| \Big(\frac{L}{2\pi}\Big)^m\big(e^{-i\frac{2\pi}{L}x_j}-1\big)^m f(x)\Big|=\Big(\frac{\pi}{2}\Big)^m  \big| (\widecheck{\delta_j \widehat{f}} ) (x)\big|
    \end{equation}
which will be repeatedly used in this proof. Here 
\[ (\widecheck{\delta_j \widehat{f}}) (x) = \frac{1}{L^3} \sum_{k \in \Lambda^*} (\delta_j \widehat{f})_k e^{ik \cdot x} \]
denotes the periodic function on $\Lambda$, with Fourier coefficients $\delta_j \widehat{f}$. 

We are going to apply (\ref{Eq:Fourier_position_multiplication}) to estimate the kernels $\tilde{\sigma}, \sigma, \gamma -\mathbbm{1}, \gamma^{-1} * \sigma$ in position space. We consider first the kernel $s$, defined by the Fourier coefficients (\ref{eq:hatsk}). To bound the discrete derivatives of $\widehat{s}_k$, it is convenient to define for $k\in \mathbb R^3\setminus \{0\}$
\[ 
\widehat{s}(k) = - \frac{\rho_0 \widehat{V}_\text{eff}(k)}{\sqrt{\Big(|k|^2+\rho_0\widehat{V}_\text{eff}(k)+\sqrt{|k|^4+2|k|^2\rho_0\widehat{V}_\text{eff}(k)}\Big)^2-\rho_0^2\widehat{V}_\text{eff}(k)^2}}\,.
%\widehat{s} (k) =  - \frac{k^2 + \rho_0 \widehat{V}_\text{eff} (k) -  \sqrt{|k|^4 + 2k^2 \rho_0 \widehat{V}_\text{eff} (k)}}{\sqrt{\rho_0^2 \widehat{V}_\text{eff} (k)^2 - \Big( k^2 + \rho_0 \widehat{V}_\text{eff} (k) -  \sqrt{|k|^4 + 2k^2 \rho_0 \widehat{V}_\text{eff} (k)}\Big)^2}}
\]
with $\widehat{V}_\text{eff}(k) = 8\pi \sin (|k| \frak{a}) /|k|$ and $\widehat{s}(0):=0$, as a function on $\bR^3$, and to control its regular derivatives. For $m \in \bN$, we find 
\begin{align}
\label{Eq:Decay_F_1}
\left| \nabla^m \widehat{s} (k)\right|  & \lesssim  \rho^{\frac{1}{4}} \, |k|^{-m-\frac{1}{2}}  \qquad \hspace{1cm} \text{for $0<|k|\leq \rho^{\frac{1}{2}}$},\\
    \label{Eq:Decay_F_2}
\left|\nabla^m \widehat{s} (k)\right| & \lesssim \rho \,  |k|^{-m-2} \qquad\hspace{1.3cm} \text{for $\rho^{1/2} \leq |k|\leq 1$}\,, \\
  \label{Eq:Decay_F_3}
\left|\nabla^m \widehat{s} (k)\right| & \lesssim \rho  \, |k|^{-3} \qquad \hspace{1.8cm}\text{for $|k|\geq 1$} \, . 
     \end{align}
which in particular imply the bound (\ref{eq:sfourier}), choosing $m=0$. To prove \eqref{Eq:Decay_F_1}, we write $\widehat{s} (k) = (k^2 / \rho_0 \widehat{V}_\text{eff} (k))^{-1/4} F_1 \big( \sqrt{k^2/ \rho_0 V_\text{eff} (k)} \big)$, where 
\[ F_1 (x) = - \frac{x^2 + 1 - x \sqrt{x^2 + 2}}{\sqrt{2 (x^2 + 1) \sqrt{x^2 +2} - 2 x^3 -4 x}}\]
is a smooth function on $\bR$. Together with the bound 
\[ \Big|  \nabla^n  \Big( \frac{k^2}{\rho_0 \widehat{V}_\text{eff} (k)} \Big)^{\alpha}  \Big| \lesssim \rho_0^{-\alpha} |k|^{2\alpha - n}  \]
valid for $\alpha \in \bR$, and with the chain rule, we arrive at \eqref{Eq:Decay_F_1}. As for \eqref{Eq:Decay_F_2} and \eqref{Eq:Decay_F_3}, we write $\widehat{s} (k) = (\rho_0 \widehat{V}_\text{eff}(k) / k^2) F_2 (\rho_0 
\widehat{V}_\text{eff} (k) / k^2)$, with 
\[ F_2 (x) = - \frac{\frac{1+x-\sqrt{1+2x}}{x^2}}{\sqrt{1 - x^2 \Big( \frac{1+x-\sqrt{1+2x}}{x^2} \Big)^2}}\,. \]
Since $x \to (1+x-\sqrt{1+2x}) / x^2$, is analytic in a neighborhood of $x = 0$, \eqref{Eq:Decay_F_2}, \eqref{Eq:Decay_F_3} follow from the bounds 
\begin{equation}\label{eq:nablanx} \Big| \nabla^n  \Big( \frac{\rho_0 \widehat{V}_\text{eff} (k)}{k^2} \Big) \Big| \lesssim \left\{ \begin{array}{ll} \rho |k|^{-2-n} \quad &\text{if $\rho^{1/2} \leq |k| \leq 1$} \\ \rho |k|^{-3} \quad &\text{if $|k| \geq 1$} \end{array} \right. \end{equation}  
and from the chain rule. 

For $|k|\geq 1$,  the estimate (\ref{Eq:Decay_F_3}) can be improved significantly by subtracting the leading order term. Writing 
\begin{align*}
    \widehat{s}(k) +\frac{\rho_0\widehat{V}_\mathrm{eff}(k)}{2|k|^2}= (\rho_0 \widehat{V}_\text{eff}(k) / k^2)^2 F_3 (\rho_0 
\widehat{V}_\text{eff} (k) / k^2)
\end{align*}
 in terms of the function $F_3(x):=\frac{F_2(x)+1/2}{x}$, which is analytic in a neighborhood of $x=0$, we find together with \eqref{eq:nablanx}
    \begin{align} 
         \label{Eq:Decay_F_5}
  \Big| \nabla^m \Big[ \widehat{s}(k) +\frac{\rho_0\widehat{V}_\mathrm{eff}(k)}{2|k|^2}\Big] \Big| & \lesssim \frac{\rho^2}{|k|^6},  \ \ \  \ \ \ \  \ \ \textrm{  for $|k|\geq 1$}.
  \end{align}
Since the slow decay in $k$ on the r.h.s. of (\ref{eq:nablanx}) (in the regime $|k| \geq 1$) is due to the lack of decay of the derivatives of $\widehat{V}_\text{eff}$, we can also improve (\ref{Eq:Decay_F_3}) dividing $\widehat{s} (k)$ by $\widehat{V}_\text{eff} (k)$. We obtain 
  \begin{align}
  \label{Eq:Decay_F_4}
    \Big| \nabla^m \Big[ \frac{\widehat{s}(k)}{\widehat{V}_\mathrm{eff}(k)} \Big] \Big| \leq \Big| \nabla^m \Big[ \frac{\widehat{s}(k)}{\widehat{V}_\mathrm{eff}(k)} + \frac{\rho_0}{2|k|^2}\Big] \Big|+\Big|\nabla^m\Big[ \frac{\rho_0}{2|k|^2} \Big] \Big| & \lesssim \frac{\rho^2}{|k|^5} + \frac{\rho}{|k|^{2+m}} 
\end{align}
for all $|k| \geq 1$. Notably 
\[ (\delta_j \widehat{s})_k  = \int_0^{1} d\tau \,  (\nabla_j \widehat{s}) \, \big( k+ 2\pi \tau e_j / L \big)  \]
and similarly we can write the higher derivatives $(\delta_j^m \widehat{s})_k$ as convex combinations of $(\nabla^m_j \widehat{s}) \, \big( k+t e_j\big)$ with $ 0\leq t \leq  2\pi  m / L $. In case $|k|\geq 4\pi m  / L $, we observe
\begin{align}
\label{Eq:coef_comp_with_t}
    | k+t e_j|^{-1}\leq \frac{2}{|k|}
\end{align}
for all $ 0\leq t \leq  2\pi  m / L $, and as a consequence the bounds \eqref{Eq:Decay_F_1}, \eqref{Eq:Decay_F_3}, \eqref{Eq:Decay_F_5}, \eqref{Eq:Decay_F_4} also translate into similar estimates for the discrete derivatives of the Fourier coefficients $\widehat{s}_k$. We obtain, for all $j =1,2,3$, $m \in \bN$,  
\begin{equation}\label{eq:bdsk1} \begin{split} 
\big| (\delta_j^m \widehat{s})_k \big| &\lesssim \rho^{\frac{1}{4}} \, |k|^{-m-\frac{1}{2}}  \hspace{1cm} \text{if $|k| \leq \rho^{1/2}$} \\ 
\big| (\delta_j^m \widehat{s})_k \big| &\lesssim \rho \, |k|^{-m-2} \hspace{1.3cm} \text{if $\rho^{1/2} \leq |k| \leq 1$} \\ 
\big| (\delta_j^m \widehat{s})_k \big| &\lesssim  \rho \, |k|^{-3} \hspace{1.8cm}  \text{if $|k| \geq 1$} \end{split} \end{equation}
and also 
\begin{equation}\label{eq:bdsk2} \begin{split}  \Big| \delta_j^m \Big[ \widehat{s}_k + \frac{\rho_0 \widehat{V}_\text{eff} (k)}{2k^2} \Big] \Big|  \lesssim \frac{\rho^2}{|k|^6} , \qquad 
 \Big| \delta_j^m \Big[ \frac{\widehat{s}_k}{\widehat{V}_\text{eff} (k)} \Big] \Big| \leq \frac{\rho^2}{|k|^5} + \frac{\rho}{|k|^{2+m} } \end{split} \end{equation} 
if $|k| \geq 1$. Regarding the finitely many exceptions $k\in \Lambda^*$ with $|k|< 4\pi m  / L $, for which \eqref{Eq:coef_comp_with_t} might fail, we have used $\widehat{s} (0)=0$ and the bound $\left|  \widehat{s} (k)\right|  \lesssim  \rho^{\frac{1}{4}} \,|k|^{-\frac{1}{2}}$, see \eqref{Eq:Decay_F_1}, in order to conclude 
\begin{align*}
   \big| (\delta_j^m \widehat{s})_k \big|\lesssim  L^m \max_{m'\in \{0,\dots,m\}}\left| \widehat{s} \, \big( k+ 2\pi m' e_j / L \big)\right|\lesssim  \rho^{\frac{1}{4}} L^{m+\frac{1}{2}}\lesssim \rho^{\frac{1}{4}} \, |k|^{-m-\frac{1}{2}}.
\end{align*}

Using the estimates (\ref{eq:bdsk1}), (\ref{eq:bdsk2}), we now show the first pointwise bound in (\ref{eq:s-point}). To this end, we first claim that 
  \begin{align}  
    \label{Eq:five_over_two_decay}
      \left|{s}(x)\right| & \lesssim\frac{\rho}{|x|}\frac{1}{(\rho^{1/2}|x|)^{3/2}} 
    \end{align}
    for all $|x|\ge \rho^{-1/2}$. To prove (\ref{Eq:five_over_two_decay}), we choose $\chi \in C_c^\infty (\bR^3)$ with $\chi (z) = 1$, for $|z| \leq 2$ and $\chi (z) = 0$ for $|z| \geq 4$ and, for $\mu > 0$, we define $\chi_\mu (z) = \chi (z/\mu)$. Then, we decompose 
\begin{equation}\label{eq:decos} \begin{split} s (x) = \; &\frac{1}{L^3} \sum_{k \in \Lambda^*} \chi_{|x|^{-1}} (k) \widehat{s}_k e^{i k \cdot x} + \frac{1}{L^3} \sum_{k \in \Lambda^*_+} \big( \chi (k) - \chi_{|x|^{-1}} (k) \big) \widehat{s}_k e^{ik\cdot x} \\ &+ \frac{1}{L^3} \sum_{k \in \Lambda^*_+} \big(1- \chi (k) \big) \widehat{s}_k e^{ik\cdot x} \\ = \; &A(x) + B(x) + C(x)\,. \end{split} \end{equation} 
The first estimate in (\ref{eq:bdsk1}), with $m=0$, implies that  
\[ | A(x)| \leq \frac{1}{L^3} \sum_{|k| \leq 2|x|^{-1}} |\widehat{s}_k|  \lesssim  \frac{\rho^{1/4}}{L^3} \sum_{|k| \leq 2|x|^{-1}} \frac{1}{\sqrt{|k|}} \lesssim \frac{\rho^{1/4}}{|x|^{5/2}}\,.  \]
 Moreover, combining (\ref{eq:bdsk1}) with the observation $(\delta \chi)_k = 0$, if $|k| < 1$ or $|k| > 5$, we find (keeping in mind the assumption $|x| \geq \rho^{-1/2}$) 
 \[ \big| |x|^3 B(x) \big| \lesssim \frac{1}{L^3} \sum_{j=1}^3 \sum_{k \in \Lambda^*} \Big| \delta_j^3 \big( (\chi - \chi_{|x|^{-1}}) \widehat{s} \big)_k  \Big| \lesssim \rho^{1/4} |x|^{1/2}  \,. \]
Hence $|B(x)| \lesssim \rho^{1/4} / |x|^{5/2}$. To bound $C(x)$, we define $\widehat{Q}_k :=\big(1-\chi_{1}(k)\big) \widehat{s}_k/\widehat{V}_\mathrm{eff}(k)$, for $k \in \Lambda^*$, and we denote by $Q$ the periodic function with Fourier coefficients $\widehat{Q}_k$. We have, recalling the normalization (\ref{eq:p-fourier}) for norms in Fourier space, 
\begin{align}
\nonumber
    \big| C (x) \big| & =\frac{1}{L^3} \Big| \sum_{k\in \L^*}e^{ik\cdot x} \widehat{Q}_k \widehat{V}_\mathrm{eff}(k)\Big| = \big| (V_\mathrm{eff}* Q)(x)\big| \lesssim \Big| \int_{\partial B_\mathfrak{a}(x)} Q(y) dy \Big| \\
    \label{eq:C_argument}
    &\lesssim |x|^{-3}\sum_{j=1}^3\sup_y \big|y_j^3 Q (y)\big| \lesssim |x|^{-3}\sum_{j=1}^3\big\| \delta_j^3 \widehat{Q} \big\|_1\lesssim \rho |x|^{-3}
\end{align}
where we first used that $|x|\leq 2|y|$ for $y\in \partial B_\mathfrak{a}(x)$ and $|x|\geq 2\mathfrak{a}$, and subsequently \eqref{eq:bdsk2} to estimate the discrete derivatives. We obtain that $|C (x)| \lesssim \rho |x|^{-3} \leq \rho^{1/4} |x|^{-5/2}$ for all $|x| \geq \rho^{-1/2}$, which concludes the proof of (\ref{Eq:five_over_two_decay}). For $2\frak{a} \leq |x| \leq \rho^{-1/2}$, we show that 
\begin{equation}\label{eq:srhox} |s(x)| \lesssim \frac{\rho}{|x|}\,. \end{equation} 
To this end, we proceed analogously to (\ref{eq:decos}), decomposing  
\[ \begin{split} s(x) =   \; &\frac{1}{L^3} \sum_{k \in \Lambda^*} \chi_{\rho^{1/2}} (k) \widehat{s}_k e^{i k \cdot x} + \frac{1}{L^3} \sum_{k \in \Lambda^*} \big( \chi_{|x|^{-1}} (k) - \chi_{\rho^{1/2}} (k)\big)  \widehat{s}_k e^{i k \cdot x} \\ &+ \frac{1}{L^3} \sum_{k \in \Lambda^*_+} \big( \chi (k) - \chi_{|x|^{-1}} (k) \big) \widehat{s}_k e^{ik\cdot x} + \frac{1}{L^3} \sum_{k \in \Lambda^*_+} \big(1- \chi (k) \big) \widehat{s}_k e^{ik\cdot x} \\ = \; &\tilde{A}_1 (x) + \tilde{A}_2 (x) +  B (x) + C(x)\,. \end{split} \] 
With (\ref{eq:bdsk1}), we estimate
\[ ||x| \tilde{A}_1 (x)| \lesssim \sum_{j=1}^3 \| \delta_j \big( \chi_{\rho^{1/2}} \widehat{s} \big) \|_1  \lesssim \rho \,.\]
 Similarly, 
\[ | \tilde{A}_2 (x) | \lesssim \frac{1}{L^3} \sum_{\rho^{1/2} \leq |k| \leq |x|^{-1}} |\widehat{s}_k| \leq \frac{\rho}{L^3} \sum_{\rho^{1/2} \leq |k| \leq |x|^{-1}} \frac{1}{k^2}  \lesssim \frac{\rho}{|x|} \]
  and 
 \[  |x^2 B(x)| \lesssim \sum_{j=1}^3 \big\| \delta_j^2 \big( (\chi - \chi_{|x|^{-1}}) \widehat{s} \big) \big\|_1 \leq \frac{1}{L^3} \sum_{|x|^{-1} \leq |k| \leq 2} \frac{\rho}{|k|^4} \lesssim \rho |x|  \, . \]
 As for $C (x)$, we can apply again (\ref{eq:C_argument}), which holds true for all $|x| \geq 2 \frak{a}$. 
 This shows (\ref{eq:srhox}) and concludes the proof of the first bound in (\ref{eq:s-point}), which, in turn, immediately implies the first two bounds in (\ref{eq:est-combi}). 
 
Next, we prove that, for $r >\frak{a}$,  
\begin{align}
        \label{Eq:q_nabla}
            \int (1-\chi_r (x))  |\nabla s(x)|^2dx & \lesssim \min\left\{r^{-4}\rho^{\frac{1}{2}},r^{-1}\rho^2\right\}\,.
\end{align} 
To show (\ref{Eq:q_nabla}), first in the case $r \leq \rho^{-1/2}$, we define the coefficients $\widehat{s}_r (k) = \widehat{s}_k  \chi_{r^{-1}} (k)$ (recall that $\chi_{r^{-1}} (k) = \chi (k r)$ and $\chi \in C^\infty_c (\bR^3)$, with $\chi (x) = 1$ for $|x| < 2$, $\chi (x) = 0$, for $|x| > 4$). We denote by $s_r$ the periodic function, with coefficients $\widehat{s}_r (k)$ and we decompose 
\begin{equation} \label{eq:long-nabla}  \begin{split} \int (1 &- \chi_r (x)) |\nabla s (x)|^2 dx \\ \lesssim \;&\int |\nabla s_r (x)|^2 dx + \frac{1}{r^2} \int x^2 |\nabla (s - s_r) (x)|^2 dx \\ \lesssim \; &\frac{1}{L^3} \sum_{k \in \Lambda^*} k^2 |\widehat{s}_r (k)|^2 + \frac{1}{L^3 r^2} \sum_{j=1}^3 \sum_{k \in \Lambda^*} \big| \delta_j \big(k \widehat{s}_k (1- \chi_{r^{-1}}  (k)) \big) \big|^2 \\ 
\lesssim \; &\frac{1}{L^3} \sum_{|k| \leq 2 r^{-1}} k^2 |\widehat{s}_k|^2 \\ &+ \frac{1}{L^3 r^2} \Big[ \sum_{j=1}^3 \sum_{|k| \geq r^{-1}} k^2 | (\delta_j \widehat{s})_k|^2 + \sum_{|k| \geq  r^{-1}} |\widehat{s}_k|^2 + r^2 \sum_{r^{-1} \leq |k| \leq 5 r^{-1}} k^2 |\widehat{s}_k|^2 \Big] \,.
 \end{split} \end{equation} 
With the bounds (\ref{eq:bdsk1}) and recalling that $r \leq \rho^{-1/2}$, we obtain 
\[ \int (1 - \chi_r (x)) |\nabla s (x)|^2 dx  \lesssim \rho^2 / r\,. \] 
If $r \geq \rho^{-1/2}$, we proceed similarly, but we estimate $1-\chi_r (x) \lesssim |x|^6 / r^6$ (rather than $1-\chi_r (x) \lesssim x^2 / r^2$, as we did in (\ref{eq:long-nabla})); this leads to 
\[ \int (1-\chi_r (x) ) |\nabla s (x)|^2 dx \lesssim \rho^{1/2} / r^4 \]
as claimed. In particular, choosing $r = \ell_0$, (\ref{Eq:q_nabla}) implies the third bound in (\ref{eq:est-combi}). 
 
 Analogously to (\ref{Eq:q_nabla}), we also find the useful bounds 
 \begin{align}
                  \label{Eq:q_Delta}
    \int (1-\chi_r (x))  \left|\Delta s (x) \right|^2dx& \lesssim   r^{-3}\rho^2,\\
    \label{Eq:q_nablaDelta}
     \int (1-\chi_r (x)) \left|\nabla \Delta s (x) \right|^2dx& \lesssim   r^{-5}\rho^2,
    \end{align}
for all $r > \frak{a}$. To show \eqref{Eq:q_Delta}, we let $\widehat{A}_r (k):=  \big(k^2\widehat{s}_k+ \rho_0 \widehat{V}_\mathrm{eff}(k) /2\big) \chi_{r^{-1}}(k)$ and we compute (recalling that $V_\text{eff} (x) = 0$  for $|x| > \frak{a}$) 
\begin{equation}
\label{eq:A18-pf}
\begin{split} 
\int (1-\chi_r (x)) & \left|\Delta s (x) \right|^2dx \\ & = \int (1-\chi_r (x))  \big| \Delta s (x) + \frac{\rho_0}{2}V_\mathrm{eff} (x)\big|^2dx\\
     & \leq \int  \big| A_r(x)\big|^2dx+r^{-4}\int x^4 \big| \Delta s (x) + \frac{\rho_0}{2}V_\mathrm{eff} (x)-A_r(x)\big|^2dx \\
     & \lesssim \| \widehat{A}_r \|_2^2+r^{-4}\sum_{j=1}^3 \big\|\delta_j^2 \big(|k|^2\widehat{s}_k+\frac{\rho_0}{2}\widehat{V}_\mathrm{eff}(k)- \widehat{A}_r (k) \big)\big\|_2^2\\
     & =  \| \widehat{A}_r \|_2^2+r^{-4}\sum_{j=1}^3 \Big\|\delta_j^2 (1-\chi_{r^{-1}}(k))|k|^2 \Big(\widehat{s}_k+ \frac{\rho_0 \widehat{V}_\mathrm{eff}(k)}{2|k|^2} \Big)\Big\|_2^2 ,
\end{split}
\end{equation}
where we have used $1- \chi_r (x) \leq x^4 / r^4$. Since $|\widehat{s}_k|\lesssim\rho /|k|^2$, see \eqref{eq:bdsk1}, and $|\widehat{V}_\mathrm{eff}(k)|\lesssim 1$, we have $\| \widehat{A}_r \|_2^2\lesssim r^{-3}\rho^2$. To bound the second term on the r.h.s. of (\ref{eq:A18-pf}), we consider first the region $|k| < 1$. Here, we can bound $\widehat{s}_k$ and $\rho_0 \widehat{V}_\mathrm{eff}(k)/2|k|^2$ separately, applying \eqref{eq:nablanx} and \eqref{eq:bdsk1}, to obtain
\begin{align*}
    \sum_{|k|< 1}\Big|\delta_j^2 (1-\chi_{r^{-1}}(k))|k|^2 \Big(\widehat{s}_k+ \frac{\rho_0 \widehat{V}_\mathrm{eff}(k)}{2|k|^2} \Big)\Big|^2\lesssim r\rho^2 , 
\end{align*}
where we used the observation that, if one or two derivatives hit the cutoff, the sum can be  restricted to $|k| < 4r^{-1}$. Regarding the region $|k|\geq 1$, we have to use cancellations between $\widehat{s}_k$ and $\rho_0 \widehat{V}_\mathrm{eff}(k)/2|k|^2$. Noticing that derivatives of the cutoff vanish for $|k| > 4r^{-1}$ (and thus only contribute here if $r < 4$), we obtain the  elementary bound
\begin{align*}
    \Big|\delta_j^2 (1-\chi_{r^{-1}}(k))|k|^2 \Big(\widehat{s}_k+ \frac{\rho_0 \widehat{V}_\mathrm{eff}(k)}{2|k|^2} \Big)\Big|\lesssim |k|^2 \sum_{m=0}^2 \Big|\delta_j^m  \Big(\widehat{s}_k+ \frac{\rho_0 \widehat{V}_\mathrm{eff}(k)}{2|k|^2} \Big)\Big|,
\end{align*}
which, together with the estimates in \eqref{Eq:Decay_F_5}, yields
\begin{align*}
      \sum_{|k|\geq 1}\Big|\delta_j^2 (1-\chi_{r^{-1}}(k))|k|^2 \Big(\widehat{s}_k+ \frac{\rho_0 \widehat{V}_\mathrm{eff}(k)}{2|k|^2} \Big)\Big|^2\lesssim \rho^4 \sum_{|k|\geq 1}|k|^{-8} \lesssim \rho^4 .
\end{align*}
This concludes the proof of \eqref{Eq:q_Delta}. As for \eqref{Eq:q_nablaDelta}, we use $1- \chi_r (x) \leq x^6 / r^6$ instead to bound the contribution from high momenta $|k|\gtrsim r^{-1}$
\begin{align*}
     \int (1-\chi_r (x)) & \left|\nabla \Delta s (x) \right|^2dx \\ 
     & \lesssim \| k\widehat{A}_r \|_2^2+r^{-6}\sum_{j=1}^3 \Big\|\delta_j^3 (1-\chi_{r^{-1}}(k))|k|^2 k \Big(\widehat{s}_k+ \frac{\rho_0 \widehat{V}_\mathrm{eff}(k)}{2|k|^2} \Big)\Big\|_2^2.
\end{align*}
Proceeding as in the proof of \eqref{Eq:q_Delta}, we obtain $\| k\widehat{A}_r \|_2^2\lesssim r^{-5}\rho^2$ and
\begin{align*}
r^{-6}\Big\|\delta_j^3 & (1-\chi_{r^{-1}}(k))|k|^2 k \Big(\widehat{s}_k+ \frac{\rho_0 \widehat{V}_\mathrm{eff}(k)}{2|k|^2} \Big)\Big\|_2^2\\
& \lesssim r^{-5}\rho^2+ r^{-6}\sum_{|k|\geq 1}\Big|\delta_j^3 (1-\chi_{r^{-1}}(k))|k|^2 k \Big(\widehat{s}_k+ \frac{\rho_0 \widehat{V}_\mathrm{eff}(k)}{2|k|^2} \Big)\Big|^2\\
& \lesssim r^{-5}\rho^2+ r^{-6}\rho^4 \sum_{|k|\geq 1}|k|^{-4}\lesssim r^{-5}\rho^2.
\end{align*}

Let us now consider (\ref{eq:zeta-combi1}). First of all, we notice that, from \eqref{Eq:Decay_F_1}-\eqref{Eq:Decay_F_3}, 
    \begin{equation*}
        \|s\|_2^2=L^{-3} \sum_{k\in\L^*} |\widehat{s}_k|^2 \lesssim \rho^{3/2}.
    \end{equation*}
With the definition (\ref{def:tl-sigma}) of $\tilde{\sigma}$, the fact that $1-C \frak{a}/\ell \leq f_\ell (x) \leq 1$ on the support of $\tilde{\sigma}$ and the bound
\[ \| \omega_\ell \|^2_2 \lesssim \int_{|x| \leq C \ell}  |x|^{-2} dx \lesssim \ell \]
we obtain 
  \begin{equation*}
        \|\tilde\sigma\|_2^2 \lesssim \rho^2 \|\omega_\ell\|_2^2+ \|s\|_2^2\lesssim \rho^{3/2}
    \end{equation*}
which shows the first bound in (\ref{eq:zeta-combi1}), for $\zeta = \tilde{\sigma}$. To prove it for $\zeta = \sigma$, it is convenient to first show (\ref{Th:est_3}). To this end, we compute 
 \begin{align} \label{Eq:Decay_Small_x}
        \int_{|x|\leq \rho^{-\frac{1}{2}}} |\tilde{\sigma}(x)|dx&\lesssim  \rho^{-\frac{3}{4}}\|\tilde\sigma\|_2\lesssim 1
    \end{align}
    and, by \eqref{Eq:five_over_two_decay},
    \begin{align*}
        \int_{|x|\geq \rho^{-\frac{1}{2}}} |\tilde\sigma(x)|dx= \int_{\rho^{-\frac{1}{2}}<|x|\leq 2 \ell_0} |\tilde\sigma(x)|dx\lesssim \rho^{1/4}\int_{\rho^{-\frac{1}{2}}<|x|\leq 2 \ell_0} |x|^{-5/2}dx\lesssim \big(\rho^{\frac{1}{2}}\ell_0\big)^\frac{1}{2}.
    \end{align*}
    This establishes \eqref{Th:est_3}, from which
    \begin{equation}\label{eq:sigma-L2}
        \|\sigma\|_2^2\lesssim \|\tilde\sigma\|_2^2 + \ell_0^{-3} \|\tilde\sigma\|_1^2 \lesssim \rho^{3/2}
    \end{equation}
proving (\ref{eq:zeta-combi1}), for $\zeta=\sigma$. Since, in Fourier space, $1 \leq \widehat{\gamma}_k = \cosh (\widehat{\eta}_k) \leq 1+ | \sinh (\widehat{\eta}_k)| = 1 + |\widehat{\sigma}_k|$, (\ref{eq:sigma-L2}) also implies $\| \zeta \|_2^2 \lesssim \rho^{3/2}$ for $\zeta= \{ \gamma - \mathbbm{1}, \mathbbm{1} - \gamma^{-1} , \gamma^{-1} * \sigma \}$. 

We proceed similarly to show the first bound in (\ref{eq:zeta-combi2}). Again, we start with $\zeta = \tilde{\sigma}$. In this case, the first estimate in (\ref{eq:zeta-combi2}) follows from (\ref{Eq:q_nabla}), together with $\| \nabla f_\ell \|_2 = \| \nabla \omega_\ell \|_2 \lesssim 1$ and with the first two bounds in (\ref{eq:est-combi}) (to handle terms involving the derivative of the cutoffs). 
With (\ref{Th:est_3}), we conclude that the first bound in (\ref{eq:zeta-combi2}) also holds for $\zeta = \sigma$ and therefore, since $0\leq \widehat{\gamma}_k -1 \leq |\widehat{\sigma}_k|$, also for $\zeta = \{ \gamma - \mathbbm{1}, \mathbbm{1} - \gamma^{-1},  \gamma^{-1} * \sigma \}$. 

The second bound in (\ref{eq:zeta-combi1}) follows from (\ref{eq:decay}), taking $m=0$. Also the third bound in (\ref{eq:zeta-combi1}) follows from (\ref{eq:decay}), since 
 \begin{align*}
        \int |\zeta(x)|dx & =\int_{|x|\leq \rho^{\frac{1}{2}}\ell_0^2} |\zeta(x)|dx+\int_{|x|> \rho^{\frac{1}{2}}\ell_0^2} |\zeta(x)|dx\\
         & \lesssim \rho^{\frac{3}{4}}\ell_0^3 \|\zeta\|_2 + \int_{|x|> \rho^{\frac{1}{2}}\ell_0^2} \frac{\rho}{\ell}\Big( \frac{\rho^{\frac{1}{4}}\ell_0^{\frac{3}{2}}}{|x|}\Big)^m dx\lesssim \rho^{\frac{3}{2}}\ell_0^3
    \end{align*}
where we have used the $L^2$-estimate in (\ref{eq:zeta-combi1}) and we have chosen $m$ sufficiently large.

  Let us now show \eqref{eq:decay}. First, we claim that
    \begin{equation}\label{Eq:Intermediate_delta_est}
        \left\|\delta^{m}_j\widehat{\sigma}\right\|_1\lesssim \frac{\rho}{\ell}\ell_0^m
    \end{equation} 
    for $j=1,2,3$ and $m\in\mathbb{N}$. To show (\ref{Eq:Intermediate_delta_est}), we first observe, using (\ref{Th:est_3}) and simple scaling, that 
        \begin{equation*}
        \left\|\delta^m_j \widehat{\sigma}\right\|_1\lesssim \big\|\delta^m_j \widehat{\tilde{\sigma}}\big\|_1+\ell_0^m \rho^{3/2}.
    \end{equation*}
    To estimate $\|\delta^m_j \widehat{\tilde{\sigma}}\|_1$, we use Cauchy-Schwarz with the bound $\| (\ell |k| + \ell^2 k^2)^{-1} \|_2 \lesssim \ell^{-3/2}$ and then we switch to position space. Since $\widetilde{\sigma}(x)$ is supported on $|x| \leq 2\ell_0$, we find 
    \begin{equation} \label{eq:delta^m_sigma}
        \big\|\delta^m_j \widehat{\tilde{\sigma}}\big\|_1\lesssim \ell^{-\frac{3}{2}}\sum_{\alpha = 1}^2 \ell^\alpha  \big\| |k|^\alpha \delta^m_j \widehat{\tilde\sigma}\big\|_2 \lesssim \ell_0^m  \sum_{\alpha=1}^2 \ell^{-\frac{3}{2}+\alpha} \big\|\nabla^\alpha  {\tilde\sigma}\big\|_2+\ell_0^{m-1}\ell^{-1/2} \|\tilde\sigma\|_2\,.
    \end{equation}
      From the definition (\ref{def:tl-sigma}) of $\tilde{\sigma}$, we have 
    \begin{equation}  \label{eq:bound_gradsigma_part}
    \begin{split}
        |\nabla^\alpha \tilde{\sigma}(x)|\lesssim\;&\chi_{\ell_0}(2x)(1-\chi_\ell(2x))\Big( |\nabla^\alpha s(x)|+ \rho  |\nabla^\alpha \omega_\ell (x)|\Big)\\
        &+\bigg|\nabla^\alpha \frac{\chi_{\ell_0}(2x)(1-\chi_\ell(2x))}{f_\ell(x)}\bigg|\big( \rho  |\omega_\ell(x)| + |s(x)|\big)\,.
    \end{split}
    \end{equation}
We can bound the norm of the contributions in the first line observing that, for $\alpha = 1,2$,  
    \begin{equation*}
        \big\| (1-\chi_\ell(2\cdot)) \nabla^\alpha \omega_\ell \big\|_2\lesssim\ell^{1/2-\alpha}
    \end{equation*} 
 %   and, using the notation introduced in \eqref{eq:splitting_notation}
 %   \begin{equation} \label{eq:nablaas}
  %  \begin{split}
  %      \big\|\chi_{\ell_0}(2\cdot)(1-\chi_\ell(2\cdot)) \nabla^a s\big\|^2_2 \le\;& \big\|\nabla^a s_{\downarrow \ell}\big\|_2^2+\ell^{-4} \sum_{j=1}^3\big\|\delta_j^2\nabla^a s_{\uparrow \ell}\big\|_2^2\\
  %      \lesssim\;& L^{-3}\sum_{|k|\le 2\ell^{-1}}|k|^{2a}|\widehat{s}_k|^2+L^{-3} \ell^{-4}\sum_{j=1}^3  \sum_{|k|\ge \ell^{-1}}|k|^{2a}\big|\delta_j^2\widehat{s}_k\big|^2\\
  %      \lesssim\;& \ell^{1-2a}\rho^2,
  %  \end{split}
  %  \end{equation}
      and, by \eqref{Eq:q_nabla} and \eqref{Eq:q_Delta},
    \begin{equation} \label{eq:nablaas}
    \begin{split}
        \big\|\chi_{\ell_0}(2\cdot)(1-\chi_\ell(2\cdot)) \nabla^\alpha s\big\|^2_2  \lesssim\;& \ell^{1-2\alpha}\rho^2.
    \end{split}
    \end{equation}
As for the last term in \eqref{eq:bound_gradsigma_part}, we notice that 
    \begin{equation}\label{eq:dcutoffs}
        \bigg|\nabla^\alpha \frac{\chi_{\ell_0}(2x)(1-\chi_\ell(2x))}{f_\ell(x)}\bigg|\lesssim\ell_0^{-\alpha} \mathbbm{1}_{\ell_0/2\le|x|\le \ell_0} + \ell^{-\alpha}\mathbbm{1}_{\ell/2\le|x|\le 2\ell}\,.
    \end{equation}
    On the one hand, we have, for $\alpha = 1,2$,  
    \begin{equation*}
        \ell_0^{-\alpha}\big\|\mathbbm{1}_{\ell_0/2\le|\cdot|\le \ell_0}\big(\rho | \omega_\ell |+ | s | \big)\big\|_2\lesssim \rho^{5/4} \, .
    \end{equation*}
 On the other hand,  applying H\"older's and Sobolev's inequalities as well as \eqref{eq:nablaas} and the first bound in (\ref{eq:zeta-combi1}), we find 
    \begin{equation*}
    \begin{split}
        \ell^{-\alpha}\big\| \mathbbm{1}_{\ell/2\le|\cdot|\le 2\ell} &\big(\rho|\omega_\ell|+|s| \big)\big\|_2 \\ \lesssim\; &\ell^{1/2-\alpha} \rho + \ell^{-\alpha}\big\|\mathbbm{1}_{\ell/2\le|\cdot|\le 2\ell} \big\|_3 \big\|\mathbbm{1}_{\ell/2\le|\cdot|\le 2\ell} \,s\big\|_6\\
        \lesssim\;&\ell^{1/2-\alpha}\rho+\ell^{-\alpha+1} \bigg(\int_{\ell/2\le|x|\le2\ell}\big( |\nabla s(x)|^2 + \ell^{-2} |s (x)|^2 \big) dx\bigg)^{1/2} \lesssim \ell^{1/2-\alpha}\rho\,.
    \end{split}
    \end{equation*}
 Thus, we obtain 
     \begin{equation*}
        \big\|\nabla^\alpha \tilde\sigma\big\|_2 \lesssim\ell^{1/2-\alpha}\rho \, . 
    \end{equation*}
Inserting in \eqref{eq:delta^m_sigma} and using the $L^2$-bound in \eqref{eq:zeta-combi1}, we find  
\begin{equation}\label{eq:deltamts} \| \delta_j^m \widehat{\tilde{\sigma}} \|_1 \lesssim \frac{\rho}{\ell} \ell_0^m\end{equation} and thus \eqref{Eq:Intermediate_delta_est}, for $j=1,2,3$ and all $m\in \bN$.

In particular, (\ref{eq:deltamts}) and (\ref{Eq:Intermediate_delta_est}) imply that $\| \widehat{\zeta} \|_1 \lesssim \rho / \ell$, for every $\zeta \in \{ \tilde{\sigma}, \sigma, \gamma-\mathbbm{1} ,\mathbbm{1}-\gamma^{-1},  \gamma^{-1} * \sigma \}$, since $1\leq \widehat{\gamma}_k \leq 1 + |\widehat{\sigma}_k|$, for all $k \in \Lambda^*$. Hence
    \begin{equation*}
        |\zeta(x)|\lesssim \big\|\widehat{\zeta}\big\|_1 \lesssim\ell^{-1}\rho 
    \end{equation*}
   which shows the validity of (\ref{eq:decay}), for $m=0$. To verify \eqref{eq:decay} in the case $m>0$, we introduce the functions $h_1 (t):=\sqrt{1+t^2}-1, h_2 (t) = 1- \frac{1}{\sqrt{1+t^2}}$ and $h_3 (t) :=t / (1+t^2)^{1/2}$ so that $\widehat{\gamma}_k - 1 = h_1 (\widehat{\sigma}_k)$, $1-1/\widehat{\gamma}_k = h_2 (\widehat{\sigma}_k)$ and $\widehat{\gamma}^{-1}_k \widehat{\sigma}_k  = h_3 (\widehat{\sigma}_k)$. Since $h_1, h_2, h_3$ are smooth, with bounded derivatives, we conclude, by the chain rule, that 
   \[ \big| (\delta_j^m h_u (\widehat{\sigma}))_k \big| \lesssim \sum_{i =1}^m \sum_{\substack{\beta_1, \dots , \beta_i \geq 1 :\\ \beta_1 + \dots + \beta_i = m}}  \prod_{r=1}^i | (\delta_j^{\beta_r} \widehat{\sigma})_k | \]
for $u=1,2,3$. Therefore, we obtain 
   \[ \| \delta_j^m \widehat{\zeta} \|_1 \lesssim \sum_{i =1}^m \sum_{\substack{\beta_1, \dots , \beta_i \geq 1 :\\ \beta_1 + \dots + \beta_i = m}}  \| \delta^{\beta_1} \widehat{\sigma} \|_1 \prod_{r=2}^i \| \delta_j^{\beta_r} \widehat{\sigma} \|_\infty \]
   for $\zeta \in \{ \tilde{\sigma}, \sigma, \gamma- \mathbbm{1}, \mathbbm{1}-\gamma^{-1}, \gamma^{-1} * \sigma \}$. 
   With (\ref{Th:est_3}), we have    
   \[ \| \delta_j^{n} \widehat{\sigma} \|_\infty \lesssim \| |x_j|^n \sigma \|_1 \lesssim \ell_0^n \| \sigma \|_1 \lesssim \ell_0^n  (\rho^{1/2} \ell_0)^{1/2}\,.\]
Together with (\ref{Eq:Intermediate_delta_est}), we conclude that 
 \[ \| \delta_j^m \widehat{\zeta} \|_1 \leq  \frac{\rho}{\ell}  (\rho^{1/2} \ell^3_0)^{m/2}\,. \]
 From (\ref{Eq:Fourier_position_multiplication}), this leads to 
 \[ |x_j|^m |\zeta (x)| \lesssim \frac{\rho}{\ell}  (\rho^{1/4} \ell^{3/2}_0)^m \]
for $j=1,2,3$ and $m \in \bN$, which implies (\ref{eq:decay}). 

The proof of the second bound in \eqref{eq:zeta-combi2} is obtained replicating the proof of \eqref{eq:decay} for $m=0$, starting with the analogous of \eqref{Eq:Intermediate_delta_est}, with $k\widehat{\sigma}_k$ replacing $\widehat{\sigma}_k$ and using in addition \eqref{Eq:q_nablaDelta}. 
 
As for (\ref{eq:s2-restricted}), it follows easily from the first two bounds in (\ref{eq:est-combi}), recalling  the definitions (\ref{def:tl-sigma}), (\ref{def:sigma}) of $\tilde{\sigma}, \sigma$ (and the bound (\ref{Th:est_3})). 
 
Next, let us show the second bound in (\ref{eq:s-point}), establishing the decay of $|\nabla s (x)|$. 
%
%\eqref{eq:snabla-point}. First, we claim that 
%\begin{equation}\label{eq:point-ds} |\nabla s (x)| \lesssim \frac{\rho  |\log \rho |}{x^2} \end{equation} 
%for all $|x| > 2\frak{a}$. 
To this end, we introduce
\begin{equation}\label{eq:Dx} D (x) = - \frac{\rho_0}{8\pi} \int \frac{\chi_\mu (x-y)}{|x-y|} V_\text{eff} (y) dy = \frac{\rho_0}{4\pi \frak{a}} \int_{\partial B_\frak{a} (0)} \frac{\chi_\mu (x-y)}{|x-y|}  dy \end{equation} 
for a fixed $\mu > 2 \frak{a}$. We also  define $E (x) = s (x) - D (x)$ for all $x \in \Lambda$. 
In view of (\ref{Eq:Decay_F_5}) and keeping in mind that
\begin{equation}\label{eq:k2}  \frac{4\pi}{k^2} = \int_{\bR^3} \frac{dx}{|x|} \, e^{i k \cdot x}  \end{equation} 
we expect $D$ to be a good approximation of $s$, at high momenta (and therefore $\widehat{E}_k$ to exhibit fast decay in $k$). With 
\[ \widehat{D}_k =  - \frac{\rho_0}{8\pi} \widehat{V}_\text{eff} (k) \int dx \, \frac{\chi_\mu (x)}{|x|} e^{-ik \cdot x}  \]
we can easily bound 
\begin{equation}\label{eq:Dk-low} |(\delta_j^\alpha \widehat{D})_k| \lesssim \rho |k|^{-\alpha} \end{equation} 
for $j=1,2,3$ and $\alpha \in \bN$. At high momenta $|k| \geq 1$, we can use (\ref{eq:k2}) to write 
\begin{equation}\label{eq:Dk-high} \widehat{D}_k = - \frac{\rho_0 \widehat{V}_\text{eff} (k)}{2k^2} + \frac{\rho_0}{8\pi} \widehat{V}_\text{eff} (k) \int_{\bR^3} dx \frac{(1-\chi_\mu (x))}{|x|} e^{-ik \cdot x} \end{equation} 
and derive the improved estimate 
\begin{equation}\label{eq:Dk-high2} \Big| \delta_j^\alpha \big[ \widehat{D}_k +   \frac{\rho_0 \widehat{V}_\text{eff} (k)}{2k^2} \big] \Big| \lesssim \rho |k|^{-m} \end{equation} 
for $j=1,2,3$ and any $m \in \bN$ (in the integral on the r.h.s. of (\ref{eq:Dk-high}), we can use the smoothness of $(1-\chi_\mu (x)/|x|$ to integrate by parts and obtain arbitrarily power-law decay in $k$). Combining (\ref{eq:bdsk1}) with (\ref{eq:Dk-low}), we find 
\[ | (\delta_j^\alpha \widehat{E})_k | \leq | (\delta_j^\alpha \widehat{s})_k| + | (\delta_j^\alpha \widehat{D})_k| \lesssim \left\{ \begin{array}{ll} \rho^{1/4} |k|^{-\alpha-1/2} \quad &\text{if $|k| \leq \rho^{1/2}$} \\
\rho |k|^{-\alpha-2} \quad &\text{if $\rho^{1/2} \leq |k| \leq 1$}  \,. \end{array} \right. \] 
Combining instead (\ref{eq:bdsk2}) with (\ref{eq:Dk-high2}), we obtain 
\[ | (\delta_j^\alpha \widehat{E})_k | \leq \Big| \delta_j^\alpha \big[ \widehat{s}_k + \frac{\rho_0 \widehat{V}_\text{eff} (k)}{2k^2} \big] \Big| + \Big| \delta_j^\alpha \big[ \widehat{D}_k + \frac{\rho_0 \widehat{V}_\text{eff} (k)}{2k^2} \big] \Big| \lesssim \frac{\rho^2}{|k|^6} + \frac{\rho}{|k|^m} \]
for $|k| \geq 1$ and for any $m \in \bN$. Hence, we conclude that 
\[ \begin{split} \Big| x^2 \nabla E (x)| &\lesssim \sum_{j=1}^3 \| \delta_j^2 k \widehat{E}_k \|_1 \\ &\leq \sum_{j=1}^3 \| k \delta_j^2 \widehat{E}_k \|_1 + \| \delta_j \widehat{E} \|_1 \\ &\lesssim \frac{1}{L^3} \Big[ \sum_{|k| \leq \rho^{1/2}}  \frac{\rho^{1/4}}{|k|^{3/2}} + \sum_{\rho^{1/2} \leq |k| \leq 1} \frac{\rho}{|k|^3} + \sum_{|k| \geq 1} \frac{\rho}{|k|^5}\Big]  \lesssim \rho | \log \rho | \,.\end{split} \] 
On the other hand, for $|x| > 2\frak{a}$, (\ref{eq:Dx}) implies, by explicit computation, that $|\nabla D (x)| \lesssim \rho |x|^{-2}$. This completes the proof of the second bound in (\ref{eq:s-point}). 

Recalling the definitions (\ref{def:tl-sigma}), (\ref{def:sigma}) of $\tilde{\sigma}, \sigma$ and the bounds  (\ref{eq:srhox}), (\ref{eq:dcutoffs}), we also obtain (\ref{eq:snabla-point}). In turn, this implies (\ref{eq:snabla-p}). As for (\ref{eq:nu-nabla-p}), it follows from 
\[ \| \nabla (\gamma^{-1} * \sigma ) \|_p = \| \gamma^{-1} * \nabla \sigma \|_p \leq \| (\mathbbm{1} - \gamma^{-1}) * \nabla \sigma \|_p + \| \nabla \sigma \|_p \leq \| \nabla \sigma \|_p \big( 1+ \| \mathbbm{1} - \gamma^{-1} \|_1\big) \]
applying (\ref{eq:snabla-p}) and the $L^1$-bound in (\ref{eq:zeta-combi1}). 
\end{proof}

\end{document}